\documentclass{lmcs}


\usepackage{hyperref}

\usepackage{amsmath}
\usepackage{graphicx}
\usepackage{amssymb}
\usepackage[all]{xy}
\usepackage{bussproofs}
\usepackage{comment}

\theoremstyle{plain}
\def\eg{{\em e.g.}}
\def\cf{{\em cf.}}

\newcommand{\bfd}{{\bf d}}
\newcommand{\bfe}{{\bf e}}
\newcommand{\bfk}{{\bf k}}
\newcommand{\bfw}{{\bf w}}

\newcommand{\bfc}{{\bf c}}
\newcommand{\bfi}{{\bf i}}

\newcommand{\comb}{{\bf B}}
\newcommand{\comc}{{\bf C}}
\newcommand{\comi}{{\bf I}}
\newcommand{\comsa}{{\bf S}}
\newcommand{\comk}{{\bf K}}
\newcommand{\combl}{\vec{\bf B}}
\newcommand{\combr}{\reflectbox{$\vec{\reflectbox{$\bf B$}}$}}
\newcommand{\comcl}{\vec{\bf C}}
\newcommand{\comcr}{\reflectbox{$\vec{\reflectbox{$\bf C$}}$}}
\newcommand{\comil}{\vec{\bf I}}
\newcommand{\comir}{\reflectbox{$\vec{\reflectbox{\bf I}}$}}
\newcommand{\comdl}{\vec{\bf D}}
\newcommand{\comdr}{\reflectbox{$\vec{\reflectbox{$\bf D$}}$}}
\newcommand{\coml}{{\bf L}}
\newcommand{\comp}{{\bf P}}
\newcommand{\comt}{{\bf T}}

\newcommand{\batui}{{\bf I}^{\times}}
\newcommand{\lapp}{\scalebox{0.95}{\textcolor[rgb]{0,0.447,0.698}{$\vec{@}$}}}
\newcommand{\rapp}{\scalebox{0.95}{\reflectbox{\textcolor[rgb]{0.835,0.368,0}{$\vec{@}$}}}}
\newcommand{\llaml}[2]{\textcolor[rgb]{0,0.447,0.698}{(} {#1} \hspace{1pt} \textcolor[rgb]{0,0.447,0.698}{\mapsto} \hspace{1pt} {#2} \textcolor[rgb]{0,0.447,0.698}{)}}
\newcommand{\rlaml}[2]{\textcolor[rgb]{0.835,0.368,0}{(} {#2} \hspace{1pt} \reflectbox{\textcolor[rgb]{0.835,0.368,0}{$\mapsto$}} \hspace{1pt} {#1} \textcolor[rgb]{0.835,0.368,0}{)}}
\newcommand{\llam}[2]{\textcolor[rgb]{0,0.447,0.698}{(} {#1} \hspace{1pt} \textcolor[rgb]{0,0.447,0.698}{\overset{\ast}{\mapsto}} \hspace{1pt} {#2} \textcolor[rgb]{0,0.447,0.698}{)}}
\newcommand{\rlam}[2]{\textcolor[rgb]{0.835,0.368,0}{(} {#2} \hspace{1pt} \reflectbox{\textcolor[rgb]{0.835,0.368,0}{$\overset{\ast}{\mapsto}$}} \hspace{1pt} {#1} \textcolor[rgb]{0.835,0.368,0}{)}}
\newcommand{\llet}[3]{\mbox{{\sf let} ${#1}$ {\sf be} ${#2}$ {\sf in} ${#3}$}}

\newcommand{\dagl}[1]{{#1}^{\textcolor[rgb]{0,0.447,0.698}{\triangleright}}}
\newcommand{\dagr}[1]{{#1}^{\textcolor[rgb]{0.835,0.368,0}{\triangleleft}}}
\newcommand{\lappk}{{\vec{@}}}
\newcommand{\rappk}{\reflectbox{{$\vec{@}$}}}
\newcommand{\llamk}[2]{( {#1} \hspace{1pt} {\mapsto} \hspace{1pt} {#2} )}
\newcommand{\rlamk}[2]{( {#2} \hspace{1pt} \reflectbox{{$\mapsto$}} \hspace{1pt} {#1} )}

\newcommand{\kuro}[1]{{#1}^{\bullet}} %kuro
\newcommand{\sk}{{\bf SK}} %BCI
\newcommand{\bci}{{\bf BCI}} %BCI
\newcommand{\bikuro}{{\bf BI} \kuro{(\mathchar`-)}} %BI(-)kuro
\newcommand{\biikuro}{{\bf BI} \batui \kuro{(\mathchar`-)}} %BII(-)kuro
\newcommand{\bdi}{{\bf BDI}}
\newcommand{\biilp}{{\bf BI} \batui {\bf LP} \kuro{(\mathchar`-)}} %BII(-)kuro
\newcommand{\siro}[1]{{#1}^{\circ}}
\newcommand{\bisiro}{\comb \comi {\comi}^{\bullet} (\haih)^{\circ}}

\newcommand{\hka}{{\mathcal A}}
\newcommand{\uhka}{|{\mathcal A}|}
\newcommand{\hkb}{{\mathcal B}}
\newcommand{\uhkb}{|{\mathcal B}|}

\newcommand{\plalam}{{\mathcal L}_{P}}
\newcommand{\plalamc}{{\mathcal L}_{Pc}}
\newcommand{\plalamcd}{{\mathcal L}'_{Pc}}
\newcommand{\plalamb}{{\mathcal L}_{B}}
\newcommand{\fbiilp}{{\mathcal F}}

\newcommand{\plalamt}{{\mathcal L}_{\otimes}}
\newcommand{\comlam}{{\mathcal C}}
\newcommand{\hkt}{{\mathcal T}}
\newcommand{\uhkt}{|{\mathcal T}|}

\newcommand{\catc}{{\mathbb{C}}}
\newcommand{\ucatc}{\underline{{\mathbb{C}}}}
\newcommand{\catd}{{\mathbb{D}}}
\newcommand{\obcc}{Ob({\mathbb{C}})}
\newcommand{\cats}{{\sf Sets}}
\newcommand{\asmca}{\sf{Asm} ({\mathcal A})}
\newcommand{\asmcad}{\sf{Asm} ({\mathcal A}')}
\newcommand{\asmcb}{\sf{Asm} ({\mathcal B})}
\newcommand{\asmc}[1]{\sf{Asm} ({#1})}
\newcommand{\moda}{\sf{Mod} ({\mathcal A})}
\newcommand{\modb}{\sf{Mod} ({\mathcal B})}
\newcommand{\rlz}[1]{\| {#1} \|}
\newcommand{\erlz}{\| \mathchar`- \|}
\newcommand{\homrm}{\mathop{\mathrm{Hom}}\nolimits}

\newcommand{\migi}{\rightarrow} 
\newcommand{\xmigi}[1]{\xrightarrow{#1}}

\newcommand{\ootimes}{\widehat{\otimes}}
\newcommand{\limp}{\multimap}
\newcommand{\rimp}{\: \reflectbox{$\multimap$} \:}
\newcommand{\rimps}{\scalebox{0.75}{\reflectbox{$\multimap$}}}
\newcommand{\rimpss}{\scalebox{0.6}{\reflectbox{$\multimap$}}}

\newcommand{\eqb}{=_{\beta}}
\newcommand{\eqbe}{=_{\beta \eta}}
\newcommand{\neqb}{\neq_{\beta}}

\newcommand{\lamst}{{\lambda}^{\ast}}
\newcommand{\gamst}{{\gamma}_{\ast}}
\newcommand{\delst}{{\delta}_{\ast}}
\newcommand{\xist}{{\xi}_{\ast}}
\newcommand{\lkakko}{\langle \! [}
\newcommand{\rkakko}{] \! \rangle}
\newcommand{\cpst}[1]{[ \hspace{-1.5pt} [ {#1} ] \hspace{-1.5pt} ]}

\newcommand{\haih}{\mathchar`-}
\newcommand{\batuih}{I^{\times}}
\DeclareMathSymbol{\exm}{\mathalpha}{operators}{"21}
\DeclareMathAlphabet{\mathbbold}{U}{bbold}{m}{n} 
\newcommand{\bfban}{\mathbbold{\exm}} 
\newcommand{\banst}{{\mathbbold{\exm}}_{\ast}}
\newcommand{\asmcaban}{\asmca_{\banst}}
\newcommand{\asmcaxi}{\asmca_{{\xi}_{\ast}}}

\begin{document}

\title[Categorical Realizability for Non-symmetric Closed Structures]{Categorical Realizability for \\Non-symmetric Closed Structures}

\author[H.~Tomita]{Haruka Tomita}

\begin{abstract}
In categorical realizability, it is common to construct categories of assemblies and categories of modest sets from applicative structures.
These categories have structures corresponding to the structures of applicative structures. In the literature, classes of applicative structures inducing categorical structures such as Cartesian closed categories and symmetric monoidal closed categories have been widely studied.

In this paper, we expand these correspondences between categories with structure and applicative structures by identifying the classes of applicative structures giving rise to closed multicategories, closed categories, monoidal bi-closed categories as well as (non-symmetric) monoidal closed categories. These applicative structures are planar in that they correspond to appropriate planar lambda calculi by combinatory completeness.

These new correspondences are tight: we show that, when a category of assemblies has one of the structures listed above, the based applicative structure is in the corresponding class.

In addition, we introduce planar linear combinatory algebras by adopting linear combinatory algebras of Abramsky, Hagjverdi and Scott to our planar setting, that give rise to categorical models of the linear exponential modality and the exchange modality on the non-symmetric multiplicative intuitionistic linear logic.
\end{abstract}

\maketitle

%%%%%%%%%%%%%%%%%%%%%%%%%%%%%%%%%%%%%%%%%%%%%%%%%%%%%%%%%%%%%%%%%
%%%%%%%%%%%%%%%%%%%%%%%%%%%%%%%%%%%%%%%%%%%%%%%%%%%%%%%%%%%%%%%%%
\section{Introduction}

{\it Realizability} started with \cite{kleene} to give interpretations for Heyting arithmetic, and subsequently has been developed in many directions.
The {\it categorical realizability} we call here is one such development, giving categorical models of various programming languages and logics.
Given a very simple algebraic structure $\hka$ called {\it applicative structure} (or often called {\it combinatory algebra}), we construct  categories $\asmca$ and $\moda$ used as categorical models.
For an applicative structure $\hka$, the category of assemblies $\asmca$ is the category of ``$\hka$-computable universe" and its categorical structure depends on the computational structure of $\hka$.
Therefore, giving $\hka$ certain conditions, we obtain $\asmca$ with corresponding categorical structures.

The best known is that the condition of $\hka$ being a {\it partial combinatory algebra (PCA)} leads that $\asmca$ (and $\moda$) is a Cartesian closed category (CCC) \cite{longley}.
A PCA is an applicative structure containing two special elements $\comsa$ and $\comk$ which expresses substitution and discarding.
(We often call $\comsa$-combinator or $\comk$-combinator as such elements.)
PCAs also can be characterized by the {\it combinatory completeness}, that is, the property that any computable functions (i.e., functions expressed as untyped lambda terms) on a PCA can be represented by elements of the PCA itself.

Categorical realizability for linear structures is also well investigated.
Assuming $\hka$ is a $\bci$-{\it algebra}, that have combinators $\comb$, $\comc$ and $\comi$, $\asmca$ and $\moda$ become symmetric monoidal closed categories (SMCCs) \cite{lenisa}.
$\comb$, $\comc$ and $\comi$ are combinators expressing composition, exchanging and identity operations respectively, and $\bci$-algebras correspond to the linear lambda calculus by the combinatory completeness.
These results for PCAs and $\bci$-algebras are used as an useful method to giving various models based on CCCs and SMCCs.

On the other hand, categorical realizability based on non-symmetric structures has been less investigated.
In our previous studies \cite{tomita1,tomita2}, we proposed ``planar realizability" giving rise to non-symmetric categorical structures, such as closed multicategories, closed categories, skew closed categories and monoidal bi-closed categories.
The aim of this paper is to summarize and develop these results.

First in section \ref{secback}, we start with recalling basic notions of categorical realizability.
Results of PCAs and $\bci$-algebras are shown in the section.
Also notions of {\it applicative morphisms} and {\it linear combinatory algebras (LCAs)} are recalled from \cite{longley,ahs,hoshino}, that are used to obtain models of linear exponential modalities on linear calculus.
Basic knowledge of category theory and the lambda calculus is assumed and not referred here.

Next in section \ref{secnonsym}, we introduce several classes of applicative structures inducing non-symmetric categorical structures.
Realizing non-symmetric closed structures is a more subtle problem than the symmetric cases like CCCs and SMCCs.
Since the $\comc$-combinator in $\bci$-algebras induces the symmetry of the monoidal structure on the category of assemblies, one may think we can obtain non-symmetric categorical structures by excluding the $\comc$-combinator.
However, simply excluding the $\comc$-combinator leads no interesting categorical structures like internal hom functors, since realizing closed structures needs some exchanging of realizers even if the closed structures are not symmetric.
We have to give applicative structures with appropriately weakened exchanging that realizes internal hom structures but does not realize symmetries.
To resolve this problem, in \cite{tomita1}, we introduced a unary operation $\kuro{(\haih)}$ on an applicative structure, which allows restricted exchanging.
In section \ref{bikuro} and \ref{secbiikuro}, we recall these results, that $\bikuro$-{\it algebras} induce (non-symmetric) closed multicategories and $\biikuro$-{\it algebras} induce closed categories.
By the combinatory completeness, these classes of applicative structures correspond to the {\it planar lambda calculus}.

By the unary operation $\kuro{(\haih)}$, we obtain non-symmetric closed structures, however, this operation is not sufficient to obtain non-symmetric {\it monoidal} structures.
Assume that $\asmca$ on a $\biikuro$-algebra $\hka$ has tensor products.
When we take realizers of tensor products of $\asmca$ in the same way that we take realizers of Cartesian/tensor products of assemblies on PCAs/$\bci$-algebras, the realizer of unitors of $\asmca$ leads a realizer of the symmetry.
That is, this attempt to get non-symmetric tensor products from $\biikuro$-algebras ends in failure that the tensor products are symmetric.
Here what matters is that the way realizing products of assemblies on PCAs/$\bci$-algebras corresponds to the representation of tensor products
\[ X \otimes Y \cong \forall \alpha. (X \limp Y \limp \alpha) \limp \alpha \]
in the second-order linear logic (\cf~\cite{polymorphic}), which is valid only if the tensor is symmetric.
Thus, categorical realizability for non-symmetric monoidal structures needs some modification on the way realizing tensor products.

In this paper, we give two answers for this problem. 
One is the way preparing a new combinator $\comp$ which directly realizes pairings.
The class of applicative structures, $\biilp$-{\it algebras}, is newly introduced in this paper and give rise to non-symmetric monoidal closed categories.
We show results about $\biilp$-algebras in section \ref{secbiilp}.
The other way is taking realizers of tensor products matching the representation
\begin{center}
$X \otimes Y \cong \forall \alpha. (\alpha \rimp Y \rimp X) \limp \alpha$
\end{center}
in the second-order linear logic, which is valid even in the non-symmetric case.
To give such realizers, the class of applicative structures, {\it bi}-$\bdi$-{\it algebras}, was introduced in \cite{tomita2}.
Bi-$\bdi$-algebras feature two kinds of applications corresponding to two kinds of implications $\rimp$ and $\limp$, and have the combinatory completeness for the lambda calculus with two kinds of applications (which we call the {\it bi-planar lambda calculus} in this paper).
In section \ref{secbdi}, we recall these results about bi-$\bdi$-algebras.

Classes of applicative structures appearing in this paper are summarized in Table \ref{tab:table1}.
Also combinators and operations are summarized in Table \ref{tab:combinator}.

\begin{table}[h]
\caption{Summary of the classes of applicative structures}
\centering
\begin{tabular}{|c|c|c|c|}
	\hline
	\ Applicative structure $\hka$ \ & Definition & \ Structure of $\asmca$ and $\moda$ \ & Proposition \\ \hline \hline
	PCA/$\comsa \comk$-algebra & \ref{defpca} & Cartesian closed category & \ref{propccc} \\ \hline
	$\bci$-algebra & \ref{defbci} & symmetric monoidal closed category & \ref{propsmcc} \\ \hline
	bi-$\bdi$-algebra & \ref{defbdi} & monoidal bi-closed category & \ref{propbiclo} \\ \hline
	$\biilp$-algebra	& \ref{defbiilp} & monoidal closed category & \ref{propmonoclo} \\ \hline
	$\biikuro$-algebra & \ref{defbii} & closed category & \ref{propclo} \\ \hline
	$\bikuro$-algebra & \ref{defbikuro} & closed multicategory & \ref{propmul} \\ \hline \hline
	
	\multicolumn{4}{|c|}{Inclusions} \\ \hline
	\multicolumn{4}{|c|}{$\comsa \comk$-algebras $\subsetneq$ $\bci$-algebras $\subsetneq$ bi-$\bdi$-algebras} \\ 
	\multicolumn{4}{|c|}{$\subsetneq$ $\biilp$-algebras $\subsetneq$ $\biikuro$-algebras $\subsetneq$ $\bikuro$-algebras} \\ \hline
\end{tabular}
\label{tab:table1}
\end{table}

\begin{table}[h]
\caption{Summary of combinators and operations}
\centering
\begin{tabular}{|c|c|}
	\hline
	Combinators and operations & Axiom \\ \hline 	\hline
	combinator $\comsa$ & $\comsa xyz \simeq xz(yz)$ \\ 	\hline
	combinator $\comk$ & $\comk xy = x$ \\ \hline
	combinator $\comb$ & $\comb xyz = x(yz)$ \\ \hline
	combinator $\comc$ & $\comc xyz = xzy$ \\ \hline
	combinator $\comi$ & $\comi x = x$ \\ \hline
	combinator $\batui$ & $\batui x \comi = x$ \\ \hline
	combinator $\coml$ and $\comp$ & $\coml x (\comp yz) = xyz$ \\ \hline
	combinator $\comdr$ & $x \lapp ((\comdr \rapp y) \rapp z) = (x \lapp y) \rapp z$ \\ \hline
	combinator $\comdl$ & $(z \lapp (y \lapp \comdl)) \rapp x = z \lapp (y \rapp x)$ \\ \hline
	unary operation $\kuro{(-)}$ & $(\kuro{x}) y = y x$ \\ \hline
	unary operation $\dagr{(-)}$ & $(\dagr{x}) \rapp y = y \lapp x$ \\ \hline
	unary operation $\dagl{(-)}$ & $y \lapp (\dagl{x}) = x \rapp y$ \\ \hline
\end{tabular}
\label{tab:combinator}
\end{table}

The classes of applicative structures in this paper form a hierarchy as summarized in Table \ref{tab:table1}.
In section \ref{secsepa}, we show that these classes are different from each other.
To show the strictness of the inclusion, it is sufficient to give examples belonging to one side and not to the other side, and we give such examples in section \ref{secsepa1}.
While these proofs in section \ref{secsepa1} are mostly straightforward and not conceptually new, sometimes it is not easy to show that some applicative structure does not belong to some class of applicative structures.
As such an example, in section \ref{secsepa2}, we show that the untyped planar lambda calculus (with no constants) is not a bi-$\bdi$-algebra.
In the next section \ref{seccom}, we give the {\it computational lambda calculus} \cite{comlambda} as a rather unexpected example of a $\biikuro$-algebra and show the computational lambda calculus is not a bi-$\bdi$-algebra.

To better clarify the relationship between applicative structures and categorical structures of categories of assemblies, in section \ref{secnec}, we show certain ``inverses" of propositions shown in section \ref{secnonsym}.
That is, assuming $\asmca$ has certain categorical structure (such as being an SMCC), we show $\hka$ belongs to the corresponding class (such as $\bci$-algebras) under several conditions.
While the propositions for the cases of $\bikuro$-algebras and $\biikuro$-algebras were already presented in \cite{tomita1}, those for the cases of $\biilp$-algebras and bi-$\bdi$-algebras are newly shown in this paper.
By integrating results of section \ref{secnonsym}, \ref{secsepa} and \ref{secnec}, we can say that, for instance, the category of assemblies on the planar lambda calculus indeed has non-symmetric closed structure.

In section \ref{secplca}, we reformulate notions of LCAs for our $\biilp$-algebras.
Although linear exponential comonads are usually defined as comonads on symmetric monoidal categories, we can also define linear exponential comonads on non-symmetric monoidal categories \cite{hasegawa1}.
In \cite{tomita2}, we defined {\it exponential relational planar linear combinatory algebras (exp-rPLCAs)} as pairs of a bi-$\bdi$-algebra and an applicative endomorphism on it, that give rise to linear exponential comonads on (non-symmetric) monoidal bi-closed categories.
The definition of exp-rPLCAs in \cite{tomita2} are the reformulation of the definition of (relational) LCAs to bi-$\bdi$-algebras.
In this paper, we generalize exp-rPLCAs a bit by changing ``bi-$\bdi$-algebras" to ``$\biilp$-algebras," and then similarly call the generalized ones as exp-rPLCAs.
New exp-rPLCAs give rise to linear exponential comonads on (non-symmetric) monoidal closed categories, and correspond to adjoint pairs of applicative morphisms between $\biilp$-algebras and PCAs.

There are also modalities on (non-symmetric) linear calculus other than the linear exponential modality.
The {\it exchange modality}, investigated in \cite{paiva}, is a modality connecting a commutative logic and a non-commutative logic (the {\it Lambek calculus}).
Categorical models of the exchange modality are given as monoidal adjunctions between monoidal bi-closed categories and SMCCs, which are called {\it Lambek adjoint models}.
In \cite{tomita2}, we defined {\it exchange relational planar linear combinatory algebras (exch-rPLCAs)} that give rise to Lambek adjoint models.
In this paper, like exp-rPLCAs, we reformulate exch-rPLCAs for $\biilp$-algebras.
New exch-rPLCAs correspond to adjoint pairs between $\biilp$-algebras and $\bci$-algebras, and give rise to monoidal adjunctions between (non-symmetric) monoidal closed categories and SMCCs, that are models of the exchange modality based on the non-symmetric multiplicative intuitionistic linear logic (that is, a fragment of the Lambek calculus without bi-closedness).

Finally in section \ref{secrel} and \ref{seccon}, we discuss related work, summarize conclusion and describe future work.

%%%%%%%%%%%%%%%%%%%%%%%%%%%%%%%%%%%%%%%%%%%%%%%%%%%%%%%%%%%%%%%%%
%%%%%%%%%%%%%%%%%%%%%%%%%%%%%%%%%%%%%%%%%%%%%%%%%%%%%%%%%%%%%%%%%
\section{Background} \label{secback}

%%%%%%%%%%%%%%%%%%%%%%%%%%%%%%%%%%%%%%%%%%%%%%%%%%%%%%%%%%%%%%%%%%%
\subsection{Applicative structures and categories of assemblies}

First we recall basic notions of the categorical realizability.
Notations and definitions in this subsection are from \cite{longley}.

\begin{defi}
A {\it partial applicative structure} $\hka$ is a pair of a set $\uhka$ and a partial binary operation $(x ,y) \mapsto x \cdot y$ on $\uhka$.
When the binary operation is total, we say $\hka$ is a {\it total applicative structure}.
\end{defi}

We often omit $\cdot$ and write $x \cdot y$ as $x y$ simply.
We also omit unnecessary parentheses assuming that application joins from the left.
For instance, $x y (z w)$ denotes $(x \cdot y) \cdot (z \cdot w)$.

In the sequel, we use two  notations ``$\downarrow$'' and ``$\simeq$.''
We write $x y \downarrow$ for that $x \cdot y$ is defined.
``$\simeq$'' denotes the Kleene equality, which means that if the one side of the equation is defined then the other side is also defined and both sides are equal.

\begin{defi}
Let $\hka$ be a partial applicative structure.
\begin{enumerate}
\item An {\it assembly} on $\hka$ is a pair $X = (|X|,\erlz_X)$, where $|X|$ is a set and $\erlz_X$ is a function sending $x \in |X|$ to a non-empty subset $\rlz{x}_X$ of $\uhka$.
We call elements of $\rlz{x}_X$ {\it realizers of} $x$.
\item For assemblies $X$ and $Y$ on $\hka$, a {\it map of assemblies} $f:X \migi Y$ is a function $f:|X| \migi |Y|$ such that there exists an element $r \in \uhka$ realizing $f$.
Here we say ``$r$ {\it realizes} $f$'' or ``$r$ {\it is a realizer of} $f$'' if $r$ satisfies that
\begin{center}
$\forall x \in |X|$, $\forall a \in \rlz{x}_X$, $r a \downarrow$ and $r a \in \rlz{f(x)}_Y$.
\end{center}
\end{enumerate}
\end{defi}

If we assume two additional conditions on a partial applicative structure, we can construct two kinds of categories.

\begin{defi} \label{defasm}
Let $\hka$ be a partial applicative structure satisfying that:
\begin{enumerate}[(i)]
\item $\uhka$ has an element $\comi$ such that $\forall x \in \uhka$, $\comi x \downarrow$ and $\comi x = x$; \label{cond1}
\item for any $r_1, r_2 \in \uhka$, there exists $r \in \uhka$ such that $\forall x \in \uhka$, $r x \simeq r_1 (r_2 x)$. \label{cond2}
\end{enumerate}
Then we construct categories as follows.
\begin{enumerate}
\item The category $\asmca$, called the {\it category of assemblies} on $\hka$, consists of assemblies on $\hka$ as its objects and maps of assemblies as its maps.
Identity maps and composition maps are the same as those of $\cats$ (the category of sets and functions).
\item We call an assembly $X$ a {\it modest set} on $\hka$ if $X$ satisfies
\begin{center} $\forall x,x' \in |X|$, $x \neq x' \Rightarrow \rlz{x}_X \cap \rlz{x'}_X = \emptyset$. \end{center}
The category $\moda$, called the {\it category of modest sets} on $\hka$, is the full subcategory of $\asmca$ whose objects are modest sets on $\hka$.
\end{enumerate}
\end{defi}

We need above two conditions \ref{cond1} and \ref{cond2} to give realizers of the identities and composition maps.
Identities are realized by $\comi$.
For maps $f_1:Y \migi Z$ realized by $r_1$ and $f_2:X \migi Y$ realized by $r_2$, we obtain $r$ given by the condition \ref{cond2}, which realizes $f_1 \circ f_2$.
Since all the classes of applicative structures introduced later satisfy these conditions, the conditions are not be much problems in this paper.

Intuitively, the category $\asmca$ (and $\moda$) can be understood as the category of ``$\hka$-computable universe.''
For an assembly $X=(|X|,\erlz_X)$ on $\hka$, elements of $\rlz{x}_X$ can be seen as ``machine-level interpretations'' of $x \in |X|$.
For a map $f:X \migi Y$ of $\asmca$, the realizer $r$ of $f$ can be seen as ``machine implementation'' of $f$, since $r$ takes interpretations of $x$ (that is, elements of $\rlz{x}_X$) as input and computes interpretations of $f(x)$ (that is, elements of $\rlz{f(x)}_Y$).

%%%%%%%%%%%%%%%%%%%%%%%%%%%%%%%%%%%%%%%%%%%%%%%%%%%%%%%%%%%%%%%%%%%
\subsection{PCAs and Cartesian closed categories}

Since $\asmca$ is the category of $\hka$-computable universe, the structure of $\asmca$ depends on the computational structure of $\hka$.
When applicative structures belong to a specific class, specific categorical structures may be found on the categories of assemblies.
The best known such class is the class of PCAs, which induce Cartesian closed categories of assemblies.
Results in this subsection are from \cite{longley}.

\begin{defi} \label{defpca}
A {\it partial combinatory algebra (PCA)} is a partial applicative structure $\hka$ which contains two special elements $\comsa$ and $\comk$ such that:
\begin{itemize}
\item $\forall x,y \in \uhka$, $\comk x \downarrow$, $\comk x y \downarrow$ and $\comk x y = y$;
\item $\forall x,y,z \in \uhka$, $\comsa x \downarrow$, $\comsa x y \downarrow$ and $\comsa x y z \simeq x z (y z)$.
\end{itemize}
When a PCA $\hka$ is a total applicative structure, we say $\hka$ is an $\sk$-{\it algebra}.
\end{defi}

The most fundamental example of PCAs is the untyped lambda calculus.

\begin{exa} \label{examlam}
Suppose infinite supply of variables $x,y,z,\dots$. {\it Untyped lambda terms} are terms constructed from the following six rules:
\begin{center}
	\AxiomC{}
	\RightLabel{\scriptsize (identity)}
	\UnaryInfC{$x \vdash x$}
	\DisplayProof \; \;
	\AxiomC{$\Gamma \vdash M $}
	\AxiomC{$\Delta \vdash N $}
	\RightLabel{\scriptsize (application)}
	\BinaryInfC{$\Gamma , \Delta \vdash MN $}
	\DisplayProof \; \;
	\AxiomC{$\Gamma , x \vdash M $}
	\RightLabel{\scriptsize (abstraction)}
	\UnaryInfC{$\Gamma \vdash \lambda x.M $}
	\DisplayProof
\end{center}
\begin{center}
	\AxiomC{$\Gamma , x , y, \Delta \vdash M $}
	\RightLabel{\scriptsize (exchange)}
	\UnaryInfC{$\Gamma , y, x, \Delta \vdash M $}
	\DisplayProof \; \;
	\AxiomC{$\Gamma , x , y \vdash M $}
	\RightLabel{\scriptsize (contraction)}
	\UnaryInfC{$\Gamma , x \vdash M[x/y] $}
	\DisplayProof \; \;
	\AxiomC{$\Gamma \vdash M $}
	\RightLabel{\scriptsize (weakening)}
	\UnaryInfC{$\Gamma, x \vdash M $}
	\DisplayProof	
\end{center}
Here, in the application rule, $\Gamma$ and $\Delta$ are sequences of distinct variables and contain no common variables.
In the contraction rule, $M[x/y]$ denotes the term obtained by substituting $x$ for all free $y$ in $M$.
In the weakening rule, $x$ is a variable not contained in $\Gamma$.

Note that abstraction rules are only applied to the rightmost variables. In order to
apply the abstraction rule to a variable in a different position, we need to use the exchange rule several times and move the variable to the rightmost place.

We define $\beta$-{\it equivalence relation} $\eqb$ on lambda terms as the congruence of the relation $(\lambda x.M)N \sim M[N/x]$.
Untyped lambda terms modulo $=_{\beta}$ form a PCA (actually an $\sk$-algebra). The underlying set of the PCA consists of $\beta$-equivalence classes of untyped closed lambda terms (i.e., lambda terms with no free variables) and the application is defined as that of lambda terms. In this example, $\lambda xyz.xz(yz)$ is the representative of $\comsa$ and $\lambda xy.x$ is the representative of $\comk$. 
\end{exa}

The correspondence between PCAs and the lambda calculus is more than just an example. PCAs have an important property called the {\it combinatory completeness}, which gives interpretations of ``computable functions'' on $\hka$ by elements of $\hka$ itself.
First, we give the definition of polynomials over an applicative structure (not restricted to PCAs).

\begin{defi} \label{defpoly}
Let $\hka$ be a partial applicative structure.
A {\it polynomial} over $\hka$ is a syntactic expression generated by variables, elements of $\uhka$ and the application of $\hka$.
For two polynomials $M$ and $N$ over $\hka$, $M \simeq N$ means that \[ M[a_1/x_1,\dots ,a_n/x_n] \simeq N[a_1/x_1,\dots ,a_n/x_n] \] holds in $\hka$ for any $a_1,\dots ,a_n \in \uhka$, where $\{ x_1,\dots ,x_n \}$ contains all the variables of $M$ and $N$.
\end{defi}

\begin{prop}[combinatory completeness for PCAs] \label{propcompca}
Let $\hka$ be a PCA and $M$ be a polynomial over $\hka$.
For any variable $x$, there exists a polynomial $M'$ such that the free variables of $M'$ are the free variables of $M$ excluding $x$ and $M' a \simeq M[a/x]$ holds for all $a \in \uhka$.
We write such $M'$ as $\lamst x.M$.
\end{prop}

\begin{proof}
We define $\lamst x.M$ by induction on the structure of $M$.
\begin{itemize}
\item $\lamst x.x := \comsa \comk \comk$
\item $\lamst x.y := \comk y$ \; \; (when $x \neq y$)
\item $\lamst x.MN := \comsa (\lamst x.M)(\lamst x.N)$ \qedhere
\end{itemize}
\end{proof}

For the special case of the above proposition, any closed lambda term is $\beta$-equivalent to some term constructed from $\lambda xy.x$ and $\lambda xyz.xz(yz)$ using applications.

Using the combinatory completeness, we can give $\asmca$ (and $\moda$) on a PCA $\hka$ the structure of Cartesian closed category (CCC).

\begin{prop} \label{propccc}
When $\hka$ is a PCA, $\asmca$ and $\moda$ are CCCs.
\end{prop}

While this result is standard, we shall outline its proof for comparison with the parallel results on various classes of combinatory algebras to be developed in this paper.
%The proposition is well known one, however, we note the proof for comparison with similar propositions of this paper.

\begin{proof}
First we prove the proposition for $\asmca$.
Let $\catc:=\asmca$.
\begin{itemize}
\item By the combinatory completeness, $\uhka$ has elements $\lamst x.x$ and $\lamst xyz.x(yz)$, which make $\hka$ satisfying the conditions \ref{cond1} and \ref{cond2} of Definition \ref{defasm}.
Thus $\catc$ is a category.

\item For objects $X$ and $Y$, the underlying set of the binary product $X \times Y$ is $|X| \times |Y|$.
Realizers are defined as \[ \rlz{(x,y)}_{X \times Y} := \{ \lamst t.tp q \mid p \in \rlz{x}_X, q \in \rlz{y}_Y \}. \] 

\item For maps $f:X \migi X'$ realized by $r_f$ and $g:Y \migi Y'$ realized by $r_g$, $f \times g$ is the function sending $(x, y)$ to $(f(x),g(y))$.
A realizer for $f \times g$ does exists as $\lamst u.u (\lamst pqt.t(r_f p)(r_g q))$.

\item The underlying set of the terminal object $1$ is the singleton $\{ \ast \}$.
Realizers are $\rlz{\ast}_1 := \uhka$. 
It is easy to see that this $1$ satisfy the conditions of the terminal object.

\item The projection $\pi:X \times Y \migi X$ is the function sending $(x,y)$ to $x$ and has a realizer $\lamst u.u(\lamst pq.p)$.
The projection $\pi':X \times Y \migi Y$ is the function sending $(x,y)$ to $y$ and has a realizer $\lamst u.u(\lamst pq.q)$.
It is easy to see that these $\pi$ and $\pi'$ satisfy the conditions of the projections of the Cartesian category.

\item For objects $X$ and $Y$, the underlying set of the exponential $Y^X$ is $\homrm_{\catc}(X,Y)$.
Realizers are \[ \rlz{f}_{Y^X} := \{ r \in \uhka \mid \mbox{$r$ realizes $f$} \}. \]

\item For maps $f:X' \migi X$ realized by $r_f$ and $g:Y \migi Y'$ realized by $r_g$, the map $g^f$ is the function sending a map  $h \in \homrm_{\catc}(X,Y)$ realized by $r_h$ to $g \circ h \circ f \in \homrm_{\catc}(X',Y')$ realized by $\lamst v.r_g (r_h (r_f v))$.
A realizer of $g^f$ is $\lamst uv.r_g (u (r_f v))$.

\item The adjunction $\Phi: \catc(X \times Y,Z) \migi \catc(X,Z^Y)$ is the function sending $f:X \times Y \migi Z$ realized by $r_f$ to the map $\Phi(f):x \mapsto (y \mapsto f(x,y))$.
$\Phi(f)$ is realized by $\lamst pq.r_f (\lamst t.tpq)$.
For a map $g:X\migi Z^Y$ realized by $r_g$, $\Phi^{-1}(g):X \times Y \migi Z$ is the map sending $(x,y)$ to $g(x)(y)$.
$\Phi^{-1}(g)$ is realized by $\lamst u.u r_g$.
It is easy to see that this $\Phi$ satisfies the condition of the adjunction of the CCC.
\end{itemize}
Therefore, $\asmca$ is a CCC.

Next we show that $\moda$ is a CCC.

Given modest sets $X$ and $Y$ on $\hka$, we define the binary product $X \times Y$ in the  same way as $\asmca$.
Here we can show that $X \times Y$ also is a modest set.
Suppose there is some $a \in \uhka$ realizing different $(x,y)$ and $(x',y')$ of $|X| \times |Y|$.
When we assume $x \neq x'$, though $\pi (x,y) \neq \pi (x',y')$, both sides have the same realizer $(\lamst u.u(\lamst pq.p))a$.
It contradicts that $X$ is a modest set.
The same contradiction is lead when $y \neq y'$.
Therefore, different $(x,y)$ and $(x',y')$ do not have common realizers and $X \times Y$ is a modest set.

For modest sets $X$ and $Y$ on $\hka$, we also define $Y^X$ in the  same way as $\asmca$.
We can show that $Y^X$ also is a modest set.
Suppose there is some $r$ realizing different $f:X \migi Y$ and $g:X \migi Y$.
Take $x \in |X|$ and $a \in \rlz{x}_X$ such that $f(x) \neq g(x)$.
Then $r a$ is an element of both $\rlz{f(x)}_Y$ and $\rlz{g(x)}_Y$.
However, it contradicts that $Y$ is a modest set.
Therefore, $Y^X$ is a modest set.

Hence, we can show that $\moda$ is a CCC by the same proof for $\asmca$.
\end{proof}

In this proof, we use the combinatory completeness for the PCA a lot to give realizers for each assembly and map.

%%%%%%%%%%%%%%%%%%%%%%%%%%%%%%%%%%%%%%%%%%%%%%%%%%%%%%%%%%%%%%%%%%%
\subsection{$\bci$-algebras and symmetric monoidal closed categories}

Given an applicative structure $\hka$ which has the different computational structure from PCAs, we obtain $\asmca$ with a different categorical structure from CCCs.
In this subsection, we recall another well-known class of applicative structures called $\bci$-algebras, which correspond to linear structures.
Results given in this subsection are from \cite{lenisa,hoshino}.

\begin{defi} \label{defbci}
A $\bci$-{\it algebra} is a total applicative structure $\hka$ which contains three elements $\comb$, $\comc$ and $\comi$ such that $\forall x, y, z \in \uhka$, $\comb x y z = x (y z)$, $\comc x y z = x z y$ and $\comi x = x$.
\end{defi}

\begin{exa}
{\it Untyped linear lambda terms} are untyped lambda terms constructed without using weakening and contraction rules (See Example \ref{examlam}).
That is, an untyped linear lambda term is an untyped lambda term whose each variable appears just once in the term.
Untyped closed linear lambda terms modulo $\eqb$ form a $\bci$-algebra.
Here $\lambda xyz.x(yz)$, $\lambda xyz.xzy$ and $\lambda x.x$ are the representatives of $\comb$, $\comc$ and $\comi$ respectively.
\end{exa}

\begin{prop}[combinatory completeness for $\bci$-algebras]
Let $\hka$ be a $\bci$-algebra and $M$ be a polynomial over $\hka$.
For any variable $x$ appearing exactly once in $M$, there exists a polynomial $\lamst x.M$ such that the free variables of $\lamst x.M$ are the free variables of $M$ excluding $x$ and $(\lamst x.M) a = M[a/x]$ for all $a \in \uhka$.
\end{prop}

\begin{proof}
We define $\lamst x.M$ by induction on the structure of $M$.
\begin{itemize}
\item $\lamst x.x := \comi$
\item $\lamst x.MN := \begin{cases}
		\comc (\lamst x.M) N & (x \in FV(M)) \\
		\comb M (\lamst x.N) & (x \in FV(N))
		\end{cases}$ \qedhere
\end{itemize}
\end{proof}

The combinatory completeness for a $\bci$-algebra allows interpreting only linear lambda terms, not the whole of lambda terms.
Thus some realizers used in the proof of Proposition \ref{propccc} (such as $\lamst u.u(\lamst pq.p)$) may not exist in a $\bci$-algebra.
For $\bci$-algebras, the categories of assemblies have other categorical structure than CCCs.

\begin{prop} \label{propsmcc}
When $\hka$ is a $\bci$-algebra, $\asmca$ is a symmetric monoidal closed category (SMCC).
\end{prop}

\begin{proof}[Proof sketch] \hfill
\begin{itemize}
\item For objects $X$ and $Y$, the underlying set of $X \otimes Y$ is $|X| \times |Y|$.
Realizers are defined as $\rlz{x \otimes y}_{X \otimes Y} := \{ \lamst t.tp q \mid p \in \rlz{x}_X, q \in \rlz{y}_Y \}$.

\item For maps $f:X \migi X'$ realized by $r_f$ and $g:Y \migi Y'$ realized by $r_g$, the map $f \otimes g$ is the function sending $x \otimes y$ to $f(x) \otimes g(y)$,
which is realized by $\lamst u.u (\lamst pqt.t(r_f p)(r_g q))$.

\item The underlying set of the unit object $I$ is the singleton $\{ \ast \}$.
The realizer is $\rlz{\ast}_I := \{ \comi \}$. 

\item The right unitor $\rho_X : X \migi X \otimes I$ is the function sending $x$ to $x \otimes \ast$, which is realized by $\lamst p.(\lamst t.tp\comi)$.
The inverse $\rho^{-1}$ is realized by $\lamst u.u(\lamst pq.qp)$.

\item Also we can take the left unitor $\lambda_X : I \otimes X \migi X$ as the function $(\ast \otimes x)\mapsto x$ and the associator
$\alpha_{XYZ} : X \otimes (Y \otimes Z) \migi (X \otimes Y) \otimes Z$ as $x \otimes (y \otimes z) \mapsto (x \otimes y) \otimes z$.

\item The symmetry $\sigma_{XY} : X \otimes Y \migi Y \otimes X$ is the function sending $x \otimes y$ to $y \otimes x$, which is realized by $\lamst u.u(\lamst pqt.tqp)$.

\item For objects $X$ and $Y$, the underlying set of the exponential\footnote{For an SMCC $\catc$, the exponential is often denoted using the symbol $\limp$ satisfying \[ \catc(X\otimes Y,Z)\cong \catc(X,Y \limp Z). \] However, here we use the reversed symbol $\rimp$ satisfying $\catc(X\otimes Y,Z)\cong \catc(X,Z \rimp Y)$ to be consistent with the notation of monoidal bi-closed categories in Section \ref{secbdi}.}
$Y \rimp X$ is $\homrm_{\asmca}(X,Y)$.
Realizers are $\rlz{f}_{Y \rimps X} := \{ r \in \uhka \mid \mbox{$r$ realizes $f$} \}$.

\item For maps $f:X' \migi X$ realized by $r_f$ and $g:Y \migi Y'$ realized by $r_g$, the map $g \rimp f$ is the function sending $h:X \migi Y$ to $g \circ h \circ f :X' \migi Y'$.
A realizer of $g \rimp f$ is $\lamst uv.r_g (u (r_f v))$.

\item The adjunction $\Phi$ sends a map $f:X \otimes Y \migi Z$ to the map $\Phi(f):x \mapsto (y \mapsto f(x \otimes y))$.
\end{itemize}
It is easy to see that the above components satisfy the axioms of the SMCC.
\end{proof}

The above proof is almost the same as the proof of $\asmca$ on a PCA $\hka$ being a CCC.
However, when we prove that $\moda$ on a $\bci$-algebra $\hka$ is an SMCC, we cannot use the same proof as for PCAs.
That is because for modest sets $X$ and $Y$ on a $\bci$-algebra $\hka$, $X \otimes Y$ given by the same way as $\asmca$ is not generally a modest set. 
The following proposition is proven with a modification to resolve the problem.

\begin{prop} \label{propsmcc2}
When $\hka$ is a $\bci$-algebra, $\moda$ is an SMCC.
\end{prop}

\begin{proof}
Let $G: \moda \hookrightarrow \asmca$ be the inclusion functor and $F:\asmca \migi \moda$ be the left adjoint of $G$.
$F$ is the functor sending an assembly $X = (|X|,\erlz_X)$ to a modest set $Z = (|X|/\approx,\erlz_Z)$.
Here the relation ``$\approx$'' is the transitive closure of the relation ``$\sim$'' defined as $x \sim x'$ $:\Leftrightarrow$ $\rlz{x}_X \cap \rlz{x'}_X \neq \emptyset$.
The realizers of $z \in |Z|$ are defined as $\rlz{z}_Z := \bigcup_{x \in z} \rlz{x}_X$.
$F$ sends a map $f$ of $\asmca$ to the canonical map of $\moda$, which is realized by realizers of $f$.

We define the tensor product $\boxtimes$ in $\moda$ as $X \boxtimes Y := F(GX \otimes GY)$.
We can prove Proposition \ref{propsmcc2} by the same proof of Proposition \ref{propsmcc} by replacing $\otimes$ to $\boxtimes$.
\end{proof}

More general about constructing monoidal structures on reflexive full subcategories, see \cite{day72}.

\begin{rem} \label{rembci}
While we define $\bci$-algebras as a class of total applicative structures, we also can define ``partial $\bci$-algebras'' naturally.
For a partial $\bci$-algebra $\hka$, we can see that:
\begin{itemize}
\item $\asmca$ is not generally an SMCC;
\item adding an extra element $\bot$ (which means ``undefined''), $\hka$ naturally extends to a total $\bci$-algebra $\hka_{\bot}$;
\item $\asmca$ is the full subcategory of $\asmc{\hka_{\bot}}$.
\end{itemize}
The same discussion is given in \cite{hoshino}.
\end{rem}

%%%%%%%%%%%%%%%%%%%%%%%%%%%%%%%%%%%%%%%%%%%%%%%%%%%%%%%%%%%%%%%%%%%
\subsection{Applicative morphisms} \label{secappmor}

In this subsection, we recall the notion of applicative morphisms from \cite{longley}.

\begin{defi} \label{defmor}
Let $\hka$ be a partial applicative structure satisfying:
\begin{enumerate}[(i)]
\item $\uhka$ has an element $\comi$ such that $\forall x \in \uhka$, $\comi x \downarrow$ and $\comi x = x$; \label{condmor1}
\item for any $r_1, r_2 \in \uhka$, there exists $r \in \uhka$ such that $\forall x \in \uhka$, $r x \simeq r_1 (r_2 x)$. \label{condmor2}
\end{enumerate}
Let $\hkb$ be another partial applicative structure satisfying the same conditions.
An {\it applicative morphism} $\gamma :\hka \migi \hkb$ is a total relation from $\uhka$ to $\uhkb$ such that there exists a {\it realizer} $r_{\gamma} \in \uhkb$ of $\gamma$ satisfying that 
\begin{center}
$\forall a, a' \in \uhka$, $\forall b \in \gamma a$, $\forall b' \in \gamma a'$, $r_{\gamma} b b' \in \gamma (a a')$ whenever $a a' \downarrow$.
\end{center}
We say $\gamma$ is {\it functional} when $\gamma a$ is a singleton for each $a \in \uhka$, and simply write $\gamma a = b$ for $\gamma a = \{ b \}$.
\end{defi}

Our definition is slightly more general than the definition in \cite{longley} that makes sense only on PCAs.
We define applicative morphisms between applicative structures satisfying the conditions of Definition \ref{defasm}.
We assume these conditions to realize identity and composition morphisms.
By the condition \ref{condmor1}, the identity applicative morphism $id:\hka \migi \hka$ can be realized by $\comi$.
For applicative morphisms $\gamma:\hka \migi \hkb$ and $\delta: \hkb \migi {\mathcal C}$ realized by $r_{\gamma}$ and $r_{\delta}$, taking $p \in \delta r_{\gamma}$, the composition $\delta \circ \gamma$ can be realized by $r \in |\hkb|$ such that $\forall b \in |\hkb|$, $r b \simeq r_{\delta} (r_{\delta} p b)$.
The condition \ref{condmor2} gives such a realizer $r$.

In the sequel, for an applicative morphism $\gamma$, when we write an indexed element $r_{\gamma}$, it denotes a realizer of $\gamma$.
Also, for $a \in \uhka$ and $S,S' \subseteq \uhka$, 
when we write $a S$, it denotes the set $\{ as \mid s \in S \}$ and we consider $as\downarrow$ for all $s \in S$,
and when we write $S S'$, it denotes the set 
$\{ ss'\mid s \in S,s' \in S' \}$ and we consider $ss' \downarrow$ for all $s \in S$ and $s' \in S'$.
For instance, the condition that $\gamma$ is an applicative morphism is denoted as 
\begin{center}
$\exists r_{\gamma} \in \uhkb, \forall a,a' \in \uhka$, $aa' \downarrow \Rightarrow r_{\gamma} (\gamma a)(\gamma a') \subseteq \gamma (aa')$. 
\end{center}

From applicative morphisms, we can obtain functors between the categories of assemblies.

\begin{defi} \label{deffun}
For an applicative morphism $\gamma:\hka \migi \hkb$, $\gamst : \asmca \migi \asmcb$ is the functor sending an object $(|X|, \erlz_X)$ to $(|X|,\gamma \erlz_X)$ and sending a map to the same function.
\end{defi}

For a map $f$ in $\asmca$ realized by $r_f$, $\gamst f$ is realized by elements of $r_{\gamma} (\gamma r_f)$.
It is obvious that $\gamst$ satisfies $\gamst(id)=id$ and $\gamst(g \circ f) = \gamst(g) \circ \gamst(f)$.

Next we recall the preorder relation $\preceq$ between applicative morphisms.

\begin{defi}
For two applicative morphisms $\gamma, \delta: \hka \migi \hkb$, $\gamma \preceq \delta$ iff there is $r \in \uhkb$ such that $\forall a \in \uhka$, $r (\gamma a) \subseteq \delta a$.
\end{defi}

Using the conditions \ref{condmor1} and \ref{condmor2} of Definition \ref{defmor}, we can easily show that $\preceq$ is a preorder.

By the preorder $\preceq$, we can define adjunctions and comonads on applicative structures.
 
\begin{defi}
For two applicative morphisms $\gamma: \hka \migi \hkb$ and $\delta: \hkb \migi \hka$, $\gamma$ is a {\it right adjoint} of $\delta$ iff $\delta \circ \gamma \preceq id_{\hka}$ and $id_{\hkb} \preceq \gamma \circ \delta$.
We write $(\delta \dashv \gamma):\hka \migi \hkb$ for these settings.
\end{defi}

\begin{defi}
An applicative morphism $\gamma:\hka \migi \hka$ is called {\it comonadic} when $\uhka$ has two elements $\bfe$ and $\bfd$ such that $\forall a \in \uhka$, $\bfe (\gamma a) \subseteq \{ a \}$ and $\bfd (\gamma a) \subseteq \gamma (\gamma a)$.
\end{defi}

For adjunctions of applicative morphisms, the following properties hold.

\begin{prop} \label{propmoradj} \hfill
\begin{enumerate}
\item An adjoint pair of applicative morphisms $(\delta \dashv \gamma):\hka \migi \hkb$ gives rise to an adjoint pair $(\delst \dashv \gamst) : \asmca \migi \asmcb$.
\item For an adjoint pair of applicative morphisms $(\delta \dashv \gamma):\hka \migi \hkb$, $\delta \circ \gamma : \hka \migi \hka$ is a comonadic applicative morphism. \label{propmoradj2}
\item For a comonadic applicative morphism $\gamma:\hka \migi \hka$, $\gamst$ is a comonad on $\asmca$.
\end{enumerate}
\end{prop}

\begin{rem} \label{remrest}
In Definition \ref{deffun}, an applicative morphism $\gamma:\hka \migi \hkb$ gives rise to the functor $\gamst:\asmca \migi \asmcb$.
However, here we cannot generally obtain a functor $\gamst: \moda \migi \modb$ since $\rlz{x}_X \cap \rlz{x'}_X = \emptyset$ does not imply $\gamma(\rlz{x}_X) \cap \gamma(\rlz{x'}_X) = \emptyset$ and $\gamst X$ may not be in $\modb$.
However, for a comonadic applicative morphism $\gamma:\hka \migi \hka$, $\gamst$ can be restricted to the endofunctor on $\moda$.
Indeed, for a modest set $X$ on $\hka$, if $a \in \gamma(\rlz{x}_X) \cap \gamma(\rlz{x'}_X)$ then 
$\bfe a$ is an element of $\rlz{x}_X \cap \rlz{x'}_X$ and thus $x = x'$ concludes.
Furthermore, this $\gamst$ is a comonad on $\moda$.
\end{rem}

%%%%%%%%%%%%%%%%%%%%%%%%%%%%%%%%%%%%%%%%%%%%%%%%%%%%%%%%%%%%%%%%%%%
\subsection{Linear combinatory algebras} \label{seclca}

In the previous subsection, we saw comonadic applicative morphisms give rise to comonads, and adjoint pairs of applicative morphisms give rise to adjoint pairs between categories of assemblies.
Using this construction, we can obtain linear exponential comonads and linear-non-linear models for the linear logic.
In this subsection, we recall notions of linear combinatory algebras (LCAs) from \cite{ahs} and relational linear combinatory algebras (rLCAs) from \cite{hoshino}.

\begin{defi}
A {\it linear combinatory algebra (LCA)} consists of:
\begin{itemize}
\item a $\bci$-algebra $\hka$;
\item a functional comonadic applicative morphism $(\bfban, \bfe, \bfd)$ on $\hka$;
\item an element $\bfk \in \uhka$ such that $\forall x, y \in \uhka$, $\bfk x (\bfban y) = x$;
\item an element $\bfw \in \uhka$ such that $\forall x, y \in \uhka$, $\bfw x (\bfban y) = x (\bfban y)(\bfban y)$.
\end{itemize}
\end{defi}

As we get comonads from comonadic applicative morphism, from LCAs, we get linear exponential comonads, which are categorical models of the linear exponential modality of the linear logic.

\begin{defi}
Let $\catc$ be a symmetric monoidal category.
A {\it linear exponential comonad} consists of the following data.

\begin{itemize}
\item A symmetric monoidal comonad $(!, \delta, \epsilon, m, m_I)$. \\
Here $!$ is an endofunctor on $\catc$,
$\delta_X :!X \migi !!X$ and $\epsilon_X :!X \migi X$ are monoidal natural transformations for the comultiplication and the counit.
The natural transformation $m_{X,Y}:!X \otimes !Y \migi !(X \otimes Y)$ and the map $m_I:I \migi !I$ make $!$ be a monoidal functor.
\item Monoidal natural transformations $e_X:!X \migi I$ and $d_X:!X \migi !X \otimes !X$.
\end{itemize}
Here these components need satisfy the following conditions for each $X$.
\begin{enumerate}[(i)]
\item $(!X,d_X,e_X)$ is a commutative comonoid in $\catc$.
\item $e_X$ and $d_X$ are coalgebra morphisms.
\item $\delta_X$ is a comonoid morphism.
\end{enumerate}
\end{defi}

\begin{prop}
For an LCA $(\hka,\bfban)$, $\banst$ is a linear exponential comonad on the SMCC $\asmca$ (or $\moda$). 
\end{prop}

LCAs can be generalized from functional applicative morphisms to not functional ones, called rLCAs.

\begin{defi} \label{defrlca}
A {\it relational linear combinatory algebra (rLCA)} consists of:
\begin{itemize}
\item a $\bci$-algebra $\hka$;
\item a comonadic applicative morphism $(\bfban, \bfe, \bfd)$ on $\hka$ such that $\bfban \preceq [\bfban ,\bfban]$ and $\bfban \preceq k_i$.
Here $[\bfban,\bfban]$ and $k_i$ are applicative morphisms defined as $[\bfban ,\bfban] (x) := \{ \lamst t.ta a' \mid a,a' \in \bfban x \}$ and $k_i (x) := \{ \comi \}$.
\end{itemize}
\end{defi}

Next proposition shows the correspondence between LCAs, rLCAs and adjoint pairs between $\bci$-algebras and PCAs.

\begin{prop} \hfill
\begin{enumerate}
\item Let $\hka$ be a $\bci$-algebra and $\hkb$ be a PCA.
For an adjoint pair
$(\delta \dashv \gamma):\hka \migi \hkb$, $(\hka, \delta \circ \gamma)$ is an rLCA.

\item Let $(\hka,\bfban)$ be an LCA.
The applicative structure $\hka_{\bfban} = (\uhka, @)$ defined by $x @ y := x (\bfban y)$ is a PCA.
Furthermore, $\gamma: \hka \migi \hka_{\bfban}$ defined as the identity function and $\delta :\hka_{\bfban} \migi \hka$ sending $a \in \uhka$ to $\bfban a$ form an adjoint pair $(\delta \dashv \gamma):\hka \migi \hka_{\bfban}$. 
\end{enumerate}
\end{prop}

From rLCAs, we also get linear exponential comonads.
Moreover, we get linear-non-linear models \cite{benton} on categories of assemblies or categories of modest sets.

\begin{defi}
A {\it linear-non-linear model} is a symmetric monoidal adjunction \\ $(F \dashv G):\catc \migi \catd$ for an SMCC $\catc$ and a CCC $\catd$.
\end{defi}

\begin{prop}
For an rLCA $(\hka,\bfban)$, $\banst$ is a linear exponential comonad on the SMCC $\asmca$ (or $\moda$). 
Furthermore, the co-Kleisli adjunction between $\asmca$ and $\asmca_{\banst}$ (or $\moda$ and $\moda_{\bfban}$) is symmetric monoidal.
Thus the adjunction forms a linear-non-linear model.
\end{prop}

%%%%%%%%%%%%%%%%%%%%%%%%%%%%%%%%%%%%%%%%%%%%%%%%%%%%%%%%%%%%%%%%%
%%%%%%%%%%%%%%%%%%%%%%%%%%%%%%%%%%%%%%%%%%%%%%%%%%%%%%%%%%%%%%%%%
\section{Constructing non-symmetric categorical structures} \label{secnonsym}

In section \ref{secback}, we saw two known results that PCAs/$\bci$-algebras induce CCCs/SMCCs as the categories of assemblies and the categories of modest sets.
It is natural to try to extend these results to other classes of applicative structures, and we introduce such new classes inducing certain ``non-symmetric'' categorical structures.
In this section we recall $\bikuro$-algebras, $\biikuro$-algebras and bi-$\bdi$-algebras from \cite{tomita1, tomita2}, and introduce a new class $\biilp$-algebras.

%%%%%%%%%%%%%%%%%%%%%%%%%%%%%%%%%%%%%%%%%%%%%%%%%%%%%%%%%%%%%%%%%%%
\subsection{$\bikuro$-algebras and closed multicategories} \label{bikuro}

When we try to obtain some non-symmetric categorical structures on categories of assemblies, we will find a subtle problem.
In a $\bci$-algebra $\hka$, the $\comc$-combinator expresses exchanging the order of arguments, and is the source of the symmetric structures of $\asmca$.
So one might guess that simply omitting $\comc$ would be sufficient for getting a non-symmetric categorical structure on $\asmca$.
However, this does not work well; $\comb$ and $\comi$ alone are too weak to give an interesting structure on $\asmca$.

For instance, if we want the internal hom functor $(- \rimp -)$ on $\asmca$ on a total applicative structure $\hka$, we need certain exchanging operation in $\hka$ even if the closed structure is not symmetric.
Take an object $A$ of $\asmca$ as $|A| := \uhka$ and $\rlz{a}_A := \{ a \}$.
For maps $f,g:A \migi A$, to realize $g \rimp f$, we need a realizer $r$ which satisfies $\forall a, a' \in \uhka$, $r a a' = r_g (a (r_f a'))$.
This $r$ acts as the exchanging to move the information of $r_f$ from the left of $a$ to the right of $a$.
(In a $\bci$-algebra, such $r$ exists as $\comc(\comb \comb (\comb r_g))r_f$.)

Therefore, when we want some non-symmetric categorical structures such as non-symmetric closed structures, we need to prepare some ``more restricted exchanging'' than the $\comc$-combinator.
One way to resolve the problem is to supply not a combinator but the unary operation $\kuro{(\haih)}$ for exchanging.
In this subsection, we introduce $\bikuro$-algebras from \cite{tomita1}, which induce non-symmetric closed multicategories.

%%%%%%%%%%%%%%%%%%%%%%%%%%%%%%%%%%%%%%%%%%%%%%%%%%%%%%%%%%%%%%%%%%%

\begin{defi} \label{defbikuro}
A total applicative structure $\hka$ is a $\bikuro$-{\it algebra} iff it contains $\comb$, $\comi$ and $\kuro{a}$ for each $a \in \uhka$, where $\kuro{a}$ is an element of $\uhka$ such that $\forall x \in \uhka$, $\kuro{a} x = x a$.
\end{defi}

This $\kuro{(\haih)}$ enable restricted exchanges than the $\comc$-combinator.
Since in a $\bci$-algebra, $\comc \comi a$ satisfies the axiom of $\kuro{a}$, all $\bci$-algebras are also $\bikuro$-algebras.

The definition of $\bikuro$-algebras may seem strange compared to the definitions of PCAs or $\bci$-algebras.
However, the definition of $\bikuro$-algebras is natural in the aspect of having a good correspondence with the ``planar" lambda calculus.

\begin{exa} \label{exampla}
{\it Untyped planar lambda terms} are untyped lambda terms constructed without using weakening, contraction nor exchange rules (See Example \ref{examlam}).
That is, untyped planar lambda terms are untyped linear lambda terms such that for each subterm $\lambda x.M$, $x$ is the rightmost free variable of $M$. 
Untyped closed planar lambda terms modulo $\eqb$ form a $\bikuro$-algebra, which we call $\plalam$ in this paper.
Here $\lambda xyz.x(yz)$ and $\lambda x.x$ are the representatives of $\comb$ and $\comi$ respectively. Given a representative $M$ of $a \in |\plalam|$, $\lambda x.xM$ is also a closed planar term and is the representative of $\kuro{a}$.
\end{exa}

\begin{rem}
The definition of construction rules of planar lambda terms has two different styles. In our definition, the abstraction rule is only allowed for the rightmost variable. Such a style is seen in \cite{temperley}. On the other hand, there is also the definition that the abstraction rule is only allowed for the leftmost variable, as in \cite{noam2}. Here we employ the former style for preservation the planarity of terms under the $\beta \eta$-conversions.
\end{rem}

\begin{prop}[combinatory completeness for $\bikuro$-algebras]
Let $\hka$ be a $\bikuro$-algebra and $M$ be a polynomial over $\uhka$.
For the rightmost variable $x$ of $M$, if $x$ appears exactly once in $M$, there exists a polynomial $\lamst x.M$ such that the free variables of $\lamst x.M$ are the free variables of $M$ excluding $x$ and $(\lamst x.M) a = M[a/x]$ for all $a \in \uhka$.
\end{prop}

\begin{proof}
We define $\lamst x.M$ by induction on the structure of $M$.
\begin{itemize}
\item $\lamst x.x := \comi$
\item $\lamst x.MN := \begin{cases}
		\comb \kuro{N} (\lamst x.M) & (x \in FV(M)) \\
		\comb M (\lamst x.N) & (x \in FV(N))
		\end{cases}$
\end{itemize}
Note that for $\lamst x.MN$, $x$ is the rightmost free variable in $MN$, and thus, if $x$ is in $FV(M)$, $N$ has no free variables and $\kuro{N}$ can be defined.
\end{proof}

%%%%%%%%%%%%%%%%%%%%%%%%%%%%%%%%%%%%%%%%%%%%%%%%%%%%%%%%%%%%%%%%%%%

Then we show $\bikuro$-algebras induce certain categorical structures on the categories of assemblies.
First we recall the definition of closed multicategories from \cite{manzyuk}.

\begin{defi}
A {\it multicategory} $\catc$ consists of the following data:
\begin{enumerate}
\item a collection $Ob(\catc)$;
\item for each $n \geq 0$ and $X_1 , X_2 , \dots  , X_n , Y \in Ob(\catc)$, a set $\catc 
(X_1 ,  \dots  , X_n ; Y)$.
We often write $f \in \catc (X_1 ,  \dots  , X_n ; Y)$ as $f:X_1 ,  \dots  , X_n \migi Y$;
\item for each $X \in Ob(\catc)$, an element $id_X \in \catc (X ; X)$, called the {\it 
identity map};
\item for each $n, m_1 , m_2 , \dots  , m_n \in \mathbb{N}$ and $X^k_j , Y_k , Z$ $(1 \leq k \leq n , 1 \leq j	\leq m_k)$, a function
\begin{center}
$\circ : \catc (Y_1 , \dots  , Y_n ; Z) \times \prod_{k}^{n} \catc (X^k_{1} , \dots  , X^k_{m_k} ;Y_k) \migi \catc (X^1_1 , \dots  ,X^1_{m_1}  ,X^2_1 , \dots  , X^n_{m_n} ; Z)$ 
\end{center}
called the {\it composition}. $g \circ (f_1 , \dots  , f_n)$ denotes the composition of $g \in 
\catc (Y_1 , \dots  , Y_n ; Z)$ and $f_k \in \catc (X^k_1 , \dots  , X^k_{m_k} ; Y_k)$ $(1 \leq k \leq n)$.
The compositions satisfy associativity and identity axioms.
\end{enumerate}
\end{defi}

\begin{defi} \label{defclomul}
A {\it closed multicategory} consists of the following data:
\begin{enumerate}
\item a multicategory $\catc$;
\item for each $X_1 , X_2 , \dots  , X_n , Y \in Ob(\catc)$, an object $\underline{\catc} 
(X_1 , X_2 , \dots  , X_n ; Y)$, called the {\it internal hom object};
\item for each $X_1 , \dots  , X_n , Y \in Ob(\catc)$, a map
\begin{center}
$ev_{X_1 , \dots  , X_n ; Y} : \underline{\catc} (X_1 , \dots  , X_n ; Y), X_1 
, \dots  , X_n \rightarrow Y$,
\end{center}
called the {\it evaluation map} such that $\forall Z_1 , Z_2 , \dots  , Z_m \in Ob(\catc)$, the function
\[ {\phi}_{Z_1 , \dots  , Z_m ; X_1 , \dots  , X_n ; Y} : \catc ( Z_1 , \dots  , Z_m ; 	
\underline{\catc} (X_1 , \dots  , X_n ; Y) ) \rightarrow \catc(Z_1 , \dots  , Z_m, X_1 , \dots  , X_n ; Y) \]
sending $f$ to $ev_{X_1 , \dots  , X_n ; Y} \circ (f, id_{X_1} , \dots  , id_{X_n} )$ is invertible. 
We write the inverse function $\Lambda_{Z_1 , \dots  , Z_m ; X_1 , \dots  , X_n ; Y}$.
\end{enumerate}
\end{defi}

\begin{rem}
Here our definition of closed multicategories is different from the original definition in \cite{manzyuk} in that the order of objects of domain of maps are reversed.
This is for ease to read by matching the orders of objects and realizers.
\end{rem}

\begin{prop} \label{propmul}
When $\hka$ is a $\bikuro$-algebra, $\asmca$ and $\moda$ are closed multicategories.
\end{prop}

\begin{proof}
Let $\catc := \asmca$.
Since $\hka$ have the $\comb$-combinator and the $\comi$-combinator, $\catc$ is a category.

First we give a bi-functor $(- \rimp -):\catc^{op} \times \catc \migi \catc$ as follows:
\begin{itemize}
\item For $X,Y \in \obcc$, $Y \rimp X$ is an assembly whose underlying set is $\homrm_{\catc}(X,Y)$ and $\rlz{f}_{Y \rimps X} := \{ r \mid \mbox{$r$ realizes $f$} \}$.
\item For two maps $f:X' \migi X$ and $g:Y \migi Y'$ in $\catc$, $(g \rimp f):(Y \rimp X) \migi (Y' \rimp X')$ is the function sending $h \in \homrm_{\catc}(X,Y)$ to $g \circ h \circ f$.
\end{itemize}
%We check that $(- \rimp -)$ certainly form a functor.
Given realizers $r_f$ of $f$ and $r_g$ of $g$, $(g \rimp f)$ is realized by 
$\lamst uv.r_g (u (r_f v))$.
% = \comb(\comb (\kuro{r_f}) \comb)(\comb r_g)$.
Thus, for any maps $f$ and $g$ in $\catc$, $(g \rimp f)$ certainly is a map of $\catc$.
It is easy to see that $(- \rimp -)$ preserves identities and compositions.

Next we give $\catc$ the structure of closed multicategory.
\begin{itemize}
\item For an object $X \in \obcc$, $\catc(;X) := |X|$ and $\ucatc(;X) := X$.

\item For objects $X_1, X_2, \dots , X_n, Y \in \obcc$ ($n \geq 1$), we define the internal hom object
\[ \ucatc(X_1,\dots ,X_n ;Y) := (\dots ((Y \rimp X_n) \rimp X_{n-1})\dots ) \rimp X_1 \]
and $\catc(X_1,\dots ,X_n ;Y)$ is the underlying set of $\ucatc(X_1,\dots ,X_n ;Y)$.
We write $f(x_1)(x_2)\dots (x_n)$ as $f(x_1,\dots ,x_n)$ for $f \in \catc(X_1,\dots ,X_n ;Y)$ and $x_i \in |X_i|$.
\item Identity maps $id_X \in \catc(X;X)$ ($X \in \obcc$) are the same as identity maps of $\catc$.

\item Suppose maps $g\in \catc(Y_1,\dots ,Y_n;Z)$ and $f_k \in \catc(X^k_1,\dots ,X^k_{m_k};Y_k)$ ($1 \leq k \leq n$). 
We define $g \circ (f_1,\dots ,f_n)$ as the function that receives 
$x^1_1,\dots , x^1_{m_1} ,\dots , x^n_1 ,\dots , x^n_{m_n}$
and returns $g(f_1(x^1_1,\dots ,x^1_{m_1}) ,\dots , f_n(x^n_1,\dots ,x^n_{m_n}))$.
Here when $m_i = 0$ for some $1 \leq i \leq n$, we define
$g \circ (f_1,\dots ,f_n)$ by giving $y_i \in |Y_i|$ pointed by $f_i \in \catc(;Y_i)$ as the $i$-th argument of $g$.
Given realizers $q \in \rlz{g}_{\ucatc(Y_1,\dots ,Y_n;Z)}$ and $p_k \in \rlz{f_k}_{\ucatc(X^k_1,\dots ,X^k_{m_k};Y_k)}$, by the combinatory completeness for $\bikuro$-algebras, there is $r \in \uhka$ such that
\[ r a^1_1 \dots  a^1_{m_1} \dots  a^n_1 \dots  a^n_{m_n} = q (p_1 a^1_1 \dots  a^1_{m_1})\dots (p_n a^n_{1} \dots  a^n_{m_n}) \]
holds for any $a^1_1,\dots ,a^n_{m_n} \in \uhka$.
This $r$ realizes $g \circ (f_1,\dots ,f_n)$ and thus $g \circ (f_1,\dots ,f_n)$ is in $\catc(X^1_1,\dots ,X^n_{m_n};Z)$.

\item The evaluation map $ev_{X_1 ,\dots , X_n ; Y} : \ucatc (X_1 ,\dots , X_n ; Y), X_1 ,\dots , X_n \migi Y$ is given as the function that receives $f,x_1,\dots ,x_n$ and returns $f(x_1,\dots ,x_n)$, which is realized by $\comi$.
Then, $\phi_{Z_1,\dots ,Z_m;X_1,\dots ,X_n;Y}$ is invertible as a function and for $g \in \catc(Z_1,\dots ,Z_m,X_1,\dots ,X_n;Y)$, $\Lambda (g)$ is indeed in $\catc(Z_1,\dots ,Z_m;\ucatc(X_1,\dots ,X_n;Y))$ since it is realized by realizers of $g$.
\end{itemize}
Therefore, $\asmca$ is a closed multicateogry.
For $\moda$, we can use the same proof as for $\asmca$.
\end{proof}

While we define $\bikuro$-algebras as a class of total applicative structures, we also can define ``partial $\bikuro$-algebra'' naturally.
For a partial $\bikuro$-algebra $\hka$, $\kuro{(\haih)}$ is a total unary operation on $\uhka$ such that $\forall a, x \in \uhka, \kuro{a} x \simeq x a$.
Unlike the case of partial $\bci$-algebras as in Remark \ref{rembci}, the proof of Proposition \ref{propmul} is applicable to the case of partial $\bikuro$-algebras.

%%%%%%%%%%%%%%%%%%%%%%%%%%%%%%%%%%%%%%%%%%%%%%%%%%%%%%%%%%%%%%%%%%%
\subsection{$\biikuro$-algebras and closed categories} \label{secbiikuro}

In this subsection, we recall a class of applicative structures from \cite{tomita1}, which induce closed categories of assemblies and modest sets.
First we recall the definition of closed categories in \cite{manzyuk}.

\begin{defi}
A {\it closed category} consists of the following data:
\begin{enumerate}
\item a locally small category $\catc$;
\item a functor $(- \rimp -): {\catc}^{op} \times \catc \migi \catc$, called the {\it 
internal hom functor}\footnote{While the internal hom object in the closed category is often written as $\ucatc(X,Y)$, $[X,Y]$ or $X \limp Y$, here we denote $Y \rimp X$ to be consistent with other categorical structures in this paper.};
\item an object $I$, called the {\it unit object};
\item a natural isomorphism $i_X : (X \rimp I) \migi X$;
\item an extranatural transformation $j_X : I \migi (X \rimp X)$;
\item a transformation $L_{Y,Z}^X : (Z \rimp Y) \migi ((Z \rimp X) \rimp (Y \rimp X))$ natural in $Y$ and $Z$ and extranatural in $X$,
\end{enumerate}
such that the following axioms hold:
\begin{enumerate}[(i)]
\item $\forall X,Y \in \catc$, $L_{Y,Y}^X \circ j_Y = j_{(Y \rimps X)}$;
\item $\forall X,Y \in \catc$, $i_{(Y \rimps X)} \circ (id_{(Y \rimps X)} \rimp j_X) \circ L_{X,Y}^X = id_{(Y \rimps X)}$;
\item $\forall X,Y,Z,W \in \catc$, the following diagram commutes:
\begin{center}
\xymatrix@C=-25pt@R=30pt{
	& (W \rimp Z) \ar[dl]_{L_{Z,W}^X} \ar[dr]^{L_{Z,W}^Y} & \\
	(W \rimp X) \rimp (Z \rimp X) \ar[d]^-{L_{(Z \rimpss X),(W \rimpss X)}^{(Y \rimpss X)}} & & 
	((W \rimp Y) \rimp (Z \rimp Y)) \ar[dd]^-{L_{Y,W}^X \rimps id} \\
	((W \rimp X) \rimp (Y \rimp X)) \rimp ((Z \rimp X) \rimp (Y \rimp X)) \ar[drr]_-{id \rimps L_{Y,Z}^X} & & \\
	& & ((W \rimp X) \rimp (Y \rimp X)) \rimp (Z \rimp Y)
	}
\end{center}
\item $\forall X,Y \in \catc$, $L_{X,Y}^I \circ (i_Y \rimp id_X) = id_{(Y \rimps I)} \rimp i_X$;
\item $\forall X,Y \in \catc$, the function $\gamma : \catc (X,Y) \migi \catc (I , (Y \rimp X))$ sending $f:X \migi Y$ to 
$(f \rimp id_X) \circ j_X$ is invertible.
\end{enumerate}
\end{defi}

Closed categories are something like monoidal closed categories without tensor products. That is, categories with internal hom functors which are defined directly, not via tensor products and adjunctions.
The structures of closed categories are very similar to the structures of closed multicategories.
As shown in \cite{manzyuk}, the category of closed categories are cat-equivalent to the category of closed multicategories with unit objects.

However, when we want to construct (non-symmetric) closed categories as categories of assemblies, it is not sufficient that the applicative structures are $\bikuro$-algebras, since realizers for $i^{-1}_X : X \migi (X \rimp I)$ may not exist.
Thus, we add another condition to a $\bikuro$-algebra to realize $i^{-1}_X$ and obtain the following definition.

\begin{defi} \label{defbii}
A $\biikuro$-{\it algebra} $\hka$ is a $\bikuro$-algebra which contains an element $\batui$ such that $\forall a \in \uhka$, $\batui a \comi = a$.
\end{defi}

\begin{rem}
In $\biikuro$-algebras, the role we expect to $\batui$ is to eliminate the ``harmless" second argument, which does not necessarily eliminate $\comi$.
Even without specifying $\batui$, we can define the same class as $\biikuro$-algebras.
For instance, for a $\bikuro$-algebra $\hka$, suppose there is $\comb^{\times} \in \uhka$ such that $\forall a \in \uhka$, $\comb^{\times} a \comb = a$.
Then this $\hka$ is a $\biikuro$-algebra since $\lamst xy.\comb^{\times} x (y \comb)$ satisfies the axiom of $\batui$.
Conversely, for a $\biikuro$-algebra, we can take $\comb^{\times} := \lamst xy.\batui x (y \comi \, \comi \, \comi)$ and thus $\biikuro$-algebras and $\comb \comi \comb^{\times} \kuro{(\haih)}$-algebras are the same classes.
\end{rem}

\begin{exa} \label{exampla2}
$\plalam$ in Example \ref{exampla} is a $\biikuro$-algebra.
Since the planar lambda calculus has the strongly normalizing property, for any closed planar term $M$, there are some $u$ and $N$ such that $M \eqb \lambda u.N$. Then
\begin{eqnarray*}
(\lambda xyz.x(yz)) M (\lambda v.v) &\eqb& \lambda z.M((\lambda v.v)z) \\
&\eqb& \lambda z.(\lambda u.N)z\\
&\eqb& \lambda z.N[z/u] \\
&=_{\alpha}& M
\end{eqnarray*}
and thus $\lambda xyz.x(yz)$ represents $\batui$.
\end{exa}

Since $\plalam$ (which nicely corresponds to $\bikuro$-algebras) is also a $\biikuro$-algebra, one might suspect that $\bikuro$-algebras and $\biikuro$-algebras are the same class.
However, these two classes are different ones.
Later in Section \ref{secsepa}, we will discuss an example that separates classes of $\bikuro$-algebras and $\biikuro$-algebras (Proposition \ref{propsepa3}).

The next example based on an ordered group is from \cite{tomita1}.
(However, here we reverse the direction of the implication symbol $\limp$ of the original example in \cite{tomita1}.)

\begin{exa} \label{examtree}
Take an ordered group $(G,\cdot,e,\leq)$.
Let $T$ be a set of elements constructed grammatically as follows:
\[ t ::= g \; | \;  t \rimp t' \; \; \; (g \in G). \]
That is, $T$ is a set of binary trees whose leaves are labeled by elements of $G$.
We further define a function $| \haih | :T \migi G$ by induction: $|g| := g$ and $|t_2 \rimp t_1| := |t_2| \cdot |t_1|^{-1}$.

Let $\uhkt$ be the powerset of $\{ t \in T \mid e \leq |t| \}$.
Then we can get a $\biikuro$-algebra $\hkt$ by $\uhkt$:
\begin{itemize}
\item For $M,N \in \uhkt$, $MN := \{ t_2 \mid \exists t_1 \in N, (t_2 \rimp t_1) \in M \}$.
\item $\comb := \{ (t_3 \rimp t_1) \rimp (t_2 \rimp t_1) \rimp (t_3 \rimp t_2) \mid t_1,t_2,t_3 \in T \}$.
Here $\rimp$ joins from the left.
\item $\comi := \{ t_1 \rimp t_1 \mid t_1 \in T \}$
\item $\batui := \{ t_1 \rimp (t_2 \rimp t_2) \rimp t_1 \mid t_1,t_2 \in T \}$.
\item For $M \in \uhkt$, $\kuro{M} := \{ t_2 \rimp (t_2 \rimp t_1) \mid t_1 \in M, t_2 \in T \}$.
\end{itemize}
\end{exa}

This example is based on $Comod(\overline{G})$ introduced in \cite{hasegawa4}, which is a category of sets and relations equipped with $G$ valued functions.
For any (not necessarily ordered) group $G$, $Comod(\overline{G})$ is a pivotal category.
$\uhkt$ is a set of maps from the unit object to a reflexive object in (ordered) $Comod(\overline{G})$.
The structure of $\hkt$ depends on $G$.
For instance,
\[ \{ (t_3 \rimp t_2 \rimp t_1) \rimp (t_3 \rimp t_1 \rimp t_2) \mid t_1, t_2, t_3 \in T \} \]
acts as the $\comc$-combinator whenever $G$ is Abelian.

The above $\hkt$ later appears several times as examples of applicative structures of other classes (Example \ref{examtree3}, \ref{examtree5}, \ref{examtree6}).

\begin{prop} \label{propclo}
When $\hka$ is a $\biikuro$-algebra, $\asmca$ and $\moda$ are closed categories.
\end{prop}

\begin{proof}
Let $\catc := \asmca$.
We give the same bi-functor $(- \rimp -):\catc^{op} \times \catc \migi \catc$ as in the proof of Proposition \ref{propmul}.
\begin{itemize}
\item We define the unit object $I$ as $(\{ \ast \}, \erlz_I)$, where $\rlz{\ast}_I := \{ \comi \}$.
\item $j_X$ is the function sending $\ast$ to $id_X$, which is realized by $\comi$.
\item $i_X$ is the function sending $(f:\ast \mapsto x)$ to $x$, which is realized by $\kuro{\comi}$. The inverse $i_X^{-1}$ is realized by $\batui$.
\item $L_{Y,Z}^X$ is the function sending $g$ to the function $(f \mapsto g \circ f)$, which is realized by $\comb$.
\item $\gamma$ is invertible. Indeed, $\gamma^{-1}$ is the function sending $g:I \migi (Y \rimp X)$ to the map $g(\ast):X \migi Y$.
\end{itemize}
It is easy to verify that $j$, $i$ and $L$ have naturality and satisfy the axioms of the closed category.

For $\moda$, we can use the same proof for $\asmca$. \qedhere

\end{proof}

While we define $\biikuro$-algebras as a class of total applicative structures, we also can define ``partial $\biikuro$-algebra'' naturally.
For a partial $\biikuro$-algebra $\hka$, $\batui$ satisfies that $\forall a \in \uhka, \batui a \comi \downarrow$ and $\batui a \comi = a$.
Proposition \ref{propclo} also holds in the case of partial $\biikuro$-algebras.

%%%%%%%%%%%%%%%%%%%%%%%%%%%%%%%%%%%%%%%%%%%%%%%%%%%%%%%%%%%%%%%%%%%
\subsection{$\biilp$-algebras and monoidal closed categories} \label{secbiilp}

In the previous two subsections, we obtain closed multicategories and closed categories as categories of assemblies.
Next we further attempt to obtain a richer categorical structure, the (non-symmetric) tensor products, by categorical realizability.

First, let us consider whether we can realize products by a $\biikuro$-algebra in the same way as PCAs and $\bci$-algebras.
Even when we use a $\biikuro$-algebra, we can take the object $X \otimes Y$ in the same way as PCAs and $\bci$-algebras (See the proofs of Proposition \ref{propccc} and \ref{propsmcc}).
That is, for a $\biikuro$-algebra $\hka$, we take an assembly $X \otimes Y$ that the underlying set  is $|X| \times |Y|$ and realizers are 
\[ \rlz{x \otimes y}_{X \otimes Y} := \{ \lamst t.tp q \mid p \in \rlz{x}_X, q \in \rlz{y}_Y \}. \]
We also take the unit object in $\asmca$ in the same way as $\bci$-algebras: $|I|:= \{ \ast \}$ and $\rlz{\ast}_I := \{ \comi \}$.
Then is this $\asmca$ a monoidal category?
Now let us assume it is.
Take an assembly 
$A:= (\uhka,\erlz_A)$, where $\rlz{a}_A := \{ a \}$.
Then since $\asmca$ is a monoidal category, the unitor 
$A \migi I \otimes A$ has a realizer $r$, which satisfies that $r a = \lamst t.t \comi a$.
Taking an elements $C := \lamst xyz.r x (\comb (\comb ( \lamst w.r y (\comb w))))z$, this $C$ satisfies the axiom of the $\comc$-combinator and make $\hka$ a $\bci$-algebra.

In summary, when we attempt to make $\asmca$ a non-symmetric monoidal category using a $\biikuro$-algebra $\hka$, it follows that $\hka$ is actually a $\bci$-algebra and $\asmca$ becomes an SMCC.
Therefore, we need some major modification on the definition of realizers of tensor products in $\asmca$ to make $\asmca$ a non-symmetric monoidal category.

One way to solve this problem is supposing a combinator $\comp$ expressing the ``pairing" operation.
And we define realizers for tensor products as
\[ \rlz{x \otimes y}_{X \otimes Y} := \{ \comp pq \mid p \in \rlz{x}_X, q \in \rlz{y}_Y \}. \]
Since $\comp pq$ itself cannot separate the data of $p$ and $q$ from $\comp pq$, we need another combinator $\coml$ to decompose $\comp pq$.

\begin{defi} \label{defbiilp}
A $\biilp$-{\it algebra} $\hka$ is a $\biikuro$-algebra which contains $\coml$ and $\comp$ such that $\forall x, y, z \in \uhka$, $\coml x (\comp y z) = x y z$.
\end{defi}

A fundamental example of $\biilp$-algebras is given as the untyped planar lambda calculus with tensor products.

\begin{exa} \label{examplaten}
Add the following term construction rules to the planar lambda calculus (Example \ref{exampla}).
\begin{center}
	\AxiomC{$\Gamma \vdash M $}
	\AxiomC{$\Delta \vdash N $}
	\RightLabel{\scriptsize (pair construction)}
	\BinaryInfC{$\Gamma , \Delta \vdash M \otimes N $}
	\DisplayProof \; \;
	\AxiomC{$\Gamma \vdash M $}
	\AxiomC{$\Delta, x, y \vdash N $}
	\RightLabel{\scriptsize (pair deconstruction)}
	\BinaryInfC{$\Delta, \Gamma \vdash \llet{x \otimes y}{M}{N}$}
	\DisplayProof
\end{center}
We define a relation $\sim$ on planar terms as the congruence of the following relations.
\begin{itemize}
\item $(\lambda x.M)N \sim M[N/x]$
\item $M \sim \lambda x.Mx$
\item $(\llet{x_1 \otimes x_2}{M_1 \otimes M_2}{N}) \sim N[M_1 /x_1][M_2 /x_2]$
\item $M \sim (\llet{x \otimes y}{M}{x \otimes y})$
\end{itemize}
Let the equational relation $\eqbe$ be the reflexive, symmetric and transitive closure of $\sim$.
Closed terms modulo $\eqbe$ form a $\biilp$-algebra, which we call $\plalamt$ in this paper.
Here $\lambda xyz.x(yz)$, $\lambda tu.(\llet{x \otimes y}{u}{txy})$ and $\lambda xy. (x \otimes y)$ are the representatives of $\batui$, $\coml$ and $\comp$ respectively.

Unlike the planar lambda calculus of Example \ref{exampla2} (that does not have tensor products) does not need the $\eta$-equality to be a $\biikuro$-algebra, the planar lambda calculus with tensor products of Example \ref{examplaten} needs the $\beta \eta$-equality to use $\lambda xyz.x(yz)$ as $\batui$.
Indeed, $(\lambda xyz.x(yz)) ((\lambda u.u) \otimes (\lambda v.v)) (\lambda w.w)$ is $\beta \eta$-equal to $(\lambda u.u) \otimes (\lambda v.v)$ but not $\beta$-equal to it.
\end{exa}

\begin{rem}
When constructing linear lambda terms with tensor products, we often suppose a constant $\star$ for the unit (\cf~\cite{linlam}). For the above example, we can add the following rules to the term construction rules.
\begin{center}
	\AxiomC{}
	\RightLabel{\scriptsize (star introduction)}
	\UnaryInfC{$\vdash \star$}
	\DisplayProof \; \;
	\AxiomC{$\vdash M $}
	\AxiomC{$\Gamma \vdash N $}
	\RightLabel{\scriptsize (star elimination)}
	\BinaryInfC{$\Gamma \vdash \llet{\star}{M}{N}$}
	\DisplayProof
\end{center}
However, for our aim that constructing monoidal categories by categorical realizability, this $\star$ is not needed since we can use $\comi$ as the realizer of the unit instead of $\star$.
\end{rem}

$\biilp$-algebras correspond to the lambda calculus with tensor products, which has components other than applications, unlike the ordinary/linear/planar lambda calculus.
Thus, we cannot state the combinatory completeness property for $\biilp$-algebras in the same way we have seen in previous sections.
Here we only show the special case of the combinatory completeness property for $\biilp$-algebras.

\begin{prop}
Any closed term $M$ in $\plalamt$ is $\beta \eta$-equivalent to some term $\lkakko M \rkakko$ that is constructed from $\comb := \lambda xyz.x(yz)$, $\comi := \lambda x.x$, $\coml := \lambda tu.(\llet{x \otimes y}{u}{txy})$ and $\comp := \lambda xy. x \otimes y$ using the application and the unary operation $\kuro{(\haih)}:M \mapsto \lambda x.xM$.
\end{prop}

\begin{proof}
We inductively define the function $\lkakko \haih \rkakko$.
\begin{itemize}
\item $\lkakko x \rkakko := x$
\item $\lkakko MN \rkakko := \lkakko M \rkakko \lkakko N \rkakko$
\item $\lkakko M \otimes N \rkakko := \comp \lkakko M \rkakko \lkakko N \rkakko$
\item $\lkakko \llet{x \otimes y}{M}{N} \rkakko := \coml \lkakko \lambda xy.N \rkakko \lkakko M \rkakko$

\item $\lkakko \lambda xy.M \rkakko := \lkakko \lambda x. \lkakko \lambda y.M \rkakko \rkakko$

\item $\lkakko \lambda x.x \rkakko := \comi$
\item $\lkakko \lambda x.MN \rkakko := \begin{cases}
\comb \kuro{\lkakko N \rkakko} \lkakko \lambda x.M \rkakko & (x \in FV(M)) \\
\comb \lkakko M \rkakko \lkakko \lambda x.N \rkakko & (x \in FV(N))
\end{cases}$

\item $\lkakko \lambda x.M \otimes N \rkakko := \begin{cases}
\comb \kuro{\lkakko N \rkakko} (\comb \comp \lkakko \lambda x.M \rkakko) & (x \in FV(M)) \\
\comb (\comp \lkakko M \rkakko) \lkakko \lambda x.N \rkakko & (x \in FV(N))
\end{cases}$

\item $\lkakko \lambda x.(\llet{y \otimes z}{M}{N}) \rkakko := \begin{cases}
\comb (\coml \lkakko \lambda yz.N \rkakko) \lkakko \lambda x.M \rkakko & (x \in FV(M)) \\
\comb \kuro{\lkakko M \rkakko} (\comb \coml \lkakko \lambda xyz.N \rkakko) & (x \in FV(N))
\end{cases}$
\end{itemize}
It is easy to see that $M \eqbe \lkakko M \rkakko$ for any closed term $M$.
\end{proof}

Next we give an example of $\biilp$-algebra similar to Example \ref{examtree}.

\begin{exa} \label{examtree2}
Take an ordered group $(G,\cdot,e,\leq)$.
Let $T'$ be a set whose elements are constructed grammatically as follows:
\[ t ::= g \; | \;  t \rimp t' \; | \; t \otimes t' \; \; \; (g \in G). \]
That is, $T'$ is a set of binary trees whose leaves are labeled by elements of $G$, and whose nodes are two colored by $\rimp$ and $\otimes$.
We further define a function $| \haih | :T' \migi G$ by induction: $|g| := g$, $|t_2 \rimp t_1| := |t_2| \cdot |t_1|^{-1}$ and $|t_1 \otimes t_2| := |t_1| \cdot |t_2|$.

Let $|\hkt'|$ be the powerset of $\{ t \in T' \mid e \leq |t| \}$.
Then we can get a $\biilp$-algebra $\hkt'$ by $|\hkt'|$:
\begin{itemize}
\item For $M,N \in |\hkt'|$, $MN := \{ t_2 \mid \exists t_1 \in N, (t_2 \rimp t_1) \in M \}$.
\item $\comb := \{ (t_3 \rimp t_1) \rimp (t_2 \rimp t_1) \rimp (t_3 \rimp t_2) \mid t_1,t_2,t_3 \in T' \}$.
\item $\comi := \{ t_1 \rimp t_1 \mid t_1 \in T' \}$.
\item $\batui := \{ t_1 \rimp (t_2 \rimp t_2) \rimp t_1 \mid t_1,t_2 \in T' \}$.
\item $\coml := \{ t_3 \rimp (t_1 \otimes t_2) \rimp (t_3 \rimp t_2 \rimp t_1) \mid t_1,t_2,t_3 \in T' \}$.
\item $\comp := \{ (t_1 \otimes t_2) \rimp t_2 \rimp t_1 \mid t_1,t_2 \in T' \}$.
\item For $M \in |\hkt'|$, $\kuro{M} := \{ t_2 \rimp (t_2 \rimp t_1) \mid t_1 \in M, t_2 \in T' \}$.
\end{itemize}
\end{exa}

In the above example, we prepare $\otimes$ in the construction of $T'$ to express $\coml$ and $\comp$.
However, in fact, $\hkt$ in Example \ref{examtree} is already a $\biilp$-algebra even without $\otimes$.

\begin{exa} \label{examtree3}
In $\hkt$ of Example \ref{examtree}, we have $\comp$-combinator and $\coml$-combinator as
\begin{itemize}
\item $\comp := \{ t_1 \rimp (e \rimp t_2) \rimp t_2 \rimp t_1 \mid t_1,t_2 \in T \}$;
\item $\coml := \{ t_3 \rimp (t_1 \rimp (e \rimp t_2)) \rimp (t_3 \rimp t_2 \rimp t_1) \mid t_1,t_2,t_3 \in T \}$.
\end{itemize}
\end{exa}

This is a less standard example in that $\hkt$ uses $t_1 \rimp (e \rimp t_2)$ as the role of $t_1 \otimes t_2$.

For another example, as well as we can construct an LCA (and the based $\bci$-algebra) from a ``reflexive object'' (See \cite{ahs} and \cite{Haghverdi}), we can get $\biilp$-algebras by appropriate settings.

\begin{exa} \label{examreflexive}
Let $(\catc,\otimes, I)$ be a monoidal closed category and \\
$\Phi:\catc(- \otimes X, -) \migi \catc(-,- \rimp X)$ be the adjunction.
Suppose an object $V$ that has:
\begin{enumerate}[(i)]
\item an isomorphism $r:(V \rimp V) \migi V$ and $s := r^{-1}$;
\item a {\it retraction} $t: (V \otimes V) \triangleleft V:u$, that is, maps $t:V \otimes V \migi V$ and $u:V \migi V \otimes V$ such that $u \circ t =id_{V \otimes V}$.
\end{enumerate}
Then the set of maps $\catc(I,V)$ is a $\biilp$-algebra.
\begin{itemize}
\item For maps $M,N:I \migi V$, the application is defined as 
\[ I \xmigi{{\rm unitor}} I \otimes I \xmigi{(s \circ M) \otimes N} (V \rimp V) \otimes V \xmigi{ev} V. \]
\item Take a map $f:(V \otimes V) \otimes V \migi V$ as
\[ (V \otimes V) \otimes V \xmigi{{\rm associator}} V \otimes (V \otimes V) \xmigi{s \otimes (ev \circ (s \otimes id))} (V \rimp V) \otimes V \xmigi{ev} V. \]
The $\comb$-combinator is given as $r \circ \Phi (r \circ \Phi (r \circ \Phi(f)) \circ \lambda_V)$, where $\lambda_V:I \otimes V \migi V$ is the unitor.
\item The $\comi$-combinator is $r \circ \Phi(\lambda_V)$.
\item The $\comb$-combinator given above satisfies the axiom of the $\batui$-combinator.
Here we use
$r \circ s =id_V$, and thus we need to assume $r$ is an isomorphism (not merely a retraction).
\item Take a map $g:V \otimes V \migi V$ as
\[ V \otimes V \xmigi{s \otimes u} (V \rimp V) \otimes (V \otimes V) \xmigi{ev \circ ({\rm associator})} V \otimes V \xmigi{ev \circ (s \otimes id)} V. \]
The $\coml$-combinator is $r \circ \Phi (r \circ \Phi(g) \circ \lambda_V)$.
\item The $\comp$-combinator is $r \circ \Phi (r \circ \Phi(t) \circ \lambda_V)$.
\item Given arbitrary $M:I \migi V$, $\kuro{M}$ is $r \circ \Phi(ev \circ (s \otimes M) \circ \rho_V \circ \lambda_V)$. Here $\rho_V :V \migi V \otimes I$ is the unitor.
\end{itemize}
\end{exa}

We will use the above $\biilp$-algebra later in the last of Section \ref{secexpplca}.

Next we show that $\biilp$-algebras induce monoidal closed categories.

\begin{prop} \label{propmonoclo}
When $\hka$ is a $\biilp$-algebra, $\asmca$ is a monoidal closed category.
\end{prop}

\begin{proof}
Since $\hka$ is also a $\bikuro$-algebra, we can use the combinatory completeness for the planar lambda calculus.
\begin{itemize}
\item For objects $X$ and $Y$, the underlying set of $X \otimes Y$ is $|X| \times |Y|$. Realizers are defined as 
\[ \rlz{x \otimes y}_{X \otimes Y} := \{ \comp p q \mid p \in \rlz{x}_X, q \in \rlz{y}_Y \}. \]
\item For $f: X \migi X'$ and $g:Y \migi Y'$, the map $f \otimes g$ is the function sending $x \otimes y$ to $f(x) \otimes g(y)$.
A realizer for $f \otimes g$ is $\coml (\lamst pq. \comp (r_f p)(r_g q))$.
\item The underlying set of the unit object $I$ is a singleton $\{ \ast \}$. The realizer is $\rlz{\ast}_I := \{ \comi \}$.
\item The left unitor $\lambda_X: I \otimes X \migi X$ sends $\ast \otimes x$ to $x$, whose realizer is $\coml \comi$.
A realizer of $\lambda_X^{-1}$ is $\comp \comi$.
\item The right unitor $\rho_X: X \migi X \otimes I$ sends $x$ to $x \otimes \ast$, whose realizer is $\lamst p. \comp p \comi$.
A realizer of $\rho_X^{-1}$ is $\coml \batui$.
\item The associator $\alpha_{XYZ}:(X \otimes Y) \otimes Z \migi X \otimes (Y \otimes Z)$ sends $(x \otimes y) \otimes z$ to $x \otimes (y \otimes z)$. 
A realizer of $\alpha_{XYZ}$ is $\coml (\coml (\lamst pqr. \comp p (\comp qr)))$.
A realizer of $\alpha_{XYZ}^{-1}$ is $\coml (\lamst pu. \coml (M p) u)$, where $M := \lamst pqr. \comp (\comp pq) r$.
\item For objects $X$ and $Y$, the underlying set of $Y \rimp X$ is $\homrm_{\asmca}(X,Y)$. Realizers are defined as
$\rlz{f}_{Y \rimps X} := \{ r \mid \mbox{$r$ realizes $f$} \}$.
\item For $f: X' \migi X$ and $g:Y \migi Y'$, $g \rimp f$ is the function sending a map $h : X \migi Y$ to $g \circ h \circ f : X' \migi Y'$.
A realizer for $g \rimp f$ is $\lamst uv. r_g (u (r_f v))$.
\item The evaluation map $ev :(Y \rimp X) \otimes X \migi Y$ sends $f \otimes x$ to $f(x)$, which is realized by $\coml \comi$.
\item For any map $f:Z \otimes X \migi Y$, there exists a unique map $g:Z \migi (Y \rimp X)$ which satisfies
$ev \circ (g \otimes id_X) = f$. This $g$ is given as the function sending $z$ to the function $x \mapsto f(z \otimes x)$, which is realized by $\lamst rp. r_f (\comp rp)$. \qedhere
\end{itemize}
\end{proof}

Similar to the case of $\bci$-algebras (Proposition \ref{propsmcc} and \ref{propsmcc2}), we cannot use the same proof of Proposition \ref{propmonoclo} to the case of $\moda$.
We prove that $\moda$ on a $\biilp$-algebra $\hka$ is a monoidal closed category by the same modification used in the proof of Proposition \ref{propsmcc2}.
That is, we take the inclusion functor $G:\moda \hookrightarrow \asmca$ and the left adjoint $F: \asmca \migi \moda$, and define the tensor product $\boxtimes$ in $\moda$ as $X \boxtimes Y := F(GX \otimes GY)$.

\begin{prop} \label{propmodmcc}
When $\hka$ is a $\biilp$-algebra, $\moda$ is a monoidal closed category. \qed
\end{prop}

For functors given by applicative morphisms between $\biilp$-algebras, the next properties hold.

\begin{prop} \label{proplax}
Let $\hka_1$ and $\hka_2$ be $\biilp$-algebras and $\gamma :\hka_1 \migi \hka_2$ is an applicative morphism. Then $\gamst:\asmc{\hka_1} \migi \asmc{\hka_2}$ is a lax monoidal functor.
\end{prop}

\begin{proof}
A realizer for $I_2 \migi \gamst (I_1)$ is in the set $\lamst u.u (\gamma(\comi_1))$. \\
A realizer for $(\gamst X) \otimes_2 (\gamst Y) \migi \gamst (X \otimes_1 Y)$ is in $\coml_2 (\lamst pq.r_{\gamma} (r_{\gamma} (\gamma \comp_1) p) q)$.
\end{proof}

\begin{prop} \label{propadjbiilp}
For $\biilp$-algebras $\hka_1$ and $\hka_2$ and an adjoint pair \\
$(\delta \dashv \gamma) : \hka_1 \migi \hka_2$, the adjunction $(\delst \dashv \gamst):\asmc{\hka_1} \migi \asmc{\hka_2}$ is monoidal.
\end{prop}

\begin{proof}
We show that the left adjoint $\delst$ is strong monoidal.
Since $\delst$ is lax monoidal by the previous proposition, it is sufficient to show that there are realizers for maps $\delst I_2 \migi I_1$ and $\delst(X \otimes_2 Y) \migi \delst X \otimes_1 \delst Y$.
A realizer for the former is 
$\lamst x. \bfe (r_{\delta} (\delta (\lamst y.y (\gamma \comi_1))) x)$.
A realizer for the latter is 
$\lamst z.\bfe (r_{\delta} (\delta (\coml (\lamst uv.r_{\gamma} (r_{\gamma} (\gamma \comp) (\bfi u) ) (\bfi v)))) z)$.
Here $\bfe \in |\hka_1|$ is an element such that $\forall x \in |\hka_1|, \bfe (\delta (\gamma x)) =x$ and $\bfi \in |\hka_2|$ is an element such that $\forall y \in |\hka_2|, \bfi y =\gamma (\delta y)$, that are obtained by the assumption that $\gamma$ and $\delta$ form an adjoint pair.
\end{proof}

%%%%%%%%%%%%%%%%%%%%%%%%%%%%%%%%%%%%%%%%%%%%%%%%%%%%%%%%%%%%%%%%%%%
\subsection{Bi-$\bdi$-algebras and monoidal bi-closed categories} \label{secbdi}

Let us consider once again why non-symmetric tensor products in categories of assemblies cannot be constructed from $\biikuro$-algebras,
from the viewpoint of the ``polymorphic encoding.''
In the second-order linear logic, a tensor product $X \otimes Y$ can be interpreted as
$\forall \alpha . (X \limp Y \limp \alpha) \limp \alpha$. (This interpretation is seen in \cite{polymorphic}, for instance.)
This formula $(X \limp Y \limp \alpha) \limp \alpha$ corresponds to the type inhabited by $\lambda t.txy$ in the typed linear lambda calculus.
This correspondence connected to that (in a PCA or a $\bci$-algebra,) a realizer of $x\otimes y \in |X \otimes Y|$ is $\lamst t.tp q$ for $p \in \rlz{x}_X$ and $q \in \rlz{y}_Y$.
What matters here is that the interpretation $X \otimes Y \cong \forall \alpha . (X \limp Y \limp \alpha) \limp \alpha$ holds only when the tensor product is symmetric.
Whereas, for the non-symmetric cases, $X \otimes Y$ is expressed as $\forall \alpha. (\alpha \rimp Y \rimp X) \limp \alpha$ or $\forall \alpha. \alpha \rimp (Y \limp X \limp \alpha)$.
Here we need to distinguish two sorts of implications $\rimp$ and $\limp$.
In an applicative structure like a $\biikuro$-algebra, we cannot distinguish them since we only have one sort of application.

Conversely, providing some structure in an applicative structure $\hka$ that allows to distinguish these two implications, we may be able to construct non-symmetric tensor products in $\asmca$.
From this viewpoint, we introduced bi-$\bdi$-algebras in \cite{tomita2}.

In this subsection, we recall bi-$\bdi$-algebras from \cite{tomita2}.
First we recall a variant of the lambda calculus, which is an example of an applicative structure with two sorts of applications.

\begin{defi}
{\it Bi-planar lambda terms} are constructed by the following rules:
\begin{center}
\AxiomC{}
\RightLabel{\scriptsize (identity)}
\UnaryInfC{$x \vdash x$}
\DisplayProof \; \; 
\AxiomC{$\Gamma, x \vdash M$}
\RightLabel{\scriptsize (right abstraction)}
\UnaryInfC{$\Gamma \vdash \rlamk{x}{M}$}
\DisplayProof \; \; 
\AxiomC{$x, \Gamma \vdash M$}
\RightLabel{\scriptsize (left abstraction)}
\UnaryInfC{$\Gamma \vdash \llamk{x}{M}$}
\DisplayProof
\end{center}
\begin{center}
\AxiomC{$\Gamma \vdash M$}
\AxiomC{$\Delta \vdash N$}
\RightLabel{\scriptsize (right application)}
\BinaryInfC{$\Gamma, \Delta \vdash M \rappk N$}
\DisplayProof \; \; 
\AxiomC{$\Delta \vdash N$}
\AxiomC{$\Gamma \vdash M$}
\RightLabel{\scriptsize (left application)}
\BinaryInfC{$\Delta, \Gamma \vdash N \lappk M$}
\DisplayProof
\end{center}
Note that here is none of weakening, contraction nor exchange rules.

For the sake of clarity, we will classify right and left by red and blue color.
That is, we write each of them as $M \rapp N$, $\rlaml{x}{M}$, $N \lapp M$ and $\llaml{x}{M}$.

We define a relation $\migi_{\beta}$ on bi-planar lambda terms as the congruence of the following relations:
\begin{itemize}
\item (right $\beta$-reduction) $\rlaml{x}{M} \rapp N \migi_{\beta} M[N/x]$
\item (left $\beta$-reduction) $N \lapp \llaml{x}{M} \migi_{\beta} M[N/x]$
\end{itemize}

The {\it bi-planar lambda calculus} consists of bi-planar lambda terms and the reflexive, symmetric and transitive closure of $\migi_{\beta}$ as the equational relation $\eqb$.
\end{defi}

Basic properties about the $\beta$-reduction $\migi_{\beta}$, such as the confluence and the strongly normalizing property, can be shown in the same way as the proof for the linear lambda calculus.

\begin{rem}
The bi-planar lambda calculus is not essentially a new concept, since it often appears as the Curry-Howard corresponding calculus with the Lambek calculus (\cf~\cite{paiva}).
However, note that unlike the calculus corresponding to the Lambek calculus, the bi-planar lambda calculus is based on untyped setting.
The reason why we use a less-standard notation is to shorten the length of terms and to make them easier to read.
\end{rem}

Then we define a class of applicative structures which we call bi-$\bdi$-algebras.

\begin{defi} \label{defbdi}
A total applicative structure $\hka=(\uhka,\rapp)$ is a {\it bi-$\bdi$-algebra} iff there is an additional total binary operation $\lapp$ on $\uhka$ and $\uhka$ contains several special elements:
\begin{enumerate}
\item $\combr \in \uhka$ such that $\forall x,y,z \in \uhka$, $((\combr \rapp x) \rapp y) \rapp z = x \rapp (y \rapp z)$.
\item $\combl \in \uhka$ such that $\forall x,y,z \in \uhka$, $z \lapp (y \lapp (x \lapp \combl)) = (z \lapp y) \lapp x$.
\item $\comdr \in \uhka$ such that $\forall x,y,z \in \uhka$, $x \lapp ((\comdr \rapp y) \rapp z) = (x \lapp y) \rapp z$.
\item $\comdl \in \uhka$ such that $\forall x,y,z \in \uhka$, $(z \lapp (y \lapp \comdl)) \rapp x = z \lapp (y \rapp x)$.
\item $\comir \in \uhka$ such that $\forall x \in \uhka$, $\comir \rapp x = x$.
\item $\comil \in \uhka$ such that $\forall x \in \uhka$, $x \lapp \comil = x$.
\item For each $a \in \uhka$, $\dagr{a} \in \uhka$ such that $\forall x \in \uhka$, $(\dagr{a}) \rapp x = x \lapp a$.
\item For each $a \in \uhka$, $\dagl{a} \in \uhka$ such that $\forall x \in \uhka$, $x \lapp (\dagl{a}) = a \rapp x$.
\end{enumerate}
We call $\rapp$ and $\lapp$ as {\it right application} and {\it left application} respectively.
We often write
$\hka = (\uhka,\rapp,\lapp)$ for a bi-$\bdi$-algebra $\hka=(\uhka,\rapp)$ with the left application $\lapp$.
\end{defi}

In the sequel, we use $\rapp$ as a left-associative operation and often omit unnecessary parentheses, while we do not omit parentheses for $\lapp$.
For instance, $(u \rapp v \rapp w) \lapp ((x \lapp y) \lapp z)$ denotes $((u \rapp v) \rapp w) \lapp ((x \lapp y) \lapp z)$.

The definition of bi-$\bdi$-algebras is intended having a good correspondence with the bi-planar lambda calculus.

\begin{exa} \label{exambipla}
Untyped closed bi-planar lambda terms modulo $\eqb$ form a bi-$\bdi$-algebra, which we call $\plalamb$ in this paper.
We give a few examples of representatives: $\rlaml{x}{\rlaml{y}{\rlaml{z}{x \rapp (y \rapp z)}}}$ represents $\combr$; $\llaml{y}{\llaml{x}{\rlaml{z}{x \lapp (y \rapp z)}}}$ represents $\comdl$; $\llaml{x}{M \rapp x}$ represents $\dagl{M}$.
\end{exa}

\begin{prop}[combinatory completeness for bi-$\bdi$-algebras]
Let $\hka = (\uhka,\rapp,\lapp)$ be a bi-$\bdi$-algebra.
A {\it polynomial} over $\hka$ is defined as a syntactic expression generated by variables, elements of $\uhka$ and the applications $\rapp$ and $\lapp$.
For a polynomial $M$ over $\hka$ and the rightmost variable $x$ of $M$, if $x$ appears exactly once in $M$, there exists a polynomial $M'$ such that the free variables of $M'$ are the free variables of $M$ excluding $x$ and $M' \rapp a = M[a/x]$ for all $a \in \uhka$. We write such $M'$ as $\rlam{x}{M}$.
Also, for a polynomial $N$ over $\hka$ and the leftmost variable $y$ of $N$, if $y$ appears exactly once in $N$, there exists a polynomial $N'$ such that the free variables of $N'$ are the free variables of $N$ excluding $y$ and $a \lapp N' = N[a/y]$ for all $a \in \uhka$. We write such $N'$ as $\llam{y}{N}$.
\end{prop}

\begin{proof}
We define $\rlam{x}{M}$ by induction on the structure of $M$.
\begin{itemize}
\item $\rlam{x}{x} := \comir$.
\item $\rlam{x}{M \rapp N} := \begin{cases}
		\combr \rapp \dagr{(\comdr \rapp \comil \rapp N)} \rapp \rlam{x}{M} & (x \in FV(M)) \\
		\combr \rapp M \rapp \rlam{x}{N} & (x \in FV(N))
		\end{cases}$ \\
Note that in case $x \in FV(M)$, $\comdr \rapp \comil \rapp N$ has no variables since $x$ is the rightmost free variable in $M \rapp N$.
\item $\rlam{x}{N \lapp M} := \begin{cases}
		N \lapp (\rlam{x}{M} \lapp \comdl) & (x \in FV(M)) \\
		\combr \rapp (\dagr{M}) \rapp \rlam{x}{N} & (x \in FV(N))
		\end{cases}$ \\
Note that in case $x \in FV(N)$, $M$ has no variables since $x$ is the rightmost free variable in $N \lapp M$.
\end{itemize}
The case of the left abstraction $\llam{y}{N}$ is given in the same way, with all the left and right constructs reversed.
\end{proof}

Next we give another example of bi-$\bdi$-algebra which is introduced in \cite{tomita2} and similar to Example \ref{examtree2}.

\begin{exa} \label{examtree4}
Take an ordered group $(G,\cdot,e,\leq)$.
Let $T''$ be a set whose elements are constructed grammatically as follows:
\[ t ::= g \; | \;  t \rimp t' \; | \; t \limp t' \; \; \; (g \in G). \]
That is, $T''$ is a set of binary trees whose leaves are labeled by elements of $G$, and whose nodes are two colored by $\rimp$ and $\limp$.
We further define a function $| \haih | :T'' \migi G$ by induction: $|g| := g$, $|t_2 \rimp t_1| := |t_2| \cdot |t_1|^{-1}$ and $|t_1 \limp t_2| := |t_1|^{-1} \cdot |t_2|$.

Let $|\hkt''|$ be the powerset of $\{ t \in T'' \mid e \leq |t| \}$.
Then we can get a bi-$\bdi$-algebra $\hkt''$ by $|\hkt''|$:
\begin{itemize}
\item For $M,N \in |\hkt''|$, $M \rapp N := \{ t_2 \mid \exists t_1 \in N, (t_2 \rimp t_1) \in M \}$.
\item For $M,N \in |\hkt''|$, $N \lapp M := \{ t_2 \mid \exists t_1 \in N, (t_1 \limp t_2) \in M \}$.
\item $\combr := \{ (t_3 \rimp t_1) \rimp (t_2 \rimp t_1) \rimp (t_3 \rimp t_2) \mid t_1,t_2,t_3 \in T'' \}$, dual for $\combl$.
\item $\comdr := \{ ((t_1 \limp t_2) \rimp t_3) \rimp (t_1 \limp (t_2 \rimp t_3)) \mid t_1,t_2,t_3 \in T'' \}$, dual for $\comdl$.
\item $\comir := \{ t_1 \rimp t_1 \mid t_1 \in T'' \}$, dual for $\comil$.
\item For $M \in |\hkt''|$, $\dagr{M} := \{ t_2 \rimp t_1 \mid (t_1 \limp t_2) \in M \}$, dual for $\dagl{M}$.
\end{itemize}
\end{exa}

In the above example, we prepare $\limp$ in the construction of $T''$ to express the left application.
However, in fact, $\hkt$ in Example \ref{examtree} is a bi-$\bdi$-algebra even without preparing $\limp$.

\begin{exa} \label{examtree5}
Let $T$ be the same set in Example \ref{examtree}.
For $t, t' \in T$, we define $t \limp t' \in T$ as $(e \rimp t) \rimp (e \rimp t')$.
Then $\hkt$ is a bi-$\bdi$-algebra, whose components are taken in the same way as Example \ref{examtree4}.
\end{exa}

Next we give some basic properties of bi-$\bdi$-algebras.

\begin{prop} \hfill
\begin{enumerate}
\item Any bi-$\bdi$-algebra is also a $\biilp$-algebra. \label{property1}
\item Any $\bci$-algebra is also a bi-$\bdi$-algebra whose left and right applications coincide. \label{property2}
\item When $\hka=(\uhka,\rapp)$ is a bi-$\bdi$-algebra, the left application is unique up to isomorphism. That is, when both $(\uhka,\rapp,\lapp_1)$ and $(\uhka,\rapp,\lapp_2)$ are bi-$\bdi$-algebras, $\hka_1 = (\uhka,\rapp_1)$ and $\hka_2 = (\uhka,\rapp_2)$ are isomorphic as applicative structures, where $x \rapp_i y := y \lapp_i x$. \label{property3}
\item Let $\hka = (\uhka,\rapp,\lapp)$ be a bi-$\bdi$-algebra and take an applicative structure
$\hka' := (\uhka,\rapp')$ by $x \rapp' y := y \lapp x$.
Then $\hka$ is a $\bci$-algebra iff $\hka'$ is a $\bci$-algebra.
Moreover, in such a case, $\hka$ and $\hka'$ are isomorphic as applicative structures. \label{property4}
\end{enumerate}
\end{prop}

\begin{proof} \hfill
\begin{enumerate}
\item $\comb$, $\comi$, $\batui$, $\coml$, $\comp$ and $\kuro{a}$ are given as $\combr$, $\comir$, $\rlam{x}{\rlam{y}{x \lapp (y \rapp \comil)}}$, $\rlam{x}{\rlam{y}{x \lapp y}}$, \\
$\rlam{x}{\rlam{y}{\llam{t}{t \rapp x \rapp y}}}$ and $\dagr{\llam{x}{x \rapp a}}$ respectively.

\item For a $\bci$-algebra $(\uhka, \rapp)$, $(\uhka, \rapp, \lapp)$ is a bi-$\bdi$-algebra when we take $y \lapp x := x \rapp y$. Here $\combr = \combl := \comb$, $\comdr = \comdl := \comc$, $\comir = \comil := \comi$ and $\dagr{a} = \dagl{a} := a$.

\item By the combinatory completeness of $\hka_2$, we have $L := \rlam{y}{\rlam{x}{y \lapp_2 x}}$ such that $L \rapp y \rapp x = y \lapp_2 x = x \rapp_2 y$.
By the combinatory completeness of $\hka_1$, we have an element
$r := \llam{x}{\llam{y}{L \rapp y \rapp x}}$, which satisfies $r \rapp_1 x \rapp_1 y = L \rapp y \rapp x = x \rapp_2 y$.
This $r$ realizes the applicative morphism $i_1 : \hka_1 \migi \hka_2$ given as the identity function on $\uhka$.
Similarly we have the inverse applicative morphism $i_2 : \hka_2 \migi \hka_1$ given as the identity function.
$i_1$ and $i_2$ are the isomorphisms between $\hka_1$ and $\hka_2$.

\item Suppose that $\hka$ is a $\bci$-algebra, that is, there is some element $\comcr \in \uhka$ such that
$\comcr \rapp x \rapp y \rapp z = x \rapp z \rapp y$.
Take an element $\comcl := \llam{x}{\llam{y}{\llam{z}{\comcr \rapp M \rapp z \rapp y \rapp x}}}$, where $M := \rlam{y}{\rlam{z}{\rlam{x}{y \lapp (z \lapp x)}}}$.
$\combl$, $\comil$ and $\comcl$ make $\hka'$ a $\bci$-algebra.
Similarly, when we suppose $\hka'$ is a $\bci$-algebra, $\hka$ is also a $\bci$-algebra.

Furthermore, when we suppose $\hka$ (and also $\hka'$) is a $\bci$-algebra, we have an element
$r := \comcr \rapp \rlam{y}{\rlam{x}{y \lapp x}}$, which realizes the applicative morphism $i : \hka' \migi \hka$ given as the identity function.
Similarly we have the inverse applicative morphism $i' : \hka \migi \hka'$ given as the identity function, and thus $\hka \cong \hka'$. \qedhere
\end{enumerate}
\end{proof}

By (\ref{property1}) and (\ref{property2}) of the above proposition, the class of bi-$\bdi$-algebras is the class of applicative structures in between $\biilp$-algebras and $\bci$-algebras. 
We named the ``$\coml$-combinator" of $\biilp$-algebras by the reason that it is represented as $\rlam{x}{\rlam{y}{x \lapp y}}$ in a bi-$\bdi$-algebra, that gives the ``left" application of two arguments.

\begin{rem}
Although $\rlam{x}{\rlam{y}{x \lapp y}}$ always acts as a $\coml$-combinator in a bi-$\bdi$-algebra, it is not the only way to take a $\coml$-combinator.
Indeed, in Example \ref{examtree5}, $\hkt$ has a $\coml$-combinator as
\begin{center}
\begin{eqnarray*}
\rlam{x}{\rlam{y}{x \lapp y}} &=& \combr \rapp \dagr{(\comir \lapp \comdl)} \rapp \comir \\
&=& \{ t_2 \rimp ((e \rimp t_1) \rimp (e \rimp t_2)) \rimp t_1 \mid t_1,t_2 \in T \} \mbox{,}
\end{eqnarray*}
\end{center}
which is different from the $\coml$-combinator taken in Example \ref{examtree3}.
\end{rem}

Since a bi-$\bdi$-algebra $\hka$ is also a $\biilp$-algebra, we know that $\asmca$ (and $\moda$) is a monoidal closed category.
Moreover, we can show that the categories of assemblies on bi-$\bdi$-algebras are not just a monoidal closed categories, but are monoidal bi-closed categories, having richer categorical structures.
A {\it monoidal bi-closed category} is a monoidal category $\catc$ with two sorts of adjunction $\catc(X \otimes Y,Z) \cong \catc(X,Z \rimp Y)$ and $\catc(X \otimes Y,Z) \cong \catc(Y,X \limp Z)$. 

\begin{prop} \label{propbiclo}
When $\hka=(\uhka,\rapp)$ is a bi-$\bdi$-algebra, $\asmca$ is a monoidal bi-closed category.
\end{prop}

\begin{proof}
Let $\lapp$ be the left application of $\hka$.
\begin{itemize}
\item A realizer for identities is $\comir$.
\item A realizer for the composition of $f:X \migi Y$ and $g:Y \migi Z$ is $\combr \rapp r_g \rapp r_f$.

\item For objects $X$ and $Y$, the underlying set of $X \otimes Y$ is $|X| \times |Y|$. Realizers are defined as 
\[ \rlz{x \otimes y} := \{ \llam{t}{t \rapp p \rapp q} \mid p \in \rlz{x}_X, q \in \rlz{y}_Y \}. \]

\item For $f: X \migi X'$ and $g:Y \migi Y'$, the map $f \otimes g$ is the function sending $x \otimes y$ to $f(x) \otimes g(y)$.
A realizer for $f \otimes g$ is $\rlam{u}{\rlam{p}{\rlam{q}{\llam{t}{t \rapp (r_f \rapp p) 
\rapp (r_g \rapp q)}}} \lapp u}$.

\item The underlying set of the unit object $I$ is a singleton $\{ \ast \}$. The realizer is $\rlz{\ast}_I := \{ \comir \}$.

\item The left unitor $\lambda_X: I \otimes X \migi X$ sends $\ast \otimes x$ to $x$, whose realizer is $\rlam{p}{\comir \lapp p}$.
A realizer of $\lambda_X^{-1}$ is $\rlam{p}{\llam{t}{t \rapp \comir \rapp p}}$.

\item The right unitor $\rho_X: X \migi X \otimes I$ sends $x$ to $x \otimes \ast$, whose realizer is $\rlam{p}{\llam{t}{t \rapp p \rapp \comir}}$.
A realizer of $\rho_X^{-1}$ is $\rlam{u}{\rlam{p}{\rlam{v}{p \lapp (v \rapp \comil)}} \lapp u}$.

\item The associator $\alpha_{XYZ}:(X \otimes Y) \otimes Z \migi X \otimes (Y \otimes Z)$ sends $(x \otimes y) \otimes z$ to $x \otimes (y \otimes z)$. 
$\alpha_{XYZ}$ is realized by $\rlam{u}{\rlam{v}{M \lapp v} \lapp u}$,
where 
$M := \rlam{p}{ \rlam{q}{\rlam{r}{\llam{t}{t \rapp p \rapp \llam{t'}{t' \rapp q \rapp r}}}}}$.
A realizer of $\alpha_{XYZ}^{-1}$ is $\rlam{u}{\rlam{p}{\rlam{v}{\rlam{q}{\rlam{r}{N}} \lapp v}} \lapp u}$, where 
$N := \llam{t}{(t \rapp \llam{t'}{t' \rapp p \rapp q}) \rapp r}$.

\item For objects $X$ and $Y$, the underlying set of $Y \rimp X$ is $\homrm_{\asmca}(X,Y)$. Realizers are
\[ \rlz{f}_{Y \rimps X} := \{ r \mid \mbox{$r$ is a realizer of $f$} \}. \]

\item For $f: X' \migi X$ and $g:Y \migi Y'$, $g \rimp f$ is the function sending a map $h : X \migi Y$ to the map $g \circ h \circ f : X' \migi Y'$.
A realizer for $g \rimp f$ is $\rlam{u}{\rlam{v}{r_g \rapp (u \rapp (r_f \rapp v))}}$.

\item The evaluation map $ev :(Y \rimp X) \otimes X \migi Y$ sends $f \otimes x$ to $f(x)$, which is realized by $\rlam{u}{\comir \lapp u}$.

\item For any map $f:Z \otimes X \migi Y$, there exists a unique map $g:Z \migi (Y \rimp X)$ which satisfies
$ev \circ (g \otimes id_X) = f$. This $g$ is given as the function sending $z$ to the function $x \mapsto f(z \otimes x)$, which is realized by $\rlam{q}{\rlam{p}{r_f \rapp \llam{t}{t \rapp q \rapp p}}}$.

\item For objects $X$ and $Y$, the underlying set of $X \limp Y$ is is $\homrm_{\asmca}(X,Y)$. Realizers are
\[ \rlz{f}_{X \limp Y} := \{ r \mid \mbox{$a \lapp r \in \rlz{f(x)}_Y$ for any $x \in |X|$ and $a \in \rlz{x}_X$} \}. \]
This set is not empty since $\dagl{(r_f)}$ is in the set for a realizer $r_f$ of $f$.

\item For $f: X' \migi X$ and $g:Y \migi Y'$, $f \limp g$ is the function sending a map $h : X \migi Y$ to the map $g \circ h \circ f : X' \migi Y'$. A realizer for $f \limp g$ is $\rlam{u}{\llam{v}{r_g \rapp ((r_f \rapp v) \lapp u)}}$.

\item The evaluation map $ev' : X \otimes (X \limp Y) \migi Y$ sends $x \otimes f$ to $f(x)$, which is realized by $\rlam{u}{\rlam{p}{\rlam{v}{p \lapp v}} \lapp u}$.

\item For any map $f:X \otimes Z \migi Y$, there exists a unique map $g:Z \migi (X \limp Y)$ which satisfies
$ev' \circ (id_X \otimes g) = f$. This $g$ is given as the function sending $z$ to the function $x \mapsto f(x \otimes z)$, which is realized by $\rlam{q}{\llam{p}{r_f \rapp \llam{t}{t \rapp p \rapp q}}}$. \qedhere
\end{itemize}
\end{proof}

For the category of modest sets, we use the same discussion as Proposition \ref{propsmcc2}.
That is, for the functor $F$ that is left adjoint of the inclusion functor $G:\moda \migi \asmca$, we define tensor products $\boxtimes$ in $\moda$ as $X \boxtimes Y := F(GX \otimes GY)$.

\begin{prop}
When $\hka=(\uhka,\rapp)$ is a bi-$\bdi$-algebra, $\moda$ is a monoidal bi-closed category. \qed
\end{prop}

In Proposition \ref{propbiclo}, $\asmca$ is the category of assemblies on the applicative structure $(\uhka, \rapp)$.
Even if we employ the left application $\lapp$ to construct the category of assemblies, we can obtain a category with the same structures as $\asmca$, as the next proposition says.

\begin{prop}
Let $\hka=(\uhka,\rapp,\lapp)$ be a bi-$\bdi$-algebra.
When we take an applicative structure $\hka'=(\uhka,\rapp')$ by $x \rapp' y := y \lapp x$, $\asmca$ and $\asmcad$ are isomorphic as categories.
Moreover, $\asmca$ is monoidally isomorphic to $\asmcad$ with the reversed tensor products.
That is, there is an isomorphism $R:\asmca \migi \asmcad$ such that $R(I) \cong I'$, $R^{-1}(I') \cong I$,
$R(X \otimes Y) \cong RY \otimes' RX$ and $R^{-1}(X' \otimes' Y') \cong R^{-1}Y' \otimes R^{-1}X'$ hold.
\end{prop}

\begin{proof}
For a map $f:X \migi Y$ in $\asmca$, the map is also a map in $\asmcad$ since the realizer exists as $\dagl{r_f}$.
Therefore, we can take a functor $R:\asmca \migi \asmcad$ which sends objects to the same objects and maps to the same maps.
Similarly we can get $R^{-1}$ which sends objects to the same objects and maps to the same maps.

$\hka'$ is a bi-$\bdi$-algebra by taking the left application $x \lapp' y := y \rapp x$.
We define the monoidal structure $(\otimes',I')$ on $\asmcad$ in the same way as Proposition \ref{propbiclo}. 
Here the realizers for tensor products are $\rlz{x \otimes' y}_{X \otimes' Y} = \{ \rlam{t}{q \lapp (p \lapp t)} \mid p \in \rlz{x}_X, q \in \rlz{y}_Y \}$.
A realizer for $R(I) \migi I'$ is $\llam{u}{u \rapp \comil}$ and a realizer for the inverse is $\llam{u}{\comir \lapp u}$.
A realizer for $R(X \otimes Y) \migi RY \otimes' RX$ is $\llam{u}{\rlam{p}{\rlam{q}{\rlam{t}{p \lapp (q \lapp t)}}} \lapp u}$ and a realizer for the inverse is $\llam{u}{u \rapp \llam{q}{\llam{p}{\llam{t}{t \rapp p \rapp q}}}}$.

Similar for the realizers related to $R^{-1}$.
\end{proof}

We can define ``partial bi-$\bdi$-algebras'' naturally.
Similar to partial $\bci$-algebras discussed in Remark \ref{rembci}, for a partial bi-$\bdi$-algebra $\hka$:
\begin{itemize}
\item $\asmca$ is not generally a monoidal bi-closed category;
\item adding an extra element $\bot$, $\hka$ naturally extends to a total bi-$\bdi$-algebra $\hka_{\bot}$;
\item $\asmca$ is the full subcategory of $\asmc{\hka_{\bot}}$.
\end{itemize}
Here the $\bot$ does not need to be two (for $\rapp$ and for $\lapp$), just one.

%%%%%%%%%%%%%%%%%%%%%%%%%%%%%%%%%%%%%%%%%%%%%%%%%%%%%%%%%%%%%%%%%
%%%%%%%%%%%%%%%%%%%%%%%%%%%%%%%%%%%%%%%%%%%%%%%%%%%%%%%%%%%%%%%%%
\section{Separation of classes of applicative strctures} \label{secsepa}

As we have already mentioned, the classes of applicative structures in this paper form a hierarchy summarized in the following table (Table \ref{tab:taiou}).
However, we have not yet shown the strictness of the hierarchy.
To show the strictness of the each inclusion, it is sufficient to provide an applicative structure separating the classes, that is, an applicative structure belonging to one side of the class but not belonging to the other.
In this section we give several such applicative structures, as summarized in Table \ref{tab:taiou}.

\begin{center}
\begin{table}[h]
\caption{Separation}
\centering
\begin{tabular}{|c|c|l|l|c|}
	\hline
	Applicative structure & Definition & It is a\dots  & It is not a\dots  & Proposition \\ \hline \hline
	$\plalamc$ & Example \ref{examplac} & $\bikuro$-algebra & $\biikuro$-algebra & \ref{propsepa3}  \\ \hline
	$\plalamcd$ & Example \ref{examplacd} & $\biikuro$-algebra & $\biilp$-algebra & \ref{propsepa4}  \\ \hline
	$\comlam$ & Definition \ref{examcom} & $\biikuro$-algebra & bi-$\bdi$-algebra & \ref{propsepa6}  \\ \hline
	$\plalam$ & Example \ref{exampla} & $\biikuro$-algebra & bi-$\bdi$-algebra & \ref{propsepa2}  \\ \hline
	$\fbiilp$ & Example \ref{examfbiilp} & $\biilp$-algebra & bi-$\bdi$-algebra & \ref{propsepa5} \\ \hline
	$\plalamb$ & Example \ref{exambipla} & bi-$\bdi$-algebra & $\bci$-algebra & \ref{propsepa1}  \\ \hline \hline
	\multicolumn{5}{|c|}{Inclusions} \\ \hline
	\multicolumn{5}{|c|}{$\bci$-algebras $\subsetneq$ bi-$\bdi$-algebras $\subsetneq$ $\biilp$-algebras $\subsetneq$ $\biikuro$-algebras $\subsetneq$ $\bikuro$-algebras} \\ \hline
\end{tabular}
\label{tab:taiou}
\end{table}
\end{center}

%%%%%%%%%%%%%%%%%%%%%%%%%%%%%%%%%%%%%%%%%%%%%%%%%%%%%%%%%%%%%%%%%
\subsection{Proofs of separations} \label{secsepa1}

First we show that the planar lambda calculus with a constant separates $\bikuro$-algebras and $\biikuro$-algebras.

\begin{exa} \label{examplac}
Suppose a constant symbol $c$ and add the following constant rule to the construction rules of planar lambda terms (See Example \ref{examlam} and \ref{exampla}).
\begin{center}
	\AxiomC{}
	\RightLabel{\scriptsize (constant)}
	\UnaryInfC{$\vdash c$}
	\DisplayProof
\end{center}
We assume no additional reduction rules about the constant.
That is, for instance, $c (\lambda x.x) c$ has no redex.
Closed planar terms (which may contain $c$) modulo $\eqb$ form a $\bikuro$-algebra, which we call $\plalamc$.
\end{exa}

Even adding the constant $c$, the planar lambda calculus still has the properties of confluence and strongly normalizing.

\begin{prop} \label{propsepa3}
$\plalamc$ is a $\bikuro$-algebra but not a $\biikuro$-algebra.
Hence, \\
$\biikuro$-algebras $\subsetneq$ $\bikuro$-algebras.
\end{prop}

\begin{proof}
Assume that $\plalamc$ is a $\biikuro$-algebra.
That is, assume there exist terms $I$ and $\batuih$ in $\plalamc$ such that $I M \eqb M$ and $\batuih M I \eqb M$ for any term $M$ in $\plalamc$.
We take $I$ and $\batuih$ as $\beta$-normal terms w.l.o.g.
If $M \eqb N$ in $\plalamc$, the number of appearance of $c$ is equal between $M$ and $N$.
Thus, since $\batuih c I \eqb c$, $I$ and $\batuih$ cannot contain $c$.
\begin{itemize}
\item When $\batuih c$ is $\beta$-normal, $\batuih c I$ is also $\beta$-normal and obviously not equal to $c$.
This contradicts to the confluence of the planar lambda calculus (with constant $c$).
\item When $\batuih = \lambda u.J$ for some $J$ and $u$, $\batuih c I \eqb (J[c/u]) I$.
\begin{itemize}
\item When $J = \lambda v.J'$ for some $J'$ and $v$, $\batuih c I \eqb J'[c/u][I/v]$.
Suppose $v$ receives just $n$ arguments $N_1,\dots ,N_n$ $(n \geq 0)$ in $J'$.
$J' = C[v N_1 \dots  N_n]$ for some context $C[-]$ which contains $u$ to the left of the hole $[-]$.
For the $\beta$-normal form $N$ of $I N_1 \dots  N_n$, $\batuih c I \eqb (C[N])[c/u]$.
$(C[N])[c/u]$ is $\beta$-normal and obviously not equal to $c$.
This contradicts to the confluence.
\item Otherwise, $J[c/u] I$ is $\beta$-normal and not equal to $c$.
This contradicts to the confluence. \qedhere
\end{itemize}
\end{itemize}
\end{proof}

Next we show that the planar lambda calculus additionally employing the $\eta$-equality separates $\biikuro$-algebras and $\biilp$-algebras.

\begin{exa} \label{examplacd}
Suppose three constant symbols $c_1$, $c_2$ and $c_3$ and add the following constant rules ($i=1,2,3$) to the construction rules of planar lambda terms.
\begin{center}
	\AxiomC{}
	\RightLabel{\scriptsize (constant)}
	\UnaryInfC{$\vdash c_i$}
	\DisplayProof
\end{center}
We assume no additional reduction rules about the constants.
Closed planar terms (that may contain constants) modulo $\eqbe$ form a $\biikuro$-algebra, which we call $\plalamcd$.
Note that the equivalence relation of $\plalamcd$ is the $\beta \eta$-equality, while that of $\plalamc$ (Example \ref{examplac}) is the $\beta$-equality.
We have $\lambda xyz.x(yz)$ as a representation of $\batui$ in $\plalamcd$.
Indeed, for any term $M$, $(\lambda xyz.x(yz)) M (\lambda w.w) \eqb \lambda z.Mz =_{\eta} M$.
\end{exa}

\begin{prop} \label{propsepa4}
$\plalamcd$ is a $\biikuro$-algebra but not a $\biilp$-algebra.
Hence, \\
$\biilp$-algebras $\subsetneq$ $\biikuro$-algebras.
\end{prop}

\begin{proof}
Assume that there are some terms $L$ and $P$ in $\plalamcd$ satisfying that for any terms $M_1$, $M_2$ and $M_3$, $L M_1 (P M_2 M_3) \eqbe M_1 M_2 M_3$.
Taking $M_1 = M_2 = M_3 := \lambda x.x$, we see that $L$ and $P$ cannot contain constants.

Taking $M_i := c_i$, we have $L c_1 (P c_2 c_3) \eqbe c_1 c_2 c_3$.
Since $L$ is a closed planar term with no constants, the $\beta \eta$-normal form of $L$ is the form $\lambda xy_1 \dots  y_m.x N_1 \dots N_n$ ($m,n \geq 0$).
Therefore, $L c_1 (P c_2 c_3) \eqbe (\lambda y_1\dots y_m.c_1 N_1 \dots N_n)(P c_2 c_3)$.
However, this term cannot be $\beta \eta$-equal to $c_1 c_2 c_3$ since $c_1$ cannot receive $c_2$ and $c_3$ as separated arguments no matter how the form of $P$ is.
\end{proof}

Next we show that the freely constructed $\biilp$-algebra separates $\biilp$-algebras and bi-$\bdi$-algebras.

\begin{exa} \label{examfbiilp}
We take $\fbiilp$ as the freely constructed $\biilp$-algebra with two constants $c_1$ and $c_2$.
That is, elements of $\fbiilp$ are constructed from $\comb$, $\comi$, $\batui$, $\coml$, $\comp$, $c_1$ and $c_2$ using the application and the unary operation $\kuro{(\haih)}$.
The equality in $\fbiilp$ is obtained by the axioms of $\biilp$-algebras and we do not assume any axioms on the constants.
\end{exa}

\begin{prop} \label{propsepa5}
$\fbiilp$ is a $\biilp$-algebra but not a bi-$\bdi$-algebra.
Hence, \\
bi-$\bdi$-algebras $\subsetneq$ $\biilp$-algebras.
\end{prop}

\begin{proof}
Assume that $\fbiilp$ is a bi-$\bdi$-algebra and write the right and left applications as $\rapp$ and $\lapp$. Here this $\rapp$ is the same application as that of $\fbiilp$ as a $\biilp$-algebra, that is, $MN$ and $M \rapp N$ denote the same element.

By the combinatory completeness, there is an element $M := \rlam{x}{\rlam{y}{x \lapp y}}$ in $\fbiilp$.
Since $M \rapp \comi \rapp \dagl{\comi} = \comi$ holds, this $M$ cannot contain $c_1$ nor $c_2$.
For this $M$, $M \rapp c_1 \rapp \dagl{c_2} = c_2 \rapp c_1$.
As we can see from the axioms of $\comb$, $\comi$, $\batui$, $\coml$, $\comp$ and $\kuro{( \mathchar`- )}$, 
it is impossible for $M$ in any form to exchange the order of two arguments $c_1$ and $\dagl{c_2}$ in $M \rapp c_1 \rapp \dagl{c_2}$.
Then it is also impossible for $\dagl{c_2}$ in any form to reduce $M \rapp c_1 \rapp \dagl{c_2}$ to $c_2 \rapp c_1$.
\end{proof}

Finally we show the bi-planar lambda calculus $\plalamb$ (Example \ref{exambipla}) separates bi-$\bdi$-algebras and $\bci$-algebras.

\begin{prop} \label{propsepa1}
$\plalamb$ is a bi-$\bdi$-algebra but not a $\bci$-algebra.
Hence, \\
$\bci$-algebras $\subsetneq$ bi-$\bdi$-algebras.
\end{prop}

\begin{proof}
Assume that there is some closed bi-planar lambda term $C$ in $\plalamb$ such that for any closed bi-planar term $M$, $N$ and $L$, $C \rapp M \rapp N \rapp L \eqb M \rapp L \rapp N$.
Let $C'$ be the $\beta$-normal form of $C \rapp \rlaml{x}{x}$.
$C' \rapp M \rapp N \eqb N \rapp M$ holds for any $M$ and $N$.
Take $M := \llaml{x}{x} \rapp \llaml{y}{y}$ and $N := \llaml{x}{x} \rapp \llaml{y}{y} \rapp \llaml{z}{z}$.
Note that for any $\beta$-normal term $P$ and a free variable $w$ of $P$, $P[M/w]$ and $P[N/w]$ are $\beta$-normal.
\begin{itemize}
\item When $C' \rapp M$ is $\beta$-normal, both $C' \rapp M \rapp N$ and $N \rapp M$ are $\beta$-normal.
However, obviously $C' \rapp M \rapp N \neqb N \rapp M$ and it contradicts to the confluence of the bi-planar lambda calculus.
\item When $C' = \rlaml{u}{C''}$ for some $C''$ and $u$, $C' \rapp M \rapp N \eqb C''[M/u] \rapp N$.
\begin{itemize}
\item When $C'' = \rlaml{v}{C'''}$ for some $C'''$ and $v$, $C' \rapp M \rapp N \eqb C'''[M/u][N/v]$.
Since $v$ is the rightmost free variable of $C'''$, $N$ is to the right of $M$ in $C'''[M/u][N/v]$.
Hence $C'''[M/u][N/v] \neqb N \rapp M$ and it contradicts to the confluence.
\item Otherwise, $C''[M/u] \rapp N$ is $\beta$-normal.
$C''[M/u] \rapp N \neqb N \rapp M$ and it contradicts to the confluence. \qedhere
\end{itemize}
\end{itemize}
\end{proof}

%%%%%%%%%%%%%%%%%%%%%%%%%%%%%%%%%%%%%%%%%%%%%%%%%%%%%%%%%%%%%%%%%
\subsection{The planar lambda calculus is not a bi-$\bdi$-algebra} \label{secsepa2}

Proofs of separations in the previous subsection are straightforward ones.
However, it is sometimes difficult to show that an applicative structure does not belong to certain class of applicative structures.
In this subsection, as an example, we will show that $\plalam$ of Example \ref{exampla} (the planar lambda calculus with no constant) is not a bi-$\bdi$-algebra.
Compared to propositions when constants exist (Proposition \ref{propsepa3} and \ref{propsepa4}), the proof is more tricky.

\begin{lem} \label{lemplai}
For any term $M$ of $\plalam$, there is a term $N$ of $\plalam$ such that $NM \eqb \lambda x.x$.
\end{lem}

\begin{proof}
Since planar lambda terms always have $\beta$-normal forms uniquely, we can assume $M$ is $\beta$-normal w.l.o.g.
We show this lemma by the induction on the number of bound variables of $M$.
When $BV(M)$ is a singleton, $M$ is $\lambda x.x$ and $N:=\lambda x.x$ satisfies $NM \eqb \lambda x.x$.

Assuming that the lemma holds till the number of bound variables of $M$ is $k$, we will show that the lemma holds for $M$ which contains $k+1$ bound variables.
Since $M$ is planar and $\beta$-normal, $M = \lambda x y_1 \dots  y_m .x P_1 \dots  P_n$ for 
some $\beta$-normal planar terms $P_1, \dots , P_n$. Here $y_1 ,\dots , y_m$ are all the free variables of $P_1 ,\dots , P_n$.
Let $Q_j$ be the term replacing all the $y_i$ in $P_j$ with $\lambda z.z$.
Each $Q_j$ is a closed planar term and has at most $k$ bound variables.
Hence, from the induction hypothesis, there exists some closed planar term $R_j$ such that $R_j Q_j \eqb \lambda x.x$.
Take $N' := \lambda w_1 \dots  w_n.(R_1 w_1)\dots (R_n w_n)$ and $N:= \lambda u.uN'(\lambda z_1.z_1)\dots (\lambda z_m.z_m)$.
Then $N'$ and $N$ are closed planar terms and 
\begin{eqnarray*}
NM &\eqb& MN'(\lambda z_1.z_1)\dots (\lambda z_m.z_m) \\
&=& (\lambda x y_1 \dots  y_m .x P_1 \dots  P_n)N'(\lambda z_1.z_1)\dots (\lambda z_m.z_m) \\
&\eqb& N'Q_1 \dots  Q_n \\
&=& (\lambda w_1 \dots  w_n.(R_1 w_1)\dots (R_n w_n))Q_1 \dots  Q_n \\
&\eqb& (R_1 Q_1)\dots (R_n Q_n) \\
&\eqb& (\lambda x.x)\dots  (\lambda x.x) \\
&\eqb& \lambda x.x.
\end{eqnarray*}
\end{proof}

\begin{lem} \label{lembcijanai}
$\plalam$ is not a $\bci$-algebra.
\end{lem}

\begin{proof}
Assume that there is a term $T$ in $\plalam$ such that $TMN \eqb NM$ for any $M$ and $N$ in $\plalam$.
(Note that a total applicative structure containing $\comb$ and $\comi$ is a $\bci$-algebra iff it has $\comt$ such that $\comt x y = y x$.
Indeed, $\comb (\comb (\comt ( \comb \comb \comt))\comb)\comt$
satisfies the axiom of the $\comc$-combinator.)
Take a term $\lambda x y_1 \dots  y_m.x P_1 \dots  P_n$ as the $\beta$-normal form of $T$.
If $n=0$, $T=\lambda x.x$ and it immediately leads contradiction.
Thus $n \geq 1$.

Since $T MN \eqb NM$ for any $M$ and $N$,
\begin{center}
$TM  \eqb TMT \eqb TMTT \eqb TMTTT \eqb \dots $.
\end{center}
Let $Q_j$ $(j=1,\dots ,n)$ be the terms replacing all the $y_i$ in $P_j$ with $T$.
Each $Q_j$ is a closed planar term.
Let $U := \lambda x.x Q_1 \dots  Q_n$.
\begin{eqnarray*}
UM &=& (\lambda x.x Q_1 \dots  Q_n)M \\
&\eqb& M Q_1 \dots  Q_n \\
&=& (M P_1 \dots  P_n)[T/y_1]\dots [T/y_m] \\
&\eqb& (\lambda x y_1 \dots  y_m.x P_1 \dots  P_n) M T \dots  T \\
&=& TMT\dots T \\
&\eqb& TM.
\end{eqnarray*}
Thus $UMN \eqb (TM)N \eqb NM$ holds for any $M$ and $N$.

From Lemma \ref{lemplai}, there exist closed terms $R_j$ $(j=1,\dots ,n)$ such that $R_j Q_j \eqb \lambda z.z$.
Take $M_0 := \lambda w_1 \dots  w_n.(R_1 w_1)\dots (R_n w_n)$.
Then for any closed planar term $N$,
\begin{eqnarray*}
NM_0 &\eqb& UM_0 N \\
&=& (\lambda x.x Q_1 \dots  Q_n)M_0 N \\
&\eqb& M_0 Q_1 \dots  Q_n N \\
&=& (\lambda w_1 \dots  w_n.(R_1 w_1)\dots (R_n w_n))Q_1 \dots  Q_n N \\
&\eqb& (R_1 Q_1) \dots  (R_n Q_n) N \\
&\eqb& (\lambda z.z)\dots (\lambda z.z) N \\
&\eqb& N.
\end{eqnarray*}

Taking $N_0 := \lambda x.x$ in $N_0 M_0 \eqb N_0$, we get $M_0 = \lambda x.x$.
Therefore, $N (\lambda x.x) \eqb N$ holds for any closed planar term $N$.
However, $N:= \lambda y.y(\lambda z.z)$ is the counterexample of this equation and it leads contradiction.
\end{proof}

\begin{prop} \label{propsepa2}
$\plalam$ is a  not a bi-$\bdi$-algebra.
\end{prop}

\begin{proof}
Assume that $\plalam$ is a bi-$\bdi$-algebra.
That is, taking $\rapp$ as the application canonically obtained by the application of planar lambda terms, assume that there is some binary operation $\lapp$ such that $(|\plalam|,\rapp,\lapp)$ becomes a bi-$\bdi$-algebra.
This $\lapp$ is the binary operation not on planar lambda terms, but on $\beta$-equivalence classes  of planar lambda terms.
However, in the sequel, we denote a lambda term $M$ indistinguishably to the equivalence class containing $M$.
For instance, for planar lambda terms $M_1$ and $M_2$, $M_1 \lapp M_2$ denotes some representation of $\overline{M_1} \lapp \overline{M_2}$, where $\overline{M_i}$ is the $\beta$-equivalence class containing $M_i$.

By the combinatory completeness for bi-$\bdi$-algebras, there is a closed planar term
$L$ representing $\rlam{x}{\rlam{y}{x \lapp y}}$.
Take a term $\lambda x y_1 \dots  y_m.x P_1 \dots  P_n$ as the $\beta$-normal form of $L$.
For a term $T$ representing $\rlam{x}{\rlam{y}{x \lapp (y \rapp \comir)}}$, dividing to the cases of $n=0$ or not, we will show that $T$ makes $\plalam$ a $\bci$-algebra and leads contradiction to Lemma \ref{lembcijanai}.

If $n=0$, $L=\lambda x.x$ and $M \rapp N \eqb (L \rapp M) \rapp N \eqb M \lapp N$ holds for any $M$ and $N$ in $\plalam$.
Given arbitrary term $N_0$ in $\plalam$, take $M := \comir$ and $N:= \dagl{N_0}$ in $M \rapp N \eqb M \lapp N$.
Then we get 
$\dagl{N_0} \eqb
N_0 \rapp \comir$.
For arbitrary $M_0$ and $N_0$ in $\plalam$,
\begin{eqnarray*}
T \rapp M_0 \rapp N_0 &=& (\rlam{x}{\rlam{y}{x \lapp (y \rapp \comir)}}) \rapp M_0 \rapp N_0 \\
&\eqb& M_0 \lapp (N_0 \rapp \comir) \\
&\eqb& M_0 \lapp \dagl{N_0} \\
&\eqb& N_0 \rapp M_0 
\end{eqnarray*}
holds.
Hence, $T$ makes $\plalam$ a $\bci$-algebra and contradicts to Lemma \ref{lembcijanai}.

Next is the case of $n \geq 1$.
Since $L \rapp M \rapp N \eqb M \lapp N$ for any $M$ and $N$,
\[ L \rapp M \eqb L \rapp M \rapp (\dagl{L}) \eqb L \rapp M \rapp (\dagl{L}) \rapp (\dagl{L}) \eqb L \rapp M \rapp (\dagl{L}) \rapp (\dagl{L}) \rapp (\dagl{L}) \eqb\dots. \]
Let $Q_j$ be the term replacing all the $y_i$ in $P_j$ with $\dagl{L}$.
Each $Q_j$ is a closed planar term.
Let \\ $V:= \lambda x.x Q_1 \dots  Q_n$.
\begin{eqnarray*}
VM &=& (\lambda x.x Q_1 \dots  Q_n)M \\
&\eqb& M Q_1 \dots  Q_n \\
&=& (M P_1 \dots  P_n)[\dagl{L}/y_1]\dots [\dagl{L}/y_m] \\
&\eqb& (\lambda x y_1 \dots  y_m.x P_1 \dots  P_n)M(\dagl{L})\dots (\dagl{L}) \\
&=& LM(\dagl{L})\dots (\dagl{L}) \\
&\eqb& LM.
\end{eqnarray*}
Thus $VMN \eqb LMN \eqb M \lapp N$ holds for any $M$ and $N$.

From Lemma \ref{lemplai}, there exists closed term $R_j$ $(j=1,\dots ,n)$ such that $R_j Q_j \eqb \lambda z.z$.
Take $M_1 := \lambda w_1 \dots  w_n.(R_1 w_1)\dots (R_n w_n)$.
Then for any closed planar term $N$,
\begin{eqnarray*}
M_1 \lapp N &\eqb& LM_1 N \\
&=& (\lambda x y_1 \dots  y_m.x P_1 \dots  P_n)M_1 N \\
&\eqb& M_1 Q_1 \dots  Q_n N \\
&=& (\lambda w_1 \dots  w_n.(R_1 w_1)\dots (R_n w_n))Q_1 \dots  Q_n N \\
&\eqb& (R_1 Q_1) \dots  (R_n Q_n) N \\
&\eqb& (\lambda z.z)\dots (\lambda z.z) N \\
&\eqb& N.
\end{eqnarray*}

Taking $N := \comil$ in $M_1 \lapp N \eqb N$, we get $M_1 = \comil$.
Therefore, $\comil \lapp N_1 \eqb N_1$ holds for any closed planar term $N_1$.
Given arbitrary $N_2$ in $\plalam$, with $N_1:= \dagl{N_2}$, we get 
\begin{eqnarray*}
\dagl{N_2} &\eqb& \comil \lapp \dagl{N_2} \\
&\eqb& N_2 \rapp \comil.
\end{eqnarray*}
For arbitrary $M_2$ and $N_2$ in $\plalam$,
\begin{eqnarray*}
T \rapp M_2 \rapp  N_2 &=& \rlam{x}{\rlam{y}{x \lapp (y \rapp \comil)}} \rapp M_2 \rapp N_2 \\
&\eqb& M_2 \lapp (N_2 \rapp \comil) \\
&\eqb& M_2 \lapp \dagl{N_2} \\
&\eqb& N_2 \rapp  M_2
\end{eqnarray*}
holds. Hence, $T$ makes $\plalam$ a $\bci$-algebra and  contradicts to Lemma \ref{lembcijanai}.
\end{proof}

We have already seen in Proposition \ref{propsepa4} that $\plalamcd$ (the planar lambda calculus with constants) is not a $\biilp$-algebra.
However, whether $\plalam$ is a $\biilp$-algebra is still open.

%%%%%%%%%%%%%%%%%%%%%%%%%%%%%%%%%%%%%%%%%%%%%%%%%%%%%%%%%%%%%%%%%
\subsection{The computational lambda calculus} \label{seccom}

Next we consider the computational lambda calculus as an applicative structure that gives rise to non-symmetric structures.
The computational lambda calculus is a variant of the lambda calculus whose evaluation rules are sound for programs with computational effects \cite{comlambda}. The following axiomatization is from \cite{sabry}.

\begin{defi} \label{examcom}
Suppose infinite supply of variables $x,y,z,\dots $.
{\it Values}, {\it terms} and {\it evaluation contexts} are defined as follows:
\begin{itemize}
\item {(values)} \; $V ::= x \; | \; \lambda x.M$
\item {(terms)} \; $M ::= V \; | \; MM'$
\item {(evaluation contexts)} \; $E[\haih] ::= [\haih] \; | \; EM \; | \; VE$
\end{itemize}
(Terms are the same ones of the ordinary lambda calculus in Example \ref{examlam}.)

An equivalence relation $=_c$ on terms is defined as the congruence of the following equations:
\begin{itemize}
\item ($\beta_V$) \; $(\lambda x.M)V =_c M[V/x]$
\item ($\eta_V$) \; $\lambda x.Vx =_c V$
\item ($\beta_{\Omega}$) \; $(\lambda x.E[x])M =_c E[M]$
\end{itemize}
Here $E[M]$ denotes the term obtained by substituting $M$ for $[\haih]$ in $E[\haih]$.

The {\it (untyped) computational lambda calculus} is the lambda calculus formed by terms and $=_c$.
\end{defi}

In \cite{tomita1}, we showed that the computational lambda calculus is a $\biikuro$-algebra but not a $\bci$-algebra.
We can get a $\biikuro$-algebra $\comlam$, whose underlying set is equivalence classes of lambda terms modulo $=_c$. (Note that terms of $\comlam$ are not restricted to closed terms.)
Here $\lambda xyz.x(yz)$, $\lambda x.x$, $\lambda xy.yx$ and $\lambda x.xM$ are representatives of $\comb$, $\comi$, $\batui$ and $\kuro{M}$ respectively.

Although the computational lambda calculus has all terms of the lambda calculus, $\comlam$ is not a PCA nor a $\bci$-algebra.
This is reasonable considering that programs with effects cannot be discarded, duplicated nor exchanged in general, and thus $\comlam$ cannot have the $\comsa$/$\comk$/$\comc$-combinator.
Moreover, we can prove the next proposition.

\begin{prop} \label{propsepa6}
$\comlam$ is not a bi-$\bdi$-algebra.
\end{prop}

To prove this proposition, we use the {\it CPS-translation} \cite{sabrycps}.
The CPS-translation $\cpst{\haih}$ sends terms of the computational lambda terms to terms of the ordinary lambda calculus and is defined inductively as follows.

\begin{itemize}
\item $\cpst{x} := \lambda k.kx$
\item $\cpst{\lambda x.M} := \lambda k.k(\lambda x.\cpst{M})$
\item $\cpst{MN} := \lambda k.\cpst{M} (\lambda f. \cpst{N} (\lambda x.fxk))$
\end{itemize}

For any term $M$ and $N$, $M =_c N$ holds in the computational lambda calculus iff $\cpst{M} \eqbe \cpst{N}$ holds in the ordinary lambda calculus.

\begin{proof}[Proof of Proposition \ref{propsepa6}]
We will lead a contradiction by assuming $\comlam$ is a bi-$\bdi$-algebra.
If $\comlam$ is a bi-$\bdi$-algebra, we have a term $L$ representing $\rlam{x}{\rlam{y}{x \lapp y}}$ and a term $\dagl{M}$ representing $\llam{x}{M \rapp x}$ for each term $M$.
For any terms $M_1$ and $M_2$, $L M_1 (\dagl{M_2}) =_c M_2 M_1$ holds, and thus $\cpst{L M_1 (\dagl{M_2})} \eqbe \cpst{M_2 M_1}$ holds.
Now we take a fresh variables $v$ and let $M_2 := vv$.
Additionally we take a fresh variable (fresh for $L$, $M_2$ and $\dagl{M_2}$) $u$ and let $M_1 := uu$.
Then

\begin{eqnarray*}
\cpst{L M_1 (\dagl{M_2})} &=& \lambda k. (\lambda k'.\cpst{L} (\lambda f'.\cpst{M_1}(\lambda x'.f'x'k'))) (\lambda f.\cpst{\dagl{M_2}}(\lambda x.fxk)) \\
&\eqbe& \lambda k. \cpst{L} (\lambda f'.\cpst{M_1} (\lambda x' .f'x'(\lambda f.\cpst{\dagl{M_2}}(\lambda x.fxk)))) \\
&\eqbe& \lambda k. \cpst{L} (\lambda f'.uu (\lambda x' .f'x'(\lambda f.\cpst{\dagl{M_2}}(\lambda x.fxk)))), \\
\cpst{M_2 M_1} &\eqbe& \lambda k.vv(\lambda f.uu(\lambda x.fxk)).
\end{eqnarray*}
In $\cpst{M_2 M_1}$, $vv$ receives the argument of the form $(\dots  uu \dots )$.
However, since $u$ and $v$ are fresh, no matter what $\cpst{L}$ is, in $\cpst{L M_1 (\dagl{M_2})}$, $vv$ cannot receive arguments containing $uu$.
Hence these terms $\cpst{L M_1 (\dagl{M_2})}$ and $\cpst{M_2 M_1}$ cannot be $\beta \eta$-equal.
It leads a contradiction to the soundness of the CPS-translation.
\end{proof}

Semantically, the untyped ordinary/linear/planar lambda calculus is modeled by a reflexive object of a CCC/SMCC/closed multicategory.
And it is related to the categorical structures of assemblies on each lambda calculus.
On the other hand, the untyped computational lambda calculus is modeled by a reflexive object of a Kleisli category.
Since the categorical structure of a Kleisli category is not monoidal in general but premonoidal (See \cite{power}), it is expected that the category of assemblies on the untyped computational lambda calculus is not a monoidal category.
Thus the computational lambda calculus $\comlam$ is expected not to be a $\biilp$-algebra inducing monoidal closed category, however, we have not proven this conjecture yet.

Here, we give an intuitive explanation for the conjecture.
Assume that $\coml$ and $\comp$ exist in the computational lambda calculus.
Take three non-values $M_1$, $M_2$ and $M_3$.
Suppose these terms are reduced to values: $\coml \migi v_L$; $\comp \migi v_P$; $M_i \migi v_i$.
In $M_1 M_2 M_3$, the evaluation proceeds as follows:
\begin{tabbing}
	\hspace{150pt} \= \ \ \ \ \ \ \= $M_1$ is reduced to $v_1$ \\
	\> $\rightsquigarrow$ \> $M_2$ is reduced to $v_2$ \\
	\> $\rightsquigarrow$ \> $v_1 v_2$ is reduced \\
	\> $\rightsquigarrow$ \> $M_3$ is reduced to $v_3$ \\
	\> $\rightsquigarrow$ \> \dots 
\end{tabbing}
On the other hand, in $\coml M_1 (\comp M_2 M_3)$, the evaluation proceeds as follows:
\begin{tabbing}
	\hspace{150pt} \= \ \ \ \ \ \ \= $\coml$ is reduced to $v_L$ \\
	\> $\rightsquigarrow$ \> $M_1$ is reduced to $v_1$ \\
	\> $\rightsquigarrow$ \> $v_L v_1$ is reduced \\
	\> $\rightsquigarrow$ \> $\comp$ is reduced to $v_P$ \\
	\> $\rightsquigarrow$ \> $M_2$ is reduced to $v_2$ \\
	\> $\rightsquigarrow$ \> $v_P v_2$ is reduced \\
	\> $\rightsquigarrow$ \> $M_3$ is reduced to $v_3$ \\
	\> $\rightsquigarrow$ \> \dots 
\end{tabbing}
These two computations seem not to coincide, since the order of the evaluations of $v_1 v_2$ and $M_3$ is reversed.

%%%%%%%%%%%%%%%%%%%%%%%%%%%%%%%%%%%%%%%%%%%%%%%%%%%%%%%%%%%%%%%%%
%%%%%%%%%%%%%%%%%%%%%%%%%%%%%%%%%%%%%%%%%%%%%%%%%%%%%%%%%%%%%%%%%
\section{Necessary conditions for inducing closed structures} \label{secnec}

We have seen that applicative structures of certain classes induce the corresponding categorical structures, in Proposition \ref{propccc} (CCCs), Proposition \ref{propsmcc} (SMCCs), Proposition \ref{propmul} (closed multicategories), Proposition \ref{propclo} (closed categories), Proposition \ref{propmonoclo} (monoidal closed categories) and Proposition \ref{propbiclo} (monoidal bi-closed categories).
In this section, we show the certain ``inverses'' of these propositions hold.

\begin{prop} \label{propnecccc}
Suppose $\hka$ is a total applicative structure and $\catc := \asmca$ happens to be a CCC.
$\hka$ is an $\sk$-algebra if the following conditions hold.
\begin{enumerate}[(i)]
\item $|Y^X| = \homrm_{\catc} (X,Y)$ and $\rlz{f}_{Y^X} = \{ r \mid \mbox{$r$ realizes $f$} \}$. \label{conhom1}
\item For $f:X' \migi X$ and $g:Y\migi Y'$, $g^f : Y^X \migi Y'^{X'}$ is the function sending $h:X\migi Y$ to $g \circ h \circ f$. \label{conhom2}
\item The forgetful functor from $\catc$ to $\cats$ strictly preserves finite products. \label{confor}
\item The adjunction $\Phi: \homrm_{\catc} (X \times Y, Z) \migi \homrm_{\catc} (X, Z^Y)$ is the function sending a function $f$ to the function $x \mapsto (y \mapsto f(x,y))$. \label{conadj}
\end{enumerate}
\end{prop}

\begin{proof}
Take an object $A := (\uhka,\erlz_A)$, where $\rlz{a}_A := \{ a \}$.
When we take $\comi$ as a realizer of $id_A$, this $\comi$ satisfies $\forall a \in \uhka$, $\comi \rlz{a}_A \subseteq \rlz{id_A (a)}_A$.
That is, $\forall a \in \uhka$, $\comi a = a$.

Applying $\Phi$ to the first projection $(a,a') \mapsto a: A \times A \migi A$,
we get a map $k:A \migi A^A$, which sends $a$ to $(a' \mapsto a)$.
(Here we use the conditions \ref{conhom1}, \ref{conadj} and \ref{confor} to clarify what the function $k$ actually is.)
When we take $\comk$ as a realizer of $k$, this $\comk$ satisfies $\forall a,a' \in \uhka$, $\comk a a' = a$.

Let $\phi : A \migi A^A$ be the function sending $a$ to the function $x \mapsto a x$.
Here $\phi(a)$ is realized by $a$ and $\phi$ is realized by $\comi$.

Applying $\Phi$ twice to the map from $((A^A)^A \times A^A) \times A$ to $A$ defined as
\[ ((A^A)^A \times A^A) \times A 
\xmigi{id \times {\rm diagonal}} ((A^A)^A \times A^A) \times (A \times A)
 \]
\[ \xmigi{{\rm symmetry}} ((A^A)^A \times A) \times (A^A \times A) \xmigi{ev \times ev} A^A \times A \xmigi{ev} A, \]
we get a map $s:(A^A)^A \migi (A^A)^{(A^A)}$ which sends a function $g:A \migi A^A$ to the function
$(f:A \migi A) \mapsto (a \mapsto g(a) (f(a)))$.
The map 
\[ A \xmigi{\phi} A^A \xmigi{{\phi}^{id}} (A^A)^A \xmigi{s} (A^A)^{(A^A)} \xmigi{id^{\phi}} (A^A)^A \]
is the function
$a \mapsto (a' \mapsto (a'' \mapsto a a'' (a' a'')))$.
(Here we use the conditions \ref{conhom2} to clarify what the functions ${\phi}^{id}$ and ${id^{\phi}}$ actually are.)
Thus, when we take $\comsa$ as a realizer of this map, $\comsa$ satisfies $\comsa x y z = x z (y z)$ for any $x,y,z \in \uhka$.
\end{proof}

To rephrase the proposition, to obtain a CCC by categorical realizability, being an $\sk$-algebra is the necessary condition on the total applicative structure (under several conditions).
We will show the similar propositions for the other classes.
Combining the propositions in this section and the separations in the previous section, we can say that, for instance, the category of assemblies on an applicative structure that is a bi-$\bdi$-algebra  but not a $\bci$-algebra (\eg, $\asmc{\plalamb}$) is indeed non-symmetric monoidal (as long as we try to take the symmetry in the canonical way).

\begin{rem} \label{remnecpar}
When we try to prove the proposition replacing ``total applicative structure'' with ``partial applicative structure'' in Proposition \ref{propnecccc}, we cannot use the same proof.
This is because $\phi : A \migi A^A$ is not always defined.
Indeed, when $a a'$ is not defined in $\hka$, $\phi (a)$ is not defined at $a'$.
It is still unclear whether we can prove the similar proposition as Proposition \ref{propnecccc} when $\hka$ is a partial applicative structure.
\end{rem}

\begin{prop} \label{propnecsmcc}
Suppose $\hka$ is a total applicative structure and $\catc := \asmca$ happens to be an SMCC.
$\hka$ is a $\bci$-algebra if the following conditions hold.
\begin{enumerate}[(i)]
\item $|Y \rimp X| = \homrm_{\catc} (X,Y)$ and $\rlz{f}_{Y \rimps X} = \{ r \mid \mbox{$r$ realizes $f$} \}$.
\item $g \rimp f : (Y \rimp X) \migi (Y' \rimp X')$ is the function sending $h:X \migi Y$ to $g \circ h \circ f$.
\item The forgetful functor from $\catc$ to $\cats$ is a strict symmetric monoidal functor.
\item The adjunction $\Phi: \homrm_{\catc} (X \otimes Y, Z) \migi \homrm_{\catc} (X, Z \rimp Y)$ is the function sending a function $f$ to the function $x \mapsto (y \mapsto f(x,y))$.
\end{enumerate}
\end{prop}

\begin{proof}
Take an object $A := (\uhka,\erlz_A)$, where $\rlz{a}_A := \{ a \}$.
When we take $\comi$ as a realizer of $id_A$, this $\comi$ satisfies $\forall a \in \uhka$, $\comi \rlz{a}_A \subseteq \rlz{id_A (a)}_A$.
That is, $\forall a \in \uhka$, $\comi a = a$.

Let $\phi : A \migi (A \rimp A)$ be the function sending $a$ to the function $x \mapsto a x$.
Here $\phi(a)$ is realized by $a$ and $\phi$ is realized by $\comi$.

Applying $\Phi$ twice to the map
\[ ((A \rimp A) \otimes (A \rimp A)) \otimes A 
\xrightarrow{{\rm associator}} (A \rimp A) \otimes ((A \rimp A) \otimes A)
\xrightarrow{id \otimes ev} (A \rimp A) \otimes A
\xrightarrow{ev} A, \]
we get a map $l: (A \rimp A) \migi ((A \rimp A) \rimp (A \rimp A))$, which sends $g:A \migi A$ to the function \\ $(f:A\migi A) \mapsto g \circ f$.
The map 
\[ A \xmigi{\phi} (A \rimp A) \xmigi{l} ((A \rimp A) \rimp (A \rimp A)) \xmigi{id \rimps \phi} ((A \rimp A) \rimp A) \]
is the function $a \mapsto (a' \mapsto (a'' \mapsto a (a' a'')))$.
Thus, when we take $\comb$ as a realizer of this map, $\comb$ satisfies $\comb x y z = x (y z)$ for any $x,y,z \in \uhka$.

Applying $\Phi$ to the map
\[ A \otimes (A \rimp A) \xmigi{{\rm symmetry}} (A \rimp A) \otimes A \xmigi{ev} A, \]
we get a map $c:A \migi (A \rimp (A \rimp A))$, which sends $a$ to $(f \mapsto f(a))$.
The map 
\[ A \xmigi{c} (A \rimp (A \rimp A)) \xmigi{id \rimps \phi} (A \rimp A) \]
is the function $a \mapsto (a' \mapsto a' a)$.
Thus, when we take $\comt$ as a realizer of this map, $\comt$ satisfies $\comt x y = y x$ for any $x,y \in \uhka$.
Let $\comc := \comb (\comb (\comt ( \comb \comb \comt))\comb)\comt$. Then $\comc xyz=xzy$ holds for any
$x,y,z \in \uhka$.
\end{proof}

\begin{prop} \label{propnecmul}
Suppose $\hka$ is a total applicative structure and $\catc := \asmca$ happens to be a closed multicategory.
$\hka$ is a $\bikuro$-algebra if the following conditions hold.
\begin{enumerate}[(i)]
\item $\catc (;X) = |X|$ and $\ucatc(;X) = X$.
\item $\catc (X;Y) = \homrm_{\catc} (X,Y)$ and $\ucatc(X;Y) =(\homrm_{\catc} (X,Y), \erlz)$. \\
Here 
$\rlz{f}= \{ r \mid \mbox{$r$ realizes $f$} \}$.
\item $\ucatc (X_1,\dots ,X_{n};Y) = \ucatc (X_1;\ucatc (X_2,\dots ,X_n;Y))$ and $\catc (X_1,\dots ,X_n;Y)$ is the underlying set of $\ucatc (X_1,\dots ,X_n;Y)$.

\item For $g:Y_1,\dots ,Y_n \migi Z$ and $f_l:X^l_1,\dots ,X^l_{k_l} \migi Y_l$, $g \circ (f_1,\dots ,f_n)$ is the function sending $x^1_1,\dots ,x^1_{k_1},\dots ,x^n_{k_n}$ to $g(f_1 (x^1_1,\dots ,x^1_{k_1}),\dots ,f_n (x^n_1,\dots ,x^n_{k_n}))$.
When $k_l = 0$ for some $1 \leq l \leq n$, $g \circ (f_1,\dots ,f_n)$ is the function given $y_l \in |Y_l|$ pointed by $f_l$ as the $l$-th argument of $g$.
\item $ev_{X_1,\dots ,X_n;Y}$ sends $f, x_1,\dots ,x_n$ to $f(x_1,\dots ,x_n)$.
\item $\Lambda_{Z_1,\dots ,Z_m;X_1,\dots ,X_n;Y}$ sends a function $(z_1,\dots ,z_m,x_1,\dots ,x_n \mapsto f(z_1,\dots ,z_m,x_1,\dots ,x_n))$ to the function $(z_1,\dots ,z_m \mapsto f(z_1,\dots ,z_m,-,\dots ,-))$.
\end{enumerate}
\end{prop}

\begin{proof}
Take an object $A := (\uhka,\erlz_A)$, where $\rlz{a}_A := \{ a \}$.
When we take $\comi$ as a realizer of $id_A$, this $\comi$ satisfies $\forall a \in \uhka$, $\comi \rlz{a}_A \subseteq \rlz{id_A (a)}_A$.
That is, $\forall a \in \uhka$, $\comi a = a$.

Let $\phi : A \migi \ucatc (A;A)$ be the function sending $a$ to the map $x \mapsto a x$.
Here $\phi(a)$ is realized by $a$ and $\phi$ is realized by $\comi$.
Take a map 
\[ b:A,A,A \xmigi{id,\phi,id} A,\ucatc (A;A),A \xmigi{\phi,ev} \ucatc (A;A),A \xmigi{ev} A,  \]
which sends $(x,y,z)$ to $x (y z)$ for any $x,y,z \in \uhka$.
When we take $\comb$ as a realizer of $\Lambda_{A;A;\ucatc(A;A)} (\Lambda_{A,A;A;A} (b))$,
\begin{eqnarray*}
\comb x y z &=& (\Lambda_{A;A;\ucatc(A;A)} (\Lambda_{A,A;A;A} (b)))(x)(y)(z) \\
&=& b(x,y,z) \\
&=& x (y z).
\end{eqnarray*}

Given arbitrary $a \in \uhka$, take a map $f_{a}:A \migi A$ as 
\[ A \xmigi{id,a} A,A  \xmigi{\phi,id} \ucatc (A;A),A \xmigi{ev} A, \]
which sends $x \in \uhka$ to $x a$.
When we take $\kuro{a}$ as a realizer of $f_{a}$, $\kuro{a} x = x a$ for any $x \in \uhka$.
\end{proof}

\begin{prop} \label{propnecclo}
Suppose $\hka$ is a total applicative structure and $\catc := \asmca$ happens to be a closed category.
$\hka$ is a $\biikuro$-algebra if the following conditions hold.
\begin{enumerate}[(i)]
\item $|Y \rimp X| = \homrm_{\catc} (X,Y)$ and $\rlz{f}_{Y \rimps X} = \{ r \mid \mbox{$r$ realizes $f$} \}$. \label{conclo1}
\item $g \rimp f : (Y \rimp X) \migi (Y' \rimp X')$ is the function sending $h:X \migi Y$ to $g \circ h \circ f$.
\item $i_X$ is the function sending a function $(f:\ast \mapsto x)$ to $x$. \label{conclo3}
\item $L_{Y,Z}^X$ is the function sending $g:Y \migi Z$ to the function $(f:X \migi Y) \mapsto g \circ f$.
\end{enumerate}
\end{prop}

\begin{rem} \label{remsingle}
In the condition \ref{conclo3}, we assume that the unit object is a singleton $\{ \ast \}$.
The assumption can be derived from the condition \ref{conclo1}.

Take an object $X := (\{ x_1,x_2 \} , \erlz_X)$ by $\rlz{x_i}_X := \uhka$.
From the condition \ref{conclo1},
$|X \rimp I|$ is $\homrm_{\catc} (I,X)$.
Since $\homrm_{\catc} (I,X) = \homrm_{\cats} (|I|,\{ x_1,x_2 \})$, $|X \rimp I| = \homrm_{\cats} (|I|,\{ x_1,x_2 \})$.
Also since $X \rimp I \cong X$, $|X \rimp I| \cong |X| = \{ x_1,x_2 \}$.
$\homrm_{\cats} (|I|,\{ x_1,x_2 \}) \cong \{ x_1,x_2 \}$ holds iff $|I|$ is the singleton.
\end{rem}

\begin{proof}[Proof of Proposition \ref{propnecclo}]
Take an object $A := (\uhka,\erlz_A)$, where $\rlz{a}_A := \{ a \}$.
When we take $\comi$ as a realizer of $id_A$, this $\comi$ satisfies $\forall a \in \uhka$, $\comi \rlz{a}_A \subseteq \rlz{id_A (a)}_A$.
That is, $\forall a \in \uhka$, $\comi a = a$.

Let $\phi : A \migi (A \rimp A)$ be the function sending $a$ to the function $x \mapsto a x$.
Here $\phi(a)$ is realized by $a$ and $\phi$ is realized by $\comi$.
The map 
\[ A \xmigi{\phi} (A \rimp A) \xmigi{L} ((A \rimp A) \rimp (A \rimp A)) \xmigi{id \rimps \phi} ((A \rimp A) \rimp A) \]
is the function $a \mapsto (a' \mapsto (a'' \mapsto a (a' a'')))$.
Thus, when we take $\comb$ as a realizer of this map, $\comb$ satisfies $\comb x y z = x (y z)$ for any $x,y,z \in \uhka$.

Since $I \cong (I \rimp I)$ and $\comi \in \rlz{id_I}_{I \rimps I}$, we can assume $\comi \in \rlz{\ast}_I$ w.l.o.g.
When we take $\batui$ as a realizer of $i_A^{-1}:A \migi (A \rimp I)$, $\batui$ satisfies $\batui a x = a$ for any $a \in \uhka$ and $x \in \rlz{\ast}_I$, especially, $\batui a \comi =a$ holds.

Given arbitrary $a \in \uhka$, let $g_{a}:I \migi A$ be the function $\ast \mapsto a$.
$g_{a}$ is realized by $\batui a$.
The map 
\[ A \xmigi{\phi} (A \rimp A) \xmigi{id \rimps g_{a}} (A \rimp I) \xmigi{i_A} A \]
is the function $a' \mapsto a' a$.
Thus, when we take $\kuro{a}$ as a realizer of this map, $\kuro{a}$ satisfies $\kuro{a} x = x a$ for any $x \in \uhka$.
\end{proof}

\begin{prop} \label{propnec5}
Suppose $\hka$ is a total applicative structure and $\catc := \asmca$ happens to be a monoidal closed category.
$\hka$ is a $\biilp$-algebra if the following conditions hold.
\begin{enumerate}[(i)]
\item $|Y \rimp X| = \homrm_{\catc} (X,Y)$ and $\rlz{f}_{Y \rimps X} = \{ r \mid \mbox{$r$ realizes $f$} \}$.
\item $g \rimp f : (Y \rimp X) \migi (Y' \rimp X')$ is the function sending $h:X \migi Y$ to $g \circ h \circ f$.
\item The forgetful functor from $\catc$ to $\cats$ is a strict monoidal functor.
\item The adjunction $\Phi: \homrm_{\catc} (X \otimes Y, Z) \migi \homrm_{\catc} (X, Z \rimp Y)$ is the function sending a function $f$ to the function $x \mapsto (y \mapsto f(x,y))$.
\end{enumerate}
\end{prop}

\begin{proof}
Applying $\Phi$ twice to the map
\[ ((Y \rimp X) \otimes (X \rimp Z)) \otimes Z \xrightarrow{{\rm associator}}
(Y \rimp X) \otimes ((X \rimp Z) \otimes Z) \xrightarrow{id \otimes ev}
(Y \rimp X) \otimes X \xrightarrow{ev} Y, \]
we get a map $L^X_{Y,Z}: (Y \rimp X) \migi ((Y \rimp Z) \rimp (X \rimp Z))$.
This $L$ is the natural transformation $L$ of the closed category $\catc$.
Applying $\Phi$ to the unitor $\rho_X :X \otimes I \migi X$, we get a map $i^{-1}_X : X \migi (X \rimp I)$.
The inverse map is the natural isomorphism $i$ of the closed category $\catc$.
We can easily check that $\hka$ and $\catc$ satisfies all the conditions of Proposition \ref{propnecclo} for these $L$ and $i$.
Hence, $\hka$ is a $\biikuro$-algebra.

Take an object $A := (\uhka,\erlz_A)$, where $\rlz{a}_A := \{ a \}$.
Let $\phi: A \migi (A \rimp A)$ be the function sending $a$ to the function $x \mapsto ax$.
Here $\phi (a)$ is realized by $a$ and $\phi$ is realized by $\comi$.

Let $l: A \migi (A \rimp (A \otimes A))$ be the map obtained by applying $\Phi$ to
\[ A \otimes (A \otimes A) \xrightarrow{{\rm associator}} (A \otimes A) \otimes A \xrightarrow{(\phi \otimes id) \otimes id} ((A \rimp A) \otimes A) \otimes A \]
\[ \xrightarrow{ev \otimes id} A \otimes A \xrightarrow{\phi \otimes id} (A \rimp A) \otimes A \xrightarrow{ev} A, \]
and let $\coml$ be a realizer of $l$.
$l$ is the function sending $x$ to the function $(y,z) \mapsto xyz$.
Also let $\comp$ be a realizer of $p := \Phi(id_{A \otimes A}) : A \migi ((A \otimes A) \rimp A)$.
$p$ is the function sending $y$ to the function $z \mapsto (y,z)$.
Then for any $x,y,z \in \uhka$, $\coml x (\comp y z) \in \rlz{l(x)(p(y)(z))}_A$ and thus
\begin{eqnarray*}
\coml x (\comp y z) &=& l(x)(p(y)(z)) \\
&=& l(x)(y,z) \\
&=& xyz.
\end{eqnarray*}
\end{proof}

The proof of the next proposition, for monoidal bi-closed categories and bi-$\bdi$-algebras, is a little more complicated than the proofs of previous propositions.
When we obtain a monoidal bi-closed category $\asmca$ by a bi-$\bdi$-algebra $\hka$,
we take realizers of elements of the object $X \limp Y$ in $\asmca$ as
\[ \rlz{f}_{X \limp Y} := \{ r \in |\hka| \mid \mbox{$a \lapp r \in \rlz{f(x)}_Y$} \} \]
(See the proof of Proposition \ref{propbiclo}).
However, in the next proposition we do not assume anything about the left application of $\hka$, and thus we also cannot assume anything about realizers for $X \limp Y$.
This makes the proof of existence for $\combl$ and $\comdl$ cumbersome.

\begin{prop} \label{propnecbiclo}
Suppose $\hka = (\uhka, \rapp)$ is a total applicative structure and $\catc := \asmca$ happens to be a monoidal bi-closed category.
$\hka$ is a bi-$\bdi$-algebra if the following conditions hold.
\begin{enumerate}[(i)]
\item $|Y \rimp X| = \homrm_{\catc} (X,Y)$ and $\rlz{f}_{Y \rimps X} = \{ r \mid \mbox{$r$ realizes $f$} \}$.
\item $g \rimp f : (Y \rimp X) \migi (Y' \rimp X')$ is the function sending $h:X \migi Y$ to $g \circ h \circ f$.
\item The forgetful functor from $\catc$ to $\cats$ is a strict monoidal functor.
\item The adjunction $\Phi: \homrm_{\catc} (X \otimes Y, Z) \migi \homrm_{\catc} (X, Z \rimp Y)$ is the function sending a function $f$ to the function $x \mapsto (y \mapsto f(x,y))$.
\item $|X \limp Y| = \homrm_{\catc} (X,Y)$.
\item $f \limp g : (X \limp Y) \migi (X' \limp Y')$ is the function sending $h:X \migi Y$ to $g \circ h \circ f$.
\item The adjunction $\Phi': \homrm_{\catc} (X \otimes Y, Z) \migi \homrm_{\catc} (Y, X \limp Z)$ is the function sending a function $f$ to the function $y \mapsto (x \mapsto f(x,y))$.
\end{enumerate}
\end{prop}

\begin{proof}
The conditions of this proposition includes all the conditions of Proposition \ref{propnec5}.
Hence, $\hka$ is a $\biilp$-algebra and have the combinatory completeness for the planar lambda calculus.
We take $\combr$ as the $\comb$-combinator and $\comir$ as the $\comi$-combinator of $\hka$.

Take an object $A := (\uhka,\erlz_A)$, where $\rlz{a}_A := \{ a \}$.
Applying $\Phi$ to the evaluation map
$ev_1:A \otimes (A \limp A) \migi A$, we get a map $l:A \migi (A \rimp (A \limp A))$, which sends $a$ to $(f \mapsto f(a))$.
Let $\coml_1$ be a realizer of $l$ and $x \lapp y := \coml_1 \rapp x \rapp y$.
We will show that $(\uhka,\rapp,\lapp)$ is a bi-$\bdi$-algebra.

Let $\phi : A \migi (A \rimp A)$ be the function sending $a$ to the function $(x \mapsto a \rapp x)$.
Here $\phi(a)$ is realized by $a$ and $\phi$ is realized by $\comir$.

Given arbitrary $a \in \uhka$, let $\dagr{a} := \lamst x.\coml_1 \rapp x \rapp a$.
For any $x \in \uhka$, 
\begin{eqnarray*}
\dagr{a} \rapp x &=& \coml_1 \rapp x \rapp a \\
&=& x \lapp a.
\end{eqnarray*}

Given arbitrary $a \in \uhka$,
take $\dagl{a}$ as an element of $\rlz{\phi(a)}_{A \limp A}$.
Then for any $x \in \uhka$, 
\begin{eqnarray*}
x \lapp \dagl{a} &=& \coml_1 \rapp x \rapp \dagl{a} \\
&=& l(x)(\phi(a)) \\
&=& \phi (a) (x) \\
&=& a \rapp x.
\end{eqnarray*}
Furthermore, we can take $\comil$ as $\dagl{(\comir)}$.

Next we obtain $\combl$.
Applying $\Phi'$ to 
\[ A \otimes A \xmigi{\phi(\coml_1) \otimes id} A \otimes A 
\xmigi{\phi \otimes id} (A \rimp A) \otimes A 
\xmigi{ev} A, \]
we get a map $\phi':A \migi (A \limp A)$, which sends $a$ to $(a' \mapsto a' \lapp a)$.
Applying $\Phi'$ three times to 
\[ A \otimes ((A \limp A) \otimes (A \limp A)) 
\xmigi{{\rm associator}} (A \otimes (A \limp A)) \otimes (A \limp A) 
\xmigi{ev \otimes id} A \otimes (A \limp A)
\xmigi{ev} A, \]
we get a map $p: I \migi (A \limp A) \limp ((A \limp A) \limp (A \limp A))$.
Define a map $b_1$ as 
\[ I \xmigi{p} (A \limp A) \limp ((A \limp A) \limp (A \limp A)) \xmigi{\phi' \limp (\phi' \limp id)} A \limp (A \limp (A \limp A)), \]
which sends $\ast$ to $x \mapsto (y \mapsto (z \mapsto (z \lapp y) \lapp x))$.
Take $M_1 \in \rlz{b_1 (\ast)}_{A \limp (A \limp (A \limp A))}$.

Let $\coml_2$ be a realizer of $\Phi (ev_2)$, where $ev_2 :A \otimes (A \limp (A \limp A)) \migi (A \limp A)$ is the evaluation map.
$\coml_2$ realizes a map $q:A \migi (A \rimp A)$ that sends $a$ to $\phi (\coml_2 \rapp a)$.
Let $\coml_3$ be a realizer of $\Phi (ev_3)$, where $ev_3: A \otimes (A \limp (A \limp (A\limp A))) \migi (A \limp (A \limp A))$ is the evaluation map.
Take $r:A \migi A$ as a map sending $x$ to $\coml_3 \rapp x \rapp M_1$, whose realizer is $\lamst x.\coml_3 \rapp x \rapp M_1$.
Applying $\Phi'$ to 
\[ A \otimes A \xmigi{q \otimes r} (A \rimp A) \otimes A \xmigi{ev} A, \]
we get a map $b_2 : I \migi (A \limp (A \limp A))$, which sends $\ast$ to $(x \mapsto (y \mapsto \coml_2 \rapp y \rapp (\coml_3 \rapp x \rapp M_1)))$.
Take $M_2 \in \rlz{b_2 (\ast)}_{A \limp (A \limp A)}$.

Let $b_3:A \migi A$ be a map sending $x$ to $\coml_2 \rapp x \rapp M_2$, whose realizer is $\lamst x.\coml_2 \rapp x \rapp M_2$.
When we take $\combl \in \rlz{b_3}_{A \limp A}$, for any $x \in \uhka$,
\begin{eqnarray*}
x \lapp \combl &=& \coml_1 \rapp x \rapp \combl \\
&=& b_3 (x) \\
&=& \coml_2 \rapp x \rapp M_2.
\end{eqnarray*}
For any $y \in \uhka$,
\begin{eqnarray*}
y \lapp (x \lapp \combl) &=& y \lapp (\coml_2 \rapp x \rapp M_2) \\
&=& \coml_1 \rapp y \rapp (\coml_2 \rapp x \rapp M_2) \\
&=& b_2 (\ast) (x) (y) \\
&=& \coml_2 \rapp y \rapp (\coml_3 \rapp x \rapp M_1).
\end{eqnarray*}
For any $z \in \uhka$,
\begin{eqnarray*}
z \lapp (y \lapp (x \lapp \combl)) &=& z \lapp (\coml_2 \rapp y \rapp (\coml_3 \rapp x \rapp M_1)) \\
&=& \coml_1 \rapp z \rapp (\coml_2 \rapp y \rapp (\coml_3 \rapp x \rapp M_1)) \\
&=& b_1(\ast)(x)(y)(z) \\
&=& (z \lapp y) \lapp x.
\end{eqnarray*}

Next we obtain $\comdr$.
Applying $\Phi'$ and $\Phi$ to 
\begin{center}
$A \otimes ((A \limp (A \rimp A)) \otimes A) \xmigi{{\rm associator}} (A \otimes (A \limp (A \rimp A))) \otimes A \xmigi{ev \otimes id} (A \rimp A) \otimes A \xmigi{ev} A$,
\end{center}
we get a map $d:(A \limp (A \rimp A)) \migi ((A \limp A) \rimp A)$, which sends a map $(a \mapsto (a' \mapsto f(a,a')))$ to the map $(a' \mapsto (a \mapsto f(a,a')))$.

When we take $\comdr$ as a realizer of
\[ A \xmigi{\phi'} 
(A \limp A) \xmigi{id \limp \phi}
(A \limp (A \rimp A)) \xmigi{d} ((A \limp A) \rimp A), \]
\begin{eqnarray*}
x \lapp (\comdr \rapp y \rapp z) 
&=& d(\phi \circ (\phi'(y)))(z)(x) \\
&=& (\phi \circ (\phi'(y)))(x)(z) \\
&=& (x \lapp y) \rapp z
\end{eqnarray*}
for any $x,y,z \in \uhka$.

Finally we obtain $\comdl$.
Applying $\Phi$ and $\Phi'$ to 
\[ (A \otimes ((A \limp A) \rimp A)) \otimes A \xmigi{{\rm associator}} A \otimes (((A \limp A) \rimp A) \otimes A) \xmigi{id \otimes ev} A \otimes (A \limp A) \xmigi{ev} A, \]
we get a map $d_1:((A \limp A) \rimp A) \migi (A \limp (A \rimp A))$, sending a map $(a' \mapsto (a \mapsto f(a',a)))$ to the map $(a \mapsto (a' \mapsto f(a',a)))$.
Take $N_1 \in \rlz{d_1 \circ (\phi' \rimp id) \circ \phi}_{A \limp (A \limp (A \rimps A))}$.

Let $\coml_4$ be a realizer of $\Phi (ev_4)$, where $ev_4: A \otimes (A \limp (A\rimp A)) \migi (A \rimp A)$ is the evaluation map.
$\coml_4$ realizes a map $s:A \migi (A \rimp A)$ sending $x$ to $\phi (\coml_4 \rapp x)$.
Let $\coml_5$ be a realizer of a map obtained by applying $\Phi$ to 
\[ ev_5: A \otimes (A \limp (A \limp (A\rimp A))) \migi (A \limp (A \rimp A)) \]
and $t:A \migi A$ be a map sending $a$ to $\coml_5 \rapp a \rapp N_1$, whose realizer is $\lamst x.\coml_5 \rapp x \rapp N_1$.
Applying $\Phi'$ to
\[ A \otimes A \xmigi{s \otimes t} (A \rimp A) \otimes A \xmigi{ev} A, \]
we get a map $d_2:A \migi (A \limp A)$ sending $y$ to
$(x \mapsto  (\coml_4 \rapp x \rapp (\coml_5 \rapp y \rapp N_1)))$.
Take a realizer $N_2 \in \rlz{d_2}_{A \limp (A \limp A)}$.

Let $d_3:A \migi A$ be a map sending $x$ to $\coml_2 \rapp x \rapp N_2$, whose realizer is $\lamst x. \coml_2 \rapp x \rapp N_2$.
When we take $\comdl \in \rlz{d_3}_{A \limp A}$, for any $y \in \uhka$, 
\begin{eqnarray*}
y \lapp \comdl &=& \coml_1 \rapp y \rapp \comdl \\
&=& d_3 (y) \\
&=& \coml_2 \rapp y \rapp N_2.
\end{eqnarray*}
For any $x \in \uhka$, 
\begin{eqnarray*}
x \lapp (y \lapp \comdl) &=& x \lapp (\coml_2 \rapp y \rapp N_2) \\
&=& \coml_1 \rapp x \rapp (\coml_2 \rapp y \rapp N_2) \\
&=& d_2 (y)(x) \\
&=& \coml_4 \rapp x \rapp (\coml_5 \rapp y \rapp N_1).
\end{eqnarray*}
For any $z \in \uhka$, 
\begin{eqnarray*}
(x \lapp (y \lapp \comdl)) \rapp z &=& \coml_4 \rapp x \rapp (\coml_5 \rapp y \rapp N_1) \rapp z \\
&=& (d_1 \circ (\phi' \rimp id) \circ \phi) (y)(x)(z) \\
&=& ((\phi' \circ (\phi(y)))(z)(x) \\
&=& x \lapp (y \rapp z).
\end{eqnarray*}
\end{proof}

In this section we showed propositions for the necessary conditions to obtain certain structures on categories of assemblies.
Next, consider whether the similar propositions hold for the cases of categories of modest sets.
The next propositions can be proven in the same way as Proposition \ref{propnecccc}, \ref{propnecmul} and \ref{propnecclo}.

\begin{prop} \label{propnecccc2}
Suppose $\hka$ is a total applicative structure and $\catc := \moda$ happens to be a CCC.
$\hka$ is an $\sk$-algebra if the following conditions hold.
\begin{enumerate}[(i)]
\item $|Y^X| = \homrm_{\catc} (X,Y)$ and $\rlz{f}_{Y^X} = \{ r \mid \mbox{$r$ realizes $f$} \}$.
\item For $f:X' \migi X$ and $g:Y\migi Y'$, $g^f : Y^X \migi Y'^{X'}$ is the function sending $h:X\migi Y$ to $g \circ h \circ f$.
\item The forgetful functor from $\catc$ to $\cats$ strictly preserves finite products. \label{confor}
\item The adjunction $\Phi: \homrm_{\catc} (X \times Y, Z) \migi \homrm_{\catc} (X, Z^Y)$ is the function sending a function $f$ to the function $x \mapsto (y \mapsto f(x,y))$. \qed
\end{enumerate}
\end{prop}

\begin{prop} \label{propnecmul2}
Suppose $\hka$ is a total applicative structure and $\catc := \moda$ happens to be a closed multicategory.
$\hka$ is a $\bikuro$-algebra if the following conditions hold.
\begin{enumerate}[(i)]
\item $\catc (;X) = |X|$ and $\ucatc(;X) = X$.
\item $\catc (X;Y) = \homrm_{\catc} (X,Y)$ and $\ucatc(X;Y) =(\homrm_{\catc} (X,Y), \erlz)$, where \\
$\rlz{f}= \{ r \mid \mbox{$r$ realizes $f$} \}$.
\item $\ucatc (X_1,\dots ,X_{n};Y) = \ucatc (X_{1};\ucatc (X_2,\dots ,X_n;Y))$ and $\catc (X_1,\dots ,X_n;Y)$ is the underlying set of $\ucatc (X_1,\dots ,X_n;Y)$.

\item For $g:Y_1,\dots ,Y_n \migi Z$ and $f_l:X^l_1,\dots ,X^l_{k_l} \migi Y_l$, $g \circ (f_1,\dots ,f_n)$ is the function sending $x^1_1,\dots ,x^1_{k_1},\dots ,x^n_{k_n}$ to $g(f_1 (x^1_1,\dots ,x^1_{k_1}),\dots ,f_n (x^n_1,\dots ,x^n_{k_n}))$.
When $k_l = 0$ for some $1 \leq l \leq n$, $g \circ (f_1,\dots ,f_n)$ is the function given $y_l \in |Y_l|$ pointed by $f_l$ as the $l$-th argument of $g$.
\item $ev_{X_1,\dots ,X_n;Y}$ sends $f, x_1,\dots ,x_n$ to $f(x_1,\dots ,x_n)$.
\item $\Lambda_{Z_1,\dots ,Z_m;X_1,\dots ,X_n;Y}$ sends a function $(z_1,\dots ,z_m,x_1,\dots ,x_n \mapsto f(z_1,\dots ,z_m,x_1,\dots ,x_n))$ to the function $(z_1,\dots ,z_m \mapsto f(z_1,\dots ,z_m,-,\dots ,-))$. \qed
\end{enumerate}
\end{prop}

\begin{prop} \label{propnecclo2}
Suppose $\hka$ is a total applicative structure and $\catc := \moda$ happens to be a closed category.
$\hka$ is a $\biikuro$-algebra if the following conditions hold.
\begin{enumerate}[(i)]
\item $|Y \rimp X| = \homrm_{\catc} (X,Y)$ and $\rlz{f}_{Y \rimps X} = \{ r \mid \mbox{$r$ realizes $f$} \}$.
\item $g \rimp f : (Y \rimp X) \migi (Y' \rimp X')$ is the function sending $h:X \migi Y$ to $g \circ h \circ f$.
\item The underlying set of the unit object $I$ is the singleton $\{ \ast \}$. \label{concloextra}
\item $i_X$ is the function sending a function $(f:\ast \mapsto x)$ to $x$.
\item $L_{Y,Z}^X$ is the function sending $g:Y \migi Z$ to the function $(f:X \migi Y) \mapsto g \circ f$. 
\end{enumerate} \qed
\end{prop}

Here note that Proposition \ref{propnecclo2} has one more condition, that the underlying set of the unit object is a singleton, than Proposition \ref{propnecclo}.
This is because the assembly $X$ we used in Remark \ref{remsingle} is not a modest set.

On the other hand, for the cases of SMCCs, monoidal closed categories and monoidal bi-closed categories, we cannot state propositions for modest sets similar to Proposition \ref{propnecsmcc},  \ref{propnec5} and \ref{propnecbiclo}.
Since we define tensor products in categories of modest sets in the different way from those of categories of assemblies (as seen in the proof of Proposition \ref{propsmcc2}), the condition ``the forgetful functor from $\catc$ to $\cats$ is strict monoidal" is not appropriate for the case of modest sets.

For the case of SMCCs, we can avoid this problem by presenting a more generalized proposition, that is for symmetric closed categories, instead of SMCCs.
A {\it symmetric closed category} is a closed category with a natural isomorphism 
\[ S_{X,Y,Z} : (Z \rimp Y) \rimp X \cong (Z \rimp X) \rimp Y \] 
satisfying appropriate axioms (\cf~\cite{laplaza2}).

\begin{prop} \label{propnecsmcc2}
Suppose $\hka$ is a total applicative structure and $\catc := \asmca$ (or $\moda$) happens to be a symmetric closed category.
$\hka$ is a $\bci$-algebra if the following conditions hold.
\begin{enumerate}[(i)]
\item $|Y \rimp X| = \homrm_{\catc} (X,Y)$ and $\rlz{f}_{Y \rimps X} = \{ r \mid \mbox{$r$ realizes $f$} \}$.
\item $g \rimp f : (Y \rimp X) \migi (Y' \rimp X')$ is the function sending $h:X \migi Y$ to $g \circ h \circ f$.
\item $L_{Y,Z}^X$ is the function sending $g:Y \migi Z$ to the function $(f:X \migi Y) \mapsto g \circ f$.
\item $S_{X,Y,Z}$ is the function sending $f: x \mapsto (y \mapsto f(x)(y))$ to $S(f) : y \mapsto (x \mapsto f(y)(x))$.
\end{enumerate}
\end{prop}

This proposition also shows that we cannot obtain $\asmca$ (or $\moda$) that is a symmetric closed category but not an SMCC, in the canonical way.

For the cases of monoidal closed categories and monoidal bi-closed categories, it is still not clear that there are any appropriate conditions to state propositions for modest sets similar to Proposition \ref{propnec5} and \ref{propnecbiclo}.

%%%%%%%%%%%%%%%%%%%%%%%%%%%%%%%%%%%%%%%%%%%%%%%%%%%%%%%%%%%%%%%%%
%%%%%%%%%%%%%%%%%%%%%%%%%%%%%%%%%%%%%%%%%%%%%%%%%%%%%%%%%%%%%%%%%
\section{Planar linear combinatory algebras} \label{secplca}

In Section \ref{seclca}, we recalled LCAs and rLCAs, that relate $\bci$-algebras and PCAs, and that induce categorical models of linear exponential modalities.
In this section, we apply the similar construction to $\biilp$-algebras.
We reformulate rLCAs for $\biilp$-algebras and PCAs, and call them {\it exp-rPLCAs}.
From an exp-rPLCA, we get a categorical model of $!$-modality on the non-symmetric multiplicative intuitionistic linear logic (MILL).
Also we reformulate rLCAs for $\biilp$-algebras and $\bci$-algebras, and call them {\it exch-rPLCAs}.
From an exch-rPLCA, we obtain a model for an exchange modality relating the non-symmetric MILL and the symmetric MILL.

In \cite{tomita2}, we already introduced the same construction called ``rPLCAs," based on bi-$\bdi$-algebras.
What defined as rPLCAs in this section are generalizations of those in \cite{tomita2}, based on $\biilp$-algebras.

%%%%%%%%%%%%%%%%%%%%%%%%%%%%%%%%%%%%%%%%%%%%%%%%%%%%%%%%%%%%%%%%%%%
\subsection{Exponential planar linear combinatory algebras} \label{secexpplca}

Linear exponential comonads on non-symmetric monoidal categories are investigated in \cite{hasegawa1}, which model $!$-modalities on non-symmetric MILL.

\begin{defi} \label{defmeiji}
A {\it linear exponential comonad} on a monoidal category $\catc$ consists of the following data.
\begin{itemize}
\item A monoidal comonad $(!, \delta, \epsilon, m, m_I)$. Here $!$ is an endofunctor on $\catc$,
$\delta_X : !X \migi !!X$ and $\epsilon : !X \migi X$ are monidal natural transformations for the comultiplication and the counit. A natural transformation $m_{X,Y} : !X \otimes !Y \migi !(X \otimes Y)$ and a map $m_I : I \migi !I$ make $!$ be a monoidal functor.
\item Monoidal natural transformations $e_X : !X \migi I$ and $d_X : !X \migi !X \otimes !X$.
\item A monidal natural transformation $\sigma_{X,Y} : !X \otimes !Y \migi !Y \otimes !X$ defined as
\[ !X \otimes !Y \xmigi{\delta_X \otimes \delta_Y} !!X \otimes !!Y \xmigi{m_{!X,!Y}} !(!X \otimes !Y) \xmigi{d_{!X \otimes !Y}} !(!X \otimes !Y) \otimes !(!X \otimes !Y) \]
\[ \xmigi{!(e_X \otimes id) \otimes !(id \otimes e_Y)} !(I \otimes !Y) \otimes !(!X \otimes I) \xmigi{!({\rm unitor}) \otimes !({\rm unitor})} !!Y \otimes !!X \xmigi{\epsilon_{!Y} \otimes \epsilon_{!X}} !Y \otimes !X. \]
\end{itemize}
Here these components need satisfy the following conditions.
\begin{enumerate}[(i)]
\item The following diagram commutes: \\
\xymatrix@C=30pt{
!X \otimes !X \otimes !Y \otimes !Y \otimes !Z \otimes!Z \ar[r]^{id \otimes \sigma \otimes id} \ar[d]_{id \otimes \sigma \otimes id} &
!X \otimes !Y \otimes !X \otimes !Y \otimes !Z \otimes !Z \ar[d]^{m \otimes m\otimes id} \\
!X \otimes !X \otimes !Y \otimes !Z \otimes !Y \otimes !Z \ar[d]_{id \otimes m \otimes m} &
!(X \otimes Y) \otimes !(X \otimes Y) \otimes !Z \otimes !Z \ar[d]^{id \otimes \sigma \otimes id} 
\\
!X \otimes !X \otimes !(Y \otimes Z) \otimes !(Y \otimes Z) \ar[d]_{id \otimes \sigma \otimes id} &
!(X \otimes Y) \otimes !Z \otimes !(X \otimes Y) \otimes !Z \ar[d]^{m \otimes m} \\
!X \otimes !(Y \otimes Z) \otimes !X \otimes !(Y \otimes Z) \ar[r]_{m \otimes m} &
!(X \otimes Y \otimes Z) \otimes !(X \otimes Y \otimes Z)
}

\item $m_{!Y ,!X} \circ \sigma_{!X , !Y} = !\sigma_{X,Y} \circ m_{!X ,!Y}$.
\item $\sigma_{X,Y}^{-1} = \sigma_{Y,X}$.
\item The following diagram commutes: \\
\xymatrix@C=50pt{
!X \otimes !Y \otimes !Z \ar[r]^{\delta_X \otimes \delta_Y \otimes id} \ar[d]_{id \otimes \sigma_{Y,Z}} &
!!X \otimes !!Y \otimes !Z \ar[r]^{m_{!X,!Y} \otimes id} &
!(!X \otimes !Y) \otimes !Z \ar[d]^{\sigma_{!X \otimes !Y, Z}} \\
!X \otimes !Z \otimes !Y \ar[rrd]_{\sigma_{X,Z} \otimes id} & &
!Z \otimes !(!X \otimes !Y) \ar[d]^{id \otimes \epsilon_{!X \otimes !Y}} \\
& & !Z \otimes !X \otimes !Y
}

\item The following diagram commutes: \\
\xymatrix@C=50pt{
!X \otimes !Y \ar[r]^{d_X \otimes d_Y} \ar[d]_{m_{X,Y}} &
!X \otimes !X \otimes !Y \otimes !Y \ar[r]^{id \otimes \sigma \otimes id} &
!X \otimes !Y \otimes !X \otimes !Y \ar[d]^{m \otimes m} \\
!(X \otimes Y) \ar[rr]_{d_{X \otimes Y}} & &
!(X \otimes Y) \otimes !(X \otimes Y)
}
\item The following diagram commutes: \\
\xymatrix@C=50pt{
I \ar[rd]^{m_I \otimes m_I} \ar[d]_{m_I} & \\
!I \ar[r]_{d_I} & !I \otimes !I
}
\item $(!X, e_X,d_X)$ is a comonoid in $\catc$.
\item $e_X$ and $d_X$ are coalgebra morphisms.
\item $\delta_X$ is a comonoid morphism.
\end{enumerate}
\end{defi}

Then we will introduce the categorical realizability to inducing linear exponential comonads on non-symmetric monoidal categories.
The results are reformulations of a part of contents in \cite{hoshino} and \cite{tomita2} to the case of $\biilp$-algebras.

\begin{defi}
An {\it exponential relational planar linear combinatory algebra (exp-rPLCA)} consists of a $\biilp$-algebra $\hka$ and a comonadic applicative morphism $(\bfban,\bfe,\bfd)$ on $\hka$ which satisfies the followings.
\begin{itemize}
\item There is $\bfk \in |\hka|$ such that $\bfk x (\bfban y) \subseteq \{ x \}$ for any $x , y \in |\hka|$.
\item There is $\bfw \in |\hka|$ such that $\bfw x (\bfban y) \subseteq x (\bfban y) (\bfban y)$ for any $x , y \in |\hka|$.
\end{itemize}
\end{defi}

While the above definition employs the different style from rLCAs of Definition \ref{defrlca}, we can also define exp-rPLCAs in the same style.

\begin{prop} \label{propstyle}
For a $\biilp$-algebra $\hka$ and a comonadic applicative morphism $(\bfban,\bfe,\bfd)$ on $\hka$, the followings are equivalent.
\begin{enumerate}
\item $(\hka,\bfban)$ is an exp-rPLCA.
\item Take two total relations $[\bfban,\bfban]:\hka \migi \hka$ and $k_i :\hka \migi \hka$ as $[\bfban,\bfban](x) := \{ \comp a a' \mid a,a' \in \bfban x \}$ and $k_i (x) := \{ \comi \}$.
Then they are applicative morphisms and $\bfban \preceq [\bfban,\bfban]$ and $\bfban \preceq k_i$ hold.
\end{enumerate}
\end{prop}

\begin{proof} \hfill \\
(1)$\Rightarrow$(2):
Realizers of $[\bfban,\bfban]$ and $k_i$ exist as $\lamst pq.\bfw \comp (r_{\bfban} (\coml \bfk p)(\coml \bfk q))$ and $\comi$.
Realizers for $\bfban \preceq [\bfban,\bfban]$ and $\bfban \preceq k_i$ are $\bfw \comp$ and $\bfk \comi$. \\
(2)$\Rightarrow$(1):
Take a realizer $r_1$ of $\bfban \preceq [\bfban,\bfban]$ and a realizer $r_2$ of $\bfban \preceq k_i$.
Then $\bfk$ and $\bfw$ exist as $\lamst xy.\batui x (r_2 y)$ and $\lamst xy. \coml x (r_1 y)$.
\end{proof}

From an exp-rPLCA, we get a linear exponential comonad.

\begin{prop} \label{propexp-rPLCA}
For an exp-rPLCA $(\hka, \bfban)$, $\banst$ is a linear exponential comonad on $\asmca$.
\end{prop}

\begin{proof} \hfill
\begin{itemize}
\item It is easy to see that the comultiplication $\delta$ and the counit $\epsilon$ are monoidal natural transformations.
From Proposition \ref{proplax}, the comonad $\banst$ is a lax monoidal functor and thus we have $m_{X,Y} : \banst X \otimes \banst Y \migi \banst (X \otimes Y)$ and $m_I : I \migi \banst I$.
Therefore, we have $\banst$ as a monoidal comonad.
\item $e_X : \banst X \migi I$ is the function sending $x$ to $\ast$. A realizer for $e_X$ is $\bfk \comi$.
\item $d_X : \banst X \migi \banst X \otimes \banst X$ is the function sending $x$ to $x \otimes x$.
A realizer for $d_X$ is $\bfw (\lamst pq. \comp pq)$.
\item It is easy to see that the $(\banst, e_X, d_X)$ satisfies conditions for linear exponential comonads. \qedhere
\end{itemize}
\end{proof}

Next we try to obtain linear-non-linear models for the non-symmetric MILL, that is, monoidal adjunctions between (non-symmetric) monoidal closed categories and CCCs.
Although now we get a linear exponential comonad $\banst$ on $\asmca$, at this point it has not concluded that we obtain a linear-non-linear model, since we have not shown that the co-Kleisli adjunction between $\asmca$ and $\asmcaban$ is a monoidal adjunction.
To show this, we use the next proposition shown in \cite{hasegawa1}.

\begin{prop}
Let $\catc$ be a monoidal closed category and $!$ be a linear exponential comonad on $\catc$. When $\catc$ has finite products, the co-Kleisli category $\catc_!$ is a CCC and the co-Kleisli adjunction is monoidal.
\end{prop}

\begin{prop} \label{propexp-rPLCA2}
For an exp-rPLCA $(\hka, \bfban)$, $\asmca$ has Cartesian products, and thus the co-Kleisli adjunction between $\asmca$ and a CCC $\asmcaban$ is monoidal.
\end{prop}

\begin{proof} \hfill
\begin{itemize}
\item The terminal object is $(\{ \ast \}, \erlz)$, where $\rlz{\ast} := |\hka|$.
\item The underlying set of $X \times Y$ is $|X| \times |Y|$. Realizers are defined as
\[ \rlz{(x,y)} := \{ \ \comp (\comp uv) a \ \mid \ \mbox{$\exists p, \exists q$, $u \in {\bfban}p$, $v \in {\bfban}q$, $p a \in \rlz{x}_X$ and $q a \in \rlz{y}_Y$} \ \}. \]
The set of realizers is not empty since for $m \in \rlz{x}_X$ and $m' \in \rlz{y}_Y$,
\[ \comp (\comp (\bfban (\bfk m)) (\bfban (\bfk m'))) (\bfban \comi) \in \rlz{(x,y)}. \]
\item For maps $f:X \migi X'$ and $g :Y \migi Y'$ in $\asmca$, $f \times g$ is the function sending $(x,y)$ to $(f(x),g(y))$.
A realizer of $f \times g$ is 
$\lamst uv. \comp (\comp (r_{\bfban} M u) (r_{\bfban} N v))$, where $M \in \bfban (\comb r_f)$ and
$N \in \bfban (\comb r_g)$.
\item A realizer for the projection $\pi :X \times Y \migi X$ is 
$\coml (\coml (\lamst uv. \bfe (\bfk uv)))$.
A realizer for the projection $\pi' :X \times Y \migi Y$ is
$\coml (\coml (\lamst uv. \bfe (\bfk \comi uv)))$.
\item For any object $Z$ and any maps $f:Z \migi X$ and $g:Z \migi Y$, there exists a unique map
$h:Z \migi X \times Y$ such that $\pi \circ h = f$ and $\pi' \circ h =g$.
$h$ is the function sending $z$ to $(f(z),g(z))$, whose realizer is in
$\comp (\comp (\bfban r_f) (\bfban r_g))$. \qedhere
\end{itemize}
\end{proof}

For an exp-rPLCA $(\hka,\bfban)$, we can restrict $\banst :\asmca \migi \asmca$ to the comonad on $\moda$, as we saw in Remark \ref{remrest}.
By the same proof as the above, we also can get a linear-non-linear model using $\moda$.

\begin{prop}
For an exp-rPLCA $(\hka, \bfban)$, $\banst$ is a linear exponential comonad on $\moda$.
Moreover, the co-Kleisli adjunction between the monoidal closed category $\moda$ and the CCC $\moda_{\banst}$ is monoidal. \qed
\end{prop}

We have seen the co-Kleisli adjunctions obtained by an exp-rPLCA $(\hka, \bfban)$ are linear-non-linear models by showing that $\asmca$ and $\moda$ have Cartesian products.
We can further show that these categories have better structures as the next proposition says.

\begin{prop}
For an exp-rPLCA $(\hka, \bfban)$, $\asmca$ and $\moda$ are finitely complete and finitely cocomplete.
\end{prop}

\begin{proof}
First we show the proposition for $\asmca$.

\begin{itemize}
\item The terminal object and binary products are those in the proof of Proposition \ref{propexp-rPLCA2}.
\item Given maps $f,g :X \migi Y$, let $Z$ be an assembly defined as $|Z| := \{ x \in |X| \mid  f(x) = g(x) \}$ and $\rlz{x}_Z := \rlz{x}_X$.
Take a map $e :Z \migi X$ as the inclusion function, realized by $\comi$.
Then it is easy to see that this $e$ is the equalizer of $f$ and $g$.
\item The initial object is the empty set.
\item Given maps $f,g :X \migi Y$, take a set $|W| := Y/\sim$, where $\sim$ is the smallest equivalence relation satisfying $\forall x \in |X|, f(x) \sim g(x)$.
Take an assembly $W = (|W|,\erlz_W)$ by $\rlz{w}_W := \bigcup_{y \in w}{\rlz{y}_Y}$.
Take a map $e':Y \migi W$ by the projection, realized by $\comi$.
Then it is easy to see that this $e'$ is the coequalizer of $f$ and $g$.

\item The underlying set of $X+Y$ is $\{ (0,x) \mid x \in |X| \} \cup \{ (1,y) \mid y \in |Y| \}$.
Realizers are defined as
\begin{center}
$\rlz{(0,x)} := \{ \comp mp \mid p \in \rlz{x}_X \}$
and
$\rlz{(1,y)} := \{ \comp nq \mid q \in \rlz{y}_Y \}$,
\end{center}
where
$m := \lamst uv.\batui (\bfe u)(\bfk \comi v)$ and $n := \lamst uv.\bfk \comi u (\bfe v)$.

The coprojections $in_X:X \migi X+Y$ and
$in_Y :Y \migi X+Y$ are given as $x \mapsto (0,x)$ and $y \mapsto (1,y)$, and realized by $\comp m$ and $\comp n$ respectively.
Given maps $f:X \migi Z$ and $g:Y \migi Z$ realized by $r_f$ and $r_g$, we have a unique map $h:X+Y \migi Z$ such that $h \circ in_X = f$ and $h \circ in_Y = g$.
$h$ is the function sending $(0,x)$ to $f(x)$ and $(1,y)$ to $g(y)$, which is realized by $\coml (\lamst uv.u(\bfban r_f) (\bfban r_g) v)$.
\end{itemize}

Therefore, $\asmca$ is finitely complete and finitely cocomplete.

Since $\moda$ is the reflexive full subcategory of $\asmca$, $\moda$ is also finitely complete and finitely cocomplete.
\end{proof}

As an adjoint pair between a $\bci$-algebra and a PCA gives rise to an rLCA, an adjoint pair between a $\biilp$-algebra and a PCA gives rise to an exp-rPLCA and a monidal adjunction.

\begin{prop}
Let $(\delta \dashv \gamma): \hka \migi \hkb$ be an adjoint pair for a $\biilp$-algebra $\hka$ and a PCA $\hkb$.
\begin{enumerate}
\item $(\hka, \delta \circ \gamma)$ forms an exp-rPLCA.
\item $(\delst \dashv \gamst): \asmca \migi \asmcb$ is a monoidal adjunction between the monoidal  category $\asmca$ and the Cartesian monoidal category $\asmcb$.
\end{enumerate}
\end{prop}

\begin{proof} \hfill
\begin{enumerate}
\item From Proposition \ref{propmoradj} (\ref{propmoradj2}), $\delta \circ \gamma$ is a comonadic applicative morphism.
Let $\bfe$ and $\bfd$ be elements for the counit and the comultiplication.
Then we can take $\bfk \in \uhka$ as an element of $\lamst xy.\coml x (\bfe (r_{\delta} (\delta M) y))$, where $M \in \lamst z.(\gamma \comi)$.
Also we can take $\bfw \in \uhka$ as an element of $\lamst xy. \coml x (\bfe (r_{\delta} (\delta N) (\bfd y)))$, where $N \in \lamst z.r_{\gamma} (r_{\gamma} (\gamma \comp) z) z$.

\item We show that the left adjoint $\delst$ is strong monoidal.
Let $\bfe \in |\hka|$ and $\bfi \in |\hkb|$ be elements such that $\forall a \in |\hka|, \bfe (\delta (\gamma a)) =a$ and $\forall b \in |\hkb|, \bfi b \in \gamma (\delta b)$.
The map $I \migi \delst 1$ is realized by $\kuro{(\delta \comk)}$.
The inverse $\delst 1 \migi I$ is realized by $\lamst a. \bfe (r_{\delta} (\delta (\lamst b.(\gamma \comi))) a)$.
The natural transformation $(\delst X) \otimes (\delst Y) \migi \delst(X \times Y)$ is realized by 
$\coml (\lamst aa'.r_{\delta} (r_{\delta} (\delta (\lamst bb't.tbb')) a) a')$.
The inverse map $\delst(X \times Y) \migi (\delst X) \otimes (\delst Y)$ is realized by 
$\lamst u.\bfe (r_{\delta} (\delta M) u)$, where 
$M \in \coml (\lamst bb'.r_{\gamma} (r_{\gamma} (\gamma \comp) (\bfi b) ) (\bfi b'))$. \qedhere
\end{enumerate}
\end{proof}

Next we consider the functional case of exp-rPLCAs, like LCAs are the functional case of rLCAs.

\begin{defi}
An {\it exponential planar linear combinatory algebra} (exp-PLCA) is an exp-rPLCA $(\hka, \bfban)$ that $\bfban$ is functional.
\end{defi}

Not only are exp-PLCAs special cases of exp-rPLCAs, but also can induce adjoint pairs between $\biilp$-algebras and PCAs.

\begin{prop} \label{propfunplca}
Let $(\hka, \bfban)$ be an exp-PLCA.
\begin{enumerate}
\item We have a PCA $\hka_{\bfban} = (\uhka, @)$ with $x @ y := x (\bfban y)$.
\item Let $\gamma : \hka \migi \hka_{\bfban}$ be the identity function and $\delta: \hka_{\bfban} \migi \hka$ be the function $x \mapsto \bfban x$. Then $\gamma$ and $\delta$ are applicative morphisms and $\delta \dashv \gamma$. \label{propfunplca1}
\end{enumerate}
\end{prop}

\begin{proof} \hfill
\begin{enumerate}
\item We have the $\comk$-combinator in $\hka_{\bfban}$ as $\lamst xy.\bfe (\bfk xy)$.
We have the $\comsa$-combinator as $\lamst xyz. \bfw (M x)(r_{\bfban} (r_{\bfban} (\bfban \comp)(\bfd y))(\bfd z))$, where 
$M:=\lamst xyz. \bfe x (\coml (\bfk \comi) (\bfe y)) (\coml (\lamst uv. r_{\bfban} u (\bfd v)) (\bfe z))$.
\item Realizers of $\gamma$ and $\delta$ are $\lamst xy.x(\bfe y)$ and $\lamst xy.r_{\bfban} x(\bfd y)$. A realizer for $\delta \circ \gamma \preceq id_{\hka}$ is $\bfe$ and for $id_{\hka_{\bfban}} \preceq \gamma \circ \delta$ is $\comi$. \qedhere
\end{enumerate}
\end{proof}

Next we give an (functional) adjoint pair between a $\biilp$-algebra and a PCA.
This example is a reformulation of the linear lambda calculus with $!$ (\cf~\cite{simpson}) to a planar variant.

\begin{exa} \label{examrplca}
Suppose infinite supply of variables $x,y,z,\dots $.
Terms are defined grammatically as follows.
\begin{center}
$M ::= x \; | \; MM' \; | \; \lambda x.M \; | \; M \otimes M' \; | \; \llet{x \otimes x'}{M}{M'} \; | \; !M \; | \; \lambda !x.M$
\end{center}
Here $x$ of $\lambda x.M$ is the rightmost free variable of $M$, appears exactly once in $M$ and is not in any scope of $!$.
Also we assume that for $\llet{x \otimes x'}{M}{M'}$, $x'$ and $x$ are the rightmost and the next rightmost free variables of $N$, appear exactly once in $N$ and are not in any scope of $!$.
Take an equational relation on terms as the congruence of the following equational axioms.
\begin{itemize}
\item $(\lambda x.M)N = M[N/x]$.
\item $M = \lambda x.Mx$.
\item $(\lambda !x.M)(!N) = M[N/x]$.
\item $\llet{x \otimes x'}{M \otimes M'}{N} = N[M/x][M'/x']$.
\item $M= \llet{x \otimes y}{M}{x \otimes y}$.
\end{itemize}
Let $\Lambda$ be the set of equivalence classes of closed terms.

Then we get a $\biilp$-algebra $\hka$, whose underlying set is $\Lambda$ and the application is that of lambda terms.
Also we get a PCA $\hkb = (\Lambda,@)$, where $M@N := M (!N)$.
Here the $\comk$-combinator and the $\comsa$-combinator of $\hkb$ exist as $\lambda !x.\lambda !y.x$ and $\lambda !x.\lambda !y.\lambda !z. x(!z) (!(y(!z)))$.

Take an applicative morphism $\gamma:\hka \migi \hkb$ as the identity function whose realizer is $\lambda !x. \lambda !y.xy$.
Take $\delta :\hkb \migi \hka$ as a function $M \mapsto !M$ whose realizer is $\lambda !x.\lambda !y. !(x(!y))$. Then we have an adjoint pair $\delta \dashv \gamma$.
\end{exa}

As well as we can construct an LCA from a ``reflexive object'' in a ``weak linear category'' (See \cite{ahs} and \cite{Haghverdi} ), we can get exp-PLCAs by appropriate settings.

\begin{defi}  \label{defwplc}
A {\it weak planar linear category (WPLC)} consists of:
\begin{enumerate}
\item a monoidal closed category $(\catc,\otimes,I)$ (not symmetric in general);
\item a monoidal functor $(!,m,m_I)$ on $\catc$;
\item a monoidal pointwise natural transformation $! \migi id_{\catc}$;
\item a monoidal pointwise natural transformation $! \migi !!$;
\item a monoidal pointwise natural transformation $! \migi ! \otimes !$;
\item a monoidal pointwise natural transformation $! \migi K_I$, where $K_I$ is the constant $I$ functor.
\end{enumerate}
Here a {\it pointwise natural transformation} $\gamma:F \migi G$ is a family of maps $\gamma_C :F(C) \migi G(C)$ 
($C \in Ob(\catc)$) satisfying that $G(f) \circ \gamma_I = \gamma_C \circ F(f)$ for any $f:I \migi C$.
\end{defi}

To be a WPLC, we need not all of the conditions for linear exponential comonads (Definition \ref{defmeiji}).
For instance, a WPLC does not require that $!$ is a comonad, and does not require the (ordinary) naturality of each transformation.

\begin{defi}
Let $(\catc,!)$ be a WPLC.
We say $V$ is a {\it reflexive object} when there are:
\begin{enumerate}[(i)]
\item a retraction $p: !V \triangleleft V :q$;
\item an isomorphism $r:(V \rimp V) \migi V$ and $s := r^{-1}$;
\item a retraction $t: (V \otimes V) \triangleleft V:u$.
\end{enumerate}
\end{defi}

As we saw in Example \ref{examreflexive}, for a reflexive object $V$ of a WPLC $\catc$, $\uhka := \catc(I,V)$ forms a $\biilp$-algebra $\hka$.
Furthermore, by giving $\bfban$ as an endofunction sending $M:I \migi V$ to
$p \circ (!M) \circ m_I$, $(\hka,\bfban)$ becomes an exp-PLCA.
The proof is the same as for WLCs and LCAs in \cite{Haghverdi}.

%%%%%%%%%%%%%%%%%%%%%%%%%%%%%%%%%%%%%%%%%%%%%%%%%%%%%%%%%%%%%%%%%%%
\subsection{Exchange planar linear combinatory algebras} \label{secexchplca}

Exchange modalities on the Lambek calculus and their categorical models are introduced in \cite{paiva}.
While the word ``Lambek calculus" may indicate various logics, type systems or grammars (\cf~\cite{lambek2,lambek3}), here we call the Lambek calculus as a variant of non-symmetric MILL with left and right implications.
The Lambek calculus is modeled by monoidal bi-closed categories.
While the order of arguments cannot be exchanged in the Lambek calculus, the Lambek calculus can be extended to a sequent calculus that allows swapping arguments with modalities.
This sequent calculus is called the {\it commutative/non-commutative (CNC) logic}, that is composed of two (commutative and non-commutative) logics, and the {\it exchange modality} connects these two parts.
Categorical models of the CNC logic are given as monoidal adjunctions between monoidal bi-closed categories and SMCCs, that are called {\it Lambek adjoint models}.
In this subsection, we introduce the similar construction to the previous subsection, inducing Lambek adjoint models.

\begin{defi}
An {\it exchange relational planar linear combinatory algebra} (exch-rPLCA) consists of a $\biilp$-algebra $\hka$ and a comonadic applicative morphism $(\xi,\bfe,\bfd)$ on $\hka$ with $\bfc \in \uhka$ satisfying $\bfc x (\xi y) (\xi z) \subseteq x (\xi z) (\xi y)$ for any $x,y,z \in \uhka$.
When $\xi$ is functional, we call $(\hka,\xi)$ an {\it exchange planar linear combinatory algebra} (exch-PLCA).
\end{defi}

\begin{prop}
For an exch-rPLCA $(\hka, \xi)$, the co-Kleisli category $\asmcaxi$ is an SMCC and the co-Kleisli adjunction between $\asmca$ and $\asmcaxi$ is monoidal.
\end{prop}

\begin{proof} \hfill
\begin{itemize}
\item We define tensor products in $\asmcaxi$ as $X \ootimes Y := (|X| \times |Y|,\erlz)$, where
\[ \rlz{x \ootimes y} := \{ \comp pq \mid \mbox{$p \in \xi \rlz{x}_X$ and $q \in \xi \rlz{y}_Y$} \}. \]

\item For maps $f:X \migi X'$ and $g:Y \migi Y'$ in $\asmcaxi$, $f \ootimes g$ is the function sending $x \ootimes y$ to $f(x) \ootimes g(y)$.
A realizer of $f \ootimes g$ is $\lamst z. \coml M (\bfe z)$, where
$M \in \lamst pq. \comp (r_{\xi} (\xi r_f) (\bfd p)) (r_{\xi} (\xi r_g) (\bfd q))$.
\item We define the unit object $J$ of $\asmcaxi$ as $(\{ \ast \}, \erlz_J)$, where $\rlz{\ast}_J := \{ \comi \}$.
\item A realizer for the left unitor $\lambda_X:J \ootimes X \migi X$ is $\lamst u. \coml (\lamst p. \bfe p \bfe) (\bfe u)$.
A realizer for the inverse $\lambda_X^{-1}$ is in $\comp (\xi \comi)$.
\item A realizer for the right unitor $\rho_X:X \migi X \ootimes J$ is in $\lamst p.\comp p (\xi \comi)$.
A realizer for the inverse $\rho_X^{-1}$ is $\lamst u. \coml (\bfc \bfe) (\bfe u)$.
\item A realizer for the associator $\alpha_{XYZ} :(X \ootimes Y) \ootimes Z \migi X \ootimes (Y \ootimes Z)$ is $\lamst u. \coml (\lamst v. \coml M (\bfe v)) (\bfe u)$, where $M \in \lamst pqr. \comp p (r_{\xi} (r_{\xi} (\xi \comp) (\bfd q) (\bfd r)))$.
A realizer for $\alpha_X^{-1}$ is $\lamst u. \coml (\lamst vw.\coml (M' v) (\bfe w)) (\bfe u)$, where $M' \in \lamst pq. \comp (r_{\xi} (r_{\xi} (\xi \comp) (\bfd p) (\bfd q)))$.
\item The symmetry $\sigma_{XY}:X \ootimes Y \migi Y \ootimes X$ is the function sending $x \ootimes y$ to $y \ootimes x$.
A realizer for $\sigma_{XY}$ and $\sigma_{XY}^{-1}$ is $\lamst u.\coml (\bfc \comp) (\bfe u)$.
\item For objects $X$ and $Y$, the exponential in $\asmcaxi$ is $Y \rimp X = (\homrm_{\asmca}(\xist X, Y), \erlz)$, where $\rlz{f} := \{ r \in \uhka \mid \mbox{$r$ realizes $f$} \}$.
\item For maps $f:X' \migi X$ and $g:Y \migi Y'$ in $\asmcaxi$, $g \rimp f$ is the function sending a map $h:X \migi Y$ in $\asmcaxi$ to $g \circ (\xist h) \circ d_X \circ (\xist f) \circ d_{X'}$, where $d_X : \xist X \migi \xist \xist X$ is the comultiplication of $\xist$.
A realizer for $g \rimp f$ is $\lamst uv. r_g (r_{\xi} u (\bfd (r_{\xi} (\xi r_f) (\bfd v))))$.
\item The evaluation map $ev_{XY} : (Y \rimp X) \ootimes X \migi Y$ is the function sending $f \ootimes x$ to $f(x)$, that is realized by $\lamst u.\coml \bfe (\bfe u)$.
\item For any map $f:Z \ootimes X \migi Y$ in $\asmcaxi$, there exists a unique map $g:Z \migi Y \rimp X$ in $\asmcaxi$, which sends $z$ to $x \mapsto f(z \ootimes x)$.
$g$ is realized by $\lamst uv.r_f (r_{\xi} (r_{\xi} (\xi \comp) (\bfd u)) (\bfd v))$.
\item Finally we show that the co-Kleisli functor $\xist :\asmcaxi \migi \asmca$ is strong monoidal. 
We can take natural isomorphisms $\xist J \migi I$ and $\xist (X \ootimes Y) \migi \xist X \otimes \xist Y$ in $\asmca$ as the identity functions.
Realizers for $\xist J \migi I$ and $\xist (X \ootimes Y) \migi \xist X \otimes \xist Y$ are $\bfe$.
A realizer for $\xist J \migi I$ is in $\lamst u.u(\xi \comi)$.
A realizer for $\xist X \otimes \xist Y \migi \xist (X \ootimes Y)$ is in $\lamst uv. r_{\xi} (r_{\xi} (\xi \comp) (\bfd u)) (\bfd v)$. \qedhere
\end{itemize}
\end{proof}

The next proposition for categories of modest sets also can be shown in the same way as the above proposition.
Here since $X \ootimes Y$ in the above proof is not generally a modest set, we take the tensor product $\widehat{\boxtimes}$ in $\moda_{\xist}$ by the same way as Proposition \ref{propmodmcc}.
That is, we take $X \widehat{\boxtimes} Y = (|Z|,\erlz_Z)$ by $|Z| := (|X| \times |Y|)/\approx$, where $\approx$ and $\erlz_Z$ are defined as the same ones in the proof of Proposition \ref{propsmcc2}.

\begin{prop}
For an exch-rPLCA $(\hka, \xi)$, the co-Kleisli category $\moda_{\xist}$ is an SMCC and the co-Kleisli adjunction between is monoidal. \qed
\end{prop}

\begin{cor}
Suppose $\hka$ is a bi-$\bdi$-algebra and $(\hka, \xi)$ is an exch-rPLCA.
Then we have a Lambek adjoint model as the co-Kleisli adjunction between the monoidal bi-closed category $\asmca$ and the SMCC $\asmcaxi$ (or between $\moda$ and $\moda_{\xist}$). \qed
\end{cor}

Similar to exp-rPLCAs, adjoint pairs between $\biilp$-algebras and $\bci$-algebras correspond to exch-rPLCAs. 

\begin{prop}
Let $(\delta \dashv \gamma):\hka \migi \hkb$ be an adjoint pair for a $\biilp$-algebra $\hka$ and a $\bci$-algebra $\hkb$.
\begin{enumerate}
\item $(\hka,\delta \circ \gamma)$ forms an exch-rPLCA.
\item $(\delst \dashv \gamst):\asmca \migi \asmcb$ is a monoidal adjunction between the monoidal category $\asmca$ and an SMCC $\asmcb$. If $\hka$ is a bi-$\bdi$-algebra, the adjunction is a Lambek adjoint model.
\end{enumerate}
\end{prop}

\begin{proof} \hfill
\begin{enumerate}
\item From Proposition \ref{propmoradj} (\ref{propmoradj2}), $\delta \circ \gamma$ is a comonadic applicative morphism.
We can take
$\bfc$ in $\hka$ as $\lamst xyz.\coml x (\bfe (M (\bfd y)(\bfd z)))$, where $M \in \lamst y.r_{\delta} (r_{\delta} (\delta N) y)$ and
$N \in \lamst yz.r_{\gamma}(r_{\gamma} (\gamma \comp) z)y$.
\item It follows from Proposition \ref{propadjbiilp}. \qedhere
\end{enumerate}
\end{proof}

Similar to exp-PLCAs, exch-PLCAs induce adjoint pairs between $\biilp$-algebras and $\bci$-algebras.

\begin{prop}
Let $(\hka, \xi)$ be an exch-PLCA.
\begin{enumerate}
\item We have a $\bci$-algebra $\hka_{\xi} = (\uhka,@)$ with $x @ y := x(\xi y)$.
\item Let $\gamma:\hka \migi \hka_{\xi}$ be the identity function and $\delta:\hka_{\xi} \migi \hka$ be the function $x \mapsto \xi x$. Then $\gamma$ and $\delta$ are applicative morphisms and $\delta \dashv \gamma$.
\end{enumerate}
\end{prop}

\begin{proof} \hfill
\begin{enumerate}
\item We have the $\comc$-combinator in $\hka_{\xi}$ as $\lamst x.\bfc (\bfe x)$.
\item Same as the proof of Proposition \ref{propfunplca} (\ref{propfunplca1}). \qedhere
\end{enumerate}
\end{proof}

For an example of exch-PLCA, we have the similar calculus to Example \ref{examrplca}.

\begin{exa}
Suppose infinite supply of variables $x,y,z,\dots $.
Terms are defined grammatically as follows.
\begin{center}
$M ::= x \; | \; MM' \; | \; \lambda x.M \; | \; M \otimes M' \; | \; \llet{x \otimes x'}{M}{M'} \; | \; \xi M \; | \; \lambda^{\xi} x.M$
\end{center}
Here $x$ of $\lambda x.M$ is the rightmost free variable of $M$, appears exactly once in $M$ and is not in any scope of $\xi$.
$x$ of $\lambda^{\xi} x.M$ need to appear exactly once in $M$.
Also we assume that for $\llet{x \otimes x'}{M}{M'}$, $x'$ and $x$ are the rightmost and the next rightmost free variables of $N$, appear exactly once in $N$ and are not in any scope of $\xi$.

The rest is the same as Example \ref{examrplca}.
\end{exa}

Finally we give an example of exch-PLCA based on $\hkt$ of Example \ref{examtree}.
This example is similar to the one introduced in \cite{tomita2}.

\begin{exa} \label{examtree6}
Let $T$ and $|\haih |$ be the same set and function defined in Example \ref{examtree}.
First we give a $\bci$-algebra $\hkt_e$ from $T$.
Take $|\hkt_e|$ as the powerset of $\{ t \in T \mid |t| =e \}$,
and a binary operation $\circledcirc$ on $|\hkt_e|$ as $M\circledcirc N := \{ t_2 \mid \exists t_1 \in N ,(t_2 \rimp t_1) \in M \}$.
Then $\hkt_e = (|\hkt_e|,\circledcirc)$ is a $\bci$-algebra, where
\begin{itemize}
\item $\comb = \{ (t_3 \rimp t_1) \rimp (t_2 \rimp t_1) \rimp (t_3 \rimp t_2) \mid t_1,t_2,t_3 \in T \}$;
\item $\comc = \{ (t_3 \rimp t_2 \rimp t_1) \rimp (t_3 \rimp t_1 \rimp t_2) \mid |t_1| = |t_2| = |t_3| =e \}$; 
\item $\comi = \{ t_1 \rimp t_1 \mid t_1 \in T \}$.
\end{itemize}

Take $\gamma: |\hkt| \migi |\hkt_e|$ as the function sending $M$ to $\{ t \rimp t \mid t \in M\}$ and $\delta : |\hkt_e| \migi |\hkt|$ as the inclusion function.
Then these function forms an (functional) adjoint pair $(\delta \dashv \gamma): \hkt \migi \hkt_e$.
Here corresponding realizers are
\begin{itemize}
\item $\{ ((t_2 \rimp t_2) \rimp (t_1 \rimp t_1)) \rimp ((t_2 \rimp t_1) \rimp (t_2 \rimp t_1)) \mid t_1,t_2 \in T \}$ realizing $\gamma$;
\item $\{ t_1 \rimp t_1 \mid t_1 \in T \}$ realizing $\delta$;
\item $\{ (t \rimp t) \rimp t \mid |t| = e \}$ realizing $id \preceq \gamma \circ \delta$;
\item $\{ t \rimp (t \rimp t) \mid |t| \geq e \}$ realizing $\delta \circ \gamma \preceq id$.
\end{itemize}
\end{exa}

\begin{rem}
The above construction also can be applied to obtain exch-PLCAs on $\hkt'$ of Example \ref{examtree2} and on $\hkt''$ of Example \ref{examtree4}.
While we gave exch-PLCAs by $T$, the same construction cannot be applied to obtain exp-PLCAs.
If we try to get some PCA of subsets of $T$, employing $M\circledcirc N := \{ t_2 \mid \exists t_1 \in N ,(t_2 \rimp t_1) \in M \}$ as the binary operation, $\comk \circledcirc M \circledcirc N = M$ hardly hold since the left hand side often lost information of $M$ when $N$ is nearly empty.
\end{rem}

\begin{rem}
As we saw in Proposition \ref{propstyle}, exp-rPLCAs can be defined in the style using not the combinators $\bfk$ and $\bfw$, but the applicative morphisms $[\bfban,\bfban]$ and $k_i$.
It is still unclear whether we can define exch-rPLCAs by the latter style, not using the combinator $\bfc$.
If we can characterize exch-rPLCAs by the latter style, we might construct exch-PLCAs using reflexive objects by the same way as exp-PLCAs and WPLCs (Definition \ref{defwplc}).
\end{rem}

%%%%%%%%%%%%%%%%%%%%%%%%%%%%%%%%%%%%%%%%%%%%%%%%%%%%%%%%%%%%%%%%%
%%%%%%%%%%%%%%%%%%%%%%%%%%%%%%%%%%%%%%%%%%%%%%%%%%%%%%%%%%%%%%%%%
\section{Related work} \label{secrel}

This paper is an extended  version of the earlier papers by the author \cite{tomita1,tomita2}.
As a result of \cite{tomita1} not introduced in this paper, we have ``$\comb \comi \kuro{\comi} (\haih)^{\circ}$-algebras'' as a class of applicative structures.
$\comb \comi \kuro{\comi} (\haih)^{\circ}$-algebras are more general than $\bikuro$-algebras, and give rise to  skew closed categories of assemblies (or modest sets).
{\it Skew closed categories}, introduced in \cite{street}, are categories with similar closed structures to closed categories, though some conditions needed in closed categories are not assumed.
(For instance, the natural transformation $i_X : (X \rimp I) \migi X$ in a skew closed category is not necessarily invertible.)
Although skew closed categories and closed multicategories are generalizations of closed categories in different directions, from Proposition \ref{propnecmul}, we can say that we cannot (canonically) obtain $\asmca$ (or $\moda$) that is a closed multicategory but not a skew closed category.
Details of these results are given in Appendix \ref{appskew}.

{\it Skew monoidal categories} introduced in \cite{skewmono}
are categories with the same components as monoidal categories but natural transformations (left and right unitors and associators) do not need to be invertible.
The relationship between skew monoidal categories and skew closed categories is similar to that between monoidal categories and closed categories.
Recalling the proof of Proposition \ref{propmonoclo}, we find that we use $\batui$ only to realize $\rho_X^{-1}:X \otimes I \migi X$.
The invertibility of $\rho_X$ is not assumed in skew monoidal categories.
Thus, when we have $\hka$ as a ``$\comb \comi \coml \comp \kuro{(\haih)}$-algebra,'' we can show that $\asmca$ is a skew monoidal category.

In \cite{hasegawa3}, the ``extensionality'' of combinatory algebras is investigated.
The extensionality defined in that paper is a more generalized condition than the standard one, seen in \eg, \cite{barendregt}.
By the extensionality in \cite{hasegawa3}, we can deal with polynomials and combinatory completeness for combinatory algebras that cannot be stated in the same way as Definition \ref{defpoly} and Proposition \ref{propcompca}, such as the braided case.
In our study, we do not need the discussions of the extensionality to state the combinatory completeness appearing in this paper, however, assuming the extensionality on an applicative structure $\hka$ may cause some structures on $\asmca$ and $\moda$.
For instance, for an ``extensional'' $\bikuro$-algebra $\hka$, since the $\comb$-combinator always satisfies the axiom of $\batui$, $\asmca$ and $\moda$ become closed categories.
There are many other possible way to define classes of applicative structures than using the existence of certain combinators, and the extensionality is such one way.

The definition of bi-$\bdi$-algebras may look like ``dual combinators'' introduced in \cite{dunn}.
Similar to bi-$\bdi$-algebras, in bianry operations of dual combinators, elements can act to elements from both left and right sides.
However, a dual combinatory logic has only one sort of application, whereas a bi-$\bdi$-algebra has two sorts of applications.
Also the reductions of dual combinatory logic do not satisfy the confluence, while the confluence of the bi-planar lambda calculus holds.

In this paper, we referred several logics and their categorical models without recalling detailed definitions.
See \cite{linear} about the linear logic.
And for the MILL and the categorical models that we deal with in this paper, see \cite{seely}.
Also, for the Lambek calculus, the word ``the Lambek calculus'' has various means as logics, and we use this word to mean a variant of non-symmetric MILL with left and right implications in this paper.
Our treatment of the Lambek calculus and its categorical semantics are from \cite{paiva}.
The basics about the Lambek calculus is in \cite{lambek2,lambek,lambek3}.

In \cite{zeilberger}, the relationships between the planar lambda calculus and planar graphs are investigated.
In that paper, the bijection between rooted trivalent planar graphs and closed planar lambda terms is given, and it is shown that such graphs can be generated by combining a few kinds of ``imploid moves.''
The theory corresponds to the combinatory completeness of $\bikuro$-algebras and the planar lambda calculus.
Similarly, we can give the bijection between rooted trivalent planar graphs and closed bi-planar terms, but here the rooted trivalent planar graphs need to have two colored (``left'' and ``right'') vertexes.

%%%%%%%%%%%%%%%%%%%%%%%%%%%%%%%%%%%%%%%%%%%%%%%%%%%%%%%%%%%%%%%%%
%%%%%%%%%%%%%%%%%%%%%%%%%%%%%%%%%%%%%%%%%%%%%%%%%%%%%%%%%%%%%%%%%
\section{Conclusion} \label{seccon}

In section \ref{secback} and \ref{secnonsym}, we introduced several classes of applicative structures and showed that they induce closed structures on categories of assemblies and categories of modest sets, as in Table \ref{table2}. (The results for $\biilp$-algebras are newly presented in this paper.)
In section \ref{secsepa}, we showed that these classes are different ones by giving several examples.
In section \ref{secnec}, we presented propositions that categorical structures of $\asmca$ induce structures of $\hka$, under some conditions.
(The propositions for $\biilp$-algebras and bi-$\bdi$-algebras are newly shown in this paper.)
By combining the results of the above, for instance we can say that we have $\asmca$ with a truly non-symmetric bi-closed structures, by using $\hka$ that is a bi-$\bdi$-algebra but not a $\bci$-algebra.
In section \ref{secplca}, we introduced exp-rPLCAs and exch-rPLCAs that give rise to categorical models for the linear exponential modality and the exchange modality on the non-symmetric MILL.
As an adjoint pair between a $\bci$-algebra and a PCA induces an rLCA, an adjoint pair between $\biilp$-algebras and a PCA/$\bci$-algebra induces an exp-rPLCA/exch-rPLCA.

\begin{table}[h]
\caption{Summary of the classes of applicative structures}
\centering
\begin{tabular}{|c|c|c|c|}
	\hline
	\ Applicative structure $\hka$ \ & Definition & \ Structure of $\asmca$ and $\moda$ \ & Proposition \\ \hline \hline
	PCA/$\comsa \comk$-algebra & \ref{defpca} & Cartesian closed category & \ref{propccc} \\ \hline
	$\bci$-algebra & \ref{defbci} & symmetric monoidal closed category & \ref{propsmcc} \\ \hline
	bi-$\bdi$-algebra & \ref{defbdi} & monoidal bi-closed category & \ref{propbiclo} \\ \hline
	$\biilp$-algebra	& \ref{defbiilp} & monoidal closed category & \ref{propmonoclo} \\ \hline
	$\biikuro$-algebra & \ref{defbii} & closed category & \ref{propclo} \\ \hline
	$\bikuro$-algebra & \ref{defbikuro} & closed multicategory & \ref{propmul} \\ \hline \hline
	
	\multicolumn{4}{|c|}{Inclusions} \\ \hline
	\multicolumn{4}{|c|}{$\comsa \comk$-algebras $\subsetneq$ $\bci$-algebras $\subsetneq$ bi-$\bdi$-algebras} \\ 
	\multicolumn{4}{|c|}{$\subsetneq$ $\biilp$-algebras $\subsetneq$ $\biikuro$-algebras $\subsetneq$ $\bikuro$-algebras} \\ \hline
\end{tabular}
\label{table2}
\end{table}

Finally we give three issues for future work.
First, there are several unsolved problems we mentioned in this paper.
Those that we consider important are:
\begin{itemize}
\item to show that the computational lambda calculus is not a $\biilp$-algebra (refer Section \ref{seccom});
\item to clarify conditions needed to show that $\hka$ is a PCA when $\asmca$ (or $\moda$) is a CCC (refer Remark \ref{remnecpar});
\item to clarify conditions needed to show that $\hka$ is a $\biilp$-algebra/bi-$\bdi$-algebra when $\moda$ is a monoidal closed category/monoidal bi-closed category (refer the end of Section \ref{secnec}).
\end{itemize}

Second, most examples given in this paper are the standard ones like the term models.
We would like to find more interesting examples of applicative structures and adjoint pairs, that should be useful for investigating non-commutative logics and their models in a systematic way.

Third, for various categorical structures not given in this paper, we want to clarify what we need to construct them via categorical realizability.
For instance, we have said (in section \ref{secnec}) that we cannot give $\asmca$ (nor $\moda$) that is a symmetric closed category but not an SMCC, in canonical ways.
Also we cannot give $\asmca$ (nor $\moda$) that is a closed multicategory but not a skew closed category.
As an example not yet mentioned, we cannot make $\asmca$ a braided monoidal category but not an SMCC.
Although there is a class of applicative structure, $\comb \comc^{\pm} \comi$-algebras, nicely corresponding the structure of braided monoidal categories and the braided lambda calculus (investigated in \cite{braid}), the construction of $\asmca$ cannot reflect the difference between two sorts of braids (realized by $\comc^{+}$ and $\comc^{-}$) and turns braids into the symmetry.
To give the categorical structures listed above, we need to change the construction of $\asmca$ (and $\moda$), rather than trying to give conditions on applicative structures.
For instance, to make $\asmca$ a braided monoidal category (not an SMCC), we may need to change that the construction of $\asmca$ is based on $\cats$, that is not only braided but also symmetric.

%%%%%%%%%%%%%%%%%%%%%%%%%%%%%%%%%%%%%%%%%%%%%%%%%%%%%%%%%%%%%%%%%
%%%%%%%%%%%%%%%%%%%%%%%%%%%%%%%%%%%%%%%%%%%%%%%%%%%%%%%%%%%%%%%%%
\section*{Acknowledgment}
  \noindent I would like to thank Masahito Hasegawa for a lot of helpful advice, discussions and comments.
This work was supported by JST SPRING, Grant Number JPMJSP2110.

\bibliographystyle{alphaurl}
\bibliography{lmcsbib}

%%%%%%%%%%%%%%%%%%%%%%%%%%%%%%%%%%%%%%%%%%%%%%%%%%%%%%%%%%%%%%%%%
%%%%%%%%%%%%%%%%%%%%%%%%%%%%%%%%%%%%%%%%%%%%%%%%%%%%%%%%%%%%%%%%%
\appendix
\section{$\bisiro$-algebras, $\bikuro$-algebras and skew closed categories} \label{appskew}

Though classes of applicative structures appearing in this paper are subclasses of $\bikuro$-algebras, it does not conclude that realizability constructions for closed structures all require $\bikuro$-algebras.
Indeed, in \cite{tomita1}, we introduced $\bisiro$-algebras, which is a more general class than $\bikuro$-algebras and gives rise to skew closed categories.

First we recall the definition of skew closed categories from \cite{street}.

\begin{defi}
A {\it (left) skew closed category} $\catc$ consists of the following data:
\begin{enumerate}
\item a locally small category $\catc$;
\item a functor $(- \rimp -):\catc^{op} \times \catc \migi \catc$, called the internal hom functor;
\item an object $I$, called the unit object;
\item an natural transformation $i_X : (X \rimp I) \migi X$;
\item an extranatural transformation $j_X : I \migi (X \rimp X)$;
\item a transformation $L_{Y,Z}^X : (Z \rimp Y) \migi ((Z \rimp X) \rimp (Y \rimp X))$ natural in $Y$ and $Z$ and extranatural in $X$,
\end{enumerate}
such that the following axioms hold:
\begin{enumerate}[(i)]
\item $\forall X,Y \in \catc$, $L_{Y,Y}^X \circ j_Y = j_{(Y \rimps X)}$;
\item $\forall X,Y \in \catc$, $i_{(Y \rimps X)} \circ (id_{(Y \rimps X)} \rimp j_X) \circ L_{X,Y}^X = id_{(Y \rimps X)}$;
\item $\forall X,Y,Z,W \in \catc$, the following diagram commutes:
\begin{center}
\xymatrix@C=-25pt@R=30pt{
	& (W \rimp Z) \ar[dl]_{L_{Z,W}^X} \ar[dr]^{L_{Z,W}^Y} & \\
	(W \rimp X) \rimp (Z \rimp X) \ar[d]^-{L_{(Z \rimpss X),(W \rimpss X)}^{(Y \rimpss X)}} & & 
	((W \rimp Y) \rimp (Z \rimp Y)) \ar[dd]^-{L_{Y,W}^X \rimps id} \\
	((W \rimp X) \rimp (Y \rimp X)) \rimp ((Z \rimp X) \rimp (Y \rimp X)) \ar[drr]_-{id \rimps L_{Y,Z}^X} & & \\
	& & ((W \rimp X) \rimp (Y \rimp X)) \rimp (Z \rimp Y)
	}
\end{center}
\item $\forall X,Y \in \catc$, $(i_Y \rimp id_{(X \rimps I)}) \circ L_{X,Y}^{I} = id_{Y} \rimp i_X$;
\item $i_I \circ j_I = id_I$.
\end{enumerate}
A skew closed category is called {\it left normal} when the function $\gamma : \catc(X,Y)\migi \catc(I,Y \rimp X)$ sending $f:X \migi Y$ to
$(f \rimp id_X) \circ j_X$ is invertible for any $X,Y \in \catc$.
\end{defi}

There is a categorical structure called {\it skew monoidal categories} introduced in \cite{skewmono}, which have the same components as monoidal categories but the invertibility of unitors and associators are not assumed.
Skew closed categories are the categorical structures determined from skew monoidal categories, like closed categories are determined from monoidal categories.
Obviously, closed categories are also left normal skew closed categories.

We investigated categorical realizability for skew closed categories in \cite{tomita1} and next we recall some of the results.

\begin{defi}
A total applicative structure $\hka$ is a $\bisiro${\it -algebra} iff it contains $\comb$, $\comi$, $\kuro{\comi}$ and $\siro{a}$ for each $a \in \uhka$ is an element of $\uhka$ such that $\forall x,y \in \uhka, (\siro{a}) xy = x (ay)$.
\end{defi}

Since $(\comb \kuro{a} \comb) xy=x(ay)$, any $\bikuro$-algebra is also a $\bisiro$-algebra.
By the similar way to the proof of Proposition \ref{propsepa5}, we can show the class of $\bisiro$-algebras is different from the class of $\bikuro$-algebras by using a freely constructed $\bisiro$-algebra (with constants).

\begin{prop}
When $\hka$ is a $\bisiro$-algebra, $\asmca$ and $\moda$ are left normal skew closed categories.
\end{prop}

The proof is almost the same as Proposition \ref{propclo}.
Here for maps $f$ and $g$, we give a realizer of $(g \rimp f)$ as $\comb \siro{(r_f)} (\comb r_g)$.

It is still not clear whether $\hka$ need to be a $\bisiro$-algebra to make $\asmca$ (or $\moda$) a skew closed category, like propositions in Section \ref{secnec}.
(In the similar setting to Proposition \ref{propnecclo} and Proposition \ref{propnecclo2}, though we can show the existence of $\comb$, $\comi$ and $\siro{(\haih)}$, we cannot show there is $\kuro{\comi}$.)

Since $\bikuro$-algebras are $\bisiro$-algebras, the next holds.

\begin{cor}
When $\hka$ is a $\bikuro$-algebra, $\asmca$ and $\moda$ are skew closed categories.
\end{cor}

From Proposition \ref{propnecmul}, we can say that we cannot (canonically) obtain $\asmca$ (nor $\moda$) that is a closed multicategory but not a skew closed category.
Although closed multicategories are a generalized closed categorical structure in a different direction from skew closed categories, skew closed categories are more general than closed multicategories as the categorical structures appearing in categories of assemblies.

Moreover, when constructing applicative structures from reflexive objects, skew closed categories can give even $\bikuro$-algebras, as well as closed multicategories give.

\begin{exa} \label{examapp1}
Suppose a skew closed category $\catc$ and an object $V$ with a retraction
$r : (V \rimp V) \triangleleft V :s$.
Then $\catc(I,V)$ forms a $\bikuro$-algebra.
\begin{itemize}
\item For $M,N:I \migi V$, the application is defined as 
\[ I \xmigi{M} V \xmigi{s} V \rimp V \xmigi{id_V \rimps N} V \rimp I \xmigi{i_V} V. \]
\item The $\comb$-combinator is
\[ I \xmigi{j_{V \rimpss V}} (V \rimp V) \rimp (V \rimp V) \xmigi{L_{V,V}^V \rimps s} ((V \rimp V) \rimp (V \rimp V)) \rimp V \]
\[ \xmigi{(r \rimps s) \rimps id_V} V \rimp V \rimp V \xmigi{r \rimps id_V} V \rimp V \xmigi{r} V. \]
\item The $\comi$-combinator is $r \circ j_V$.
\item Given arbitrary $M:I \migi V$, $\kuro{M}$ is
\[ I \xmigi{j_V} V \rimp V \xmigi{s \rimps id_V} (V \rimp V) \rimp V \xmigi{(id_V \rimps M) \rimps id_V} (V \rimp I) \rimp V \xmigi{i_V \rimps id_V} V \rimp V \xmigi{r} V. \]
\end{itemize}
\end{exa}

\begin{exa} \label{examapp2}
Suppose a closed multicategory $\catc$ and an object $V$ with a retraction \\
$r : \ucatc(V;V) \triangleleft V :s$.
Then $\catc(;V)$ forms a $\bikuro$-algebra.
\begin{itemize}
\item For $M,N \in \ucatc(;V)$, the application is defined as 
\[ \xmigi{M,N} V,V \xmigi{s,id_V} \ucatc(V;V),V \xmigi{ev} V. \]
\item Take a map $f:V,V,V \migi V$ as
\[ V,V,V \xmigi{id_V,s,id_V} V,\ucatc(V;V),V \xmigi{s,ev} \ucatc(V;V),V \xmigi{ev} V. \]
The $\comb$-combinator is given as $r \circ \Lambda_{;V;V} (r \circ \Lambda_{V;V;V}(r \circ \Lambda_{V,V;V;V}(f)))$. Here $\Lambda$ is the function in Definition \ref{defclomul}.
\item The $\comi$-combinator is $r \circ \Lambda_{;V;V}(id_V)$.
\item Given arbitrary $M \in \ucatc(;V)$, $\kuro{M}$ is $r \circ \Lambda_{;V;V}(ev \circ (s,M))$.
\end{itemize}
\end{exa}

When we assume the retraction $r : (V \rimp V) \migi V$ of Example \ref{examapp1} is an isomorphism, the $\comb$-combinator further satisfies the axiom of $\batui$ and $\catc(I,V)$ forms a $\biikuro$-algebra.
Similarly, when we assume the retraction $r : \ucatc(V;V) \migi V$ of Example \ref{examapp2} is an isomorphism, the $\comb$-combinator satisfies the axiom of $\batui$ and $\catc(;V)$ forms a $\biikuro$-algebra.

\end{document}